\documentclass[a4paper,11pt]{article}

\usepackage{authblk}

\usepackage[
    left=25mm,
    right=25mm,
    top=25mm,
    bottom=25mm
]{geometry}
\usepackage[x11names]{xcolor}

\usepackage{hyperref}
\hypersetup{
    colorlinks, citecolor=DarkRed, linkcolor=DarkRed
}

\usepackage{graphicx}
\usepackage{tikz}
\usetikzlibrary{quantikz2}
\usetikzlibrary{arrows.meta}
\usetikzlibrary{decorations.pathmorphing}
\usetikzlibrary{shapes.geometric}
\usetikzlibrary{shapes.symbols}
\usepackage{lscape}
\usepackage{float}
\usepackage{subcaption}
\usepackage{wrapfig}

\usepackage{amssymb}
\usepackage{amsfonts}
\usepackage{bm}
\usepackage{youngtab}
\usepackage{mathdots}
\usepackage{mathtools}
\usepackage{relsize}
\usepackage{tikz-cd}

\newcommand{\cont}{\mathrm{cont}}
\newcommand{\AC}{\mathrm{AC}}
\newcommand{\Cells}{\mathrm{C}}
\newcommand{\RC}{\mathrm{RC}}

\newcommand{\U}{\mathrm{U}}

 \newcommand{\Span}{\mathrm{span}}
\newcommand{\sign}{\mathrm{sign}}
\newcommand{\Res}{\mathrm{Res}}
\newcommand{\Ind}{\mathrm{Ind}}

\newcommand{\F}{\mathfrak{F}}
\newcommand{\C}{\mathbb{C}}

\newcommand{\T}{\mathcal{T}}
\newcommand{\Rep}{\mathrm{R}}
\newcommand{\Path}{\mathcal{P}}
\newcommand{\pt}{\mathbin{\vdash}}

\newcommand{\N}{\mathcal{N}}
\ExplSyntaxOn
\NewDocumentCommand{\IndMap}{O{} O{}}{\ensuremath{\mathrm{U}\sb{\mathrm{Ind}\tl_if_blank:nF {#1} {, #1}}\tl_if_blank:nF {#2} {^{(#2)}}}}

\NewDocumentCommand{\IndMapDag}{O{} O{}}{\ensuremath{\mathrm{U}\sb{\mathrm{Ind}\tl_if_blank:nF {#1} {, #1}}\tl_if_blank:nTF {#2}
      {^{\dagger}}
      {^{(#2)\dagger}}}}

\NewDocumentCommand{\IndMapDash}{O{} O{}}{\ensuremath{\mathrm{U}\sb{\mathrm{Ind}\tl_if_blank:nF {#1} {, #1}}\tl_if_blank:nTF {#2}
      {'}
      {'^{(#2)}}}}
\ExplSyntaxOff

\usepackage{amsthm} \usepackage{thm-restate}
\usepackage{thmtools} \declaretheoremstyle[bodyfont=\normalfont]{plainstyle}
\declaretheorem[name=Definition, within=section]{definition}
\declaretheorem[name=Theorem, sibling=definition]{theorem}
\declaretheorem[name=Lemma, sibling=definition]{lemma}
\declaretheorem[name=Corollary, sibling=definition]{corollary}
 
\declaretheorem[name=Proposition, sibling=definition]{proposition}

\usepackage{array}
\usepackage{booktabs}

\usepackage{nomencl}
\makenomenclature

\usepackage{algorithm}
\usepackage{algpseudocode}

\usepackage[nameinlink,capitalize,noabbrev]{cleveref}

\usepackage[backend=biber,style=apa]{biblatex}
\tikzset{snake it/.style={decorate, decoration=snake}}

\newcommand{\defeq}{\vcentcolon=}

\newcolumntype{M}[1]{>{\centering\arraybackslash}m{#1}}

\makeatletter
\pgfdeclareshape{reptensor}{
\savedmacro{\reptensor@width}{\def\reptensor@width{1.1cm}}
    \savedmacro{\reptensor@rectheight}{\def\reptensor@rectheight{0.6cm}}
    \savedmacro{\reptensor@arcHeight}{\def\reptensor@arcHeight{0.5cm}}

\anchor{center}{
        \pgfpoint{0}{-0.2*\reptensor@rectheight}
    }
    \anchor{west}{
        \pgfpoint{-0.5*\reptensor@width}{-0.2*\reptensor@rectheight}
    }
    \anchor{east}{
        \pgfpoint{0.5*\reptensor@width}{-0.2*\reptensor@rectheight}
    }
    \anchor{south}{
        \pgfpoint{0}{-0.8*\reptensor@rectheight}
    }
    \anchor{text}{
        \pgfmathsetlength{\pgf@x}{-0.5*\wd\pgfnodeparttextbox}
        \pgfmathsetlength{\pgf@y}{-0.5*\ht\pgfnodeparttextbox}    
    }
    
\backgroundpath{
        \pgfpathmoveto{\pgfpoint{-0.5*\reptensor@width}{0.2*\reptensor@rectheight}}
        \pgfpathlineto{\pgfpoint{-0.5*\reptensor@width}{-0.8*\reptensor@rectheight}}
        \pgfpathlineto{\pgfpoint{0.5*\reptensor@width}{-0.8*\reptensor@rectheight}}
        \pgfpathlineto{\pgfpoint{0.5*\reptensor@width}{0.2*\reptensor@rectheight}}
        \pgfpatharc{0}{180}{0.5*\reptensor@width and 0.5*\reptensor@arcHeight + 0.2*\reptensor@rectheight}
        \pgfpathclose
    }
}
\makeatother

\usepackage{xstring} \usepackage{ifthen}  

\newcommand{\w}{0.4}
\newcommand{\p}{0.1}

\newcommand{\drawsquare}[3][]{\StrBefore{#2}{,}[\x]
  \StrBehind{#2}{,}[\y]
  \def\option{#1}
  \def\labeltext{#3}
  \pgfmathsetmacro{\xl}{\x - \w/2}
  \pgfmathsetmacro{\xr}{\x + \w/2}
  \pgfmathsetmacro{\yt}{\y + \w/2}
  \pgfmathsetmacro{\yb}{\y - \w/2}
  \pgfmathsetmacro{\xlp}{\x - \w/2 + \p}
  \pgfmathsetmacro{\xrp}{\x + \w/2 - \p}
  \pgfmathsetmacro{\ytp}{\y + \w/2 - \p}
  \pgfmathsetmacro{\ybp}{\y - \w/2 + \p}

  \ifthenelse{\equal{\option}{dashed}}{\draw[dashed] (\xl,\yt) rectangle (\xr,\yb);
  }{\draw (\xl,\yt) rectangle (\xr,\yb);
  }

  \ifthenelse{\equal{\option}{cross}}{\draw (\xl,\yt) -- (\xr,\yb);
    \draw (\xr,\yt) -- (\xl,\yb);
  }{}

  \ifthenelse{\not\equal{\labeltext}{}}{\ifthenelse{\equal{\option}{cross}}{\draw[fill=white,draw=white] (\xlp,\ytp) rectangle (\xrp,\ybp);
    }{}
    \node at (\x,\y) {\(\labeltext\)};
  }{}
}

\title{Quantum Fourier transform for the symmetric group}

\author[1,2]{Carli Bruinsma}
\author[1,2,3]{Dmitry Grinko}
\author[1,2,3]{Maris Ozols}

\affil[1]{QuSoft, Amsterdam, The Netherlands}
\affil[2]{Institute for Logic, Language and Computation, University of Amsterdam, The Netherlands}
\affil[3]{Korteweg-de Vries Institute for Mathematics, University of Amsterdam, The Netherlands}
\affil[ ]{Email addresses: \href{mailto: carlibruinsma@gmail.com}{carlibruinsma@gmail.com},
\href{mailto: d.grinko@uva.nl}{d.grinko@uva.nl}, \href{mailto: marozols@gmail.com}{marozols@gmail.com}}

\date{}

\begin{document}
\definecolor{DarkRed}{rgb}{0.6471, 0.1098, 0.1882}
\maketitle

\begin{abstract}
    Quantum Fourier transforms (QFT) for general groups were recognized to be fundamental already early in the field. A canonical example of non-abelian QFT for the symmetric group was outlined by \textcite{Beals}. Later, a more detailed analysis of this algorithm was carried out by \textcite{SnQFT_Speedup}.
    In this paper, we revisit that construction.
    After a careful analysis, we revise their gate complexity to $\widetilde{\mathcal{O}}(n^{3.5})$ and circuit depth to $\widetilde{\mathcal{O}}(n^3)$.
Moreover, we observe that their construction is not optimal in the choice of transversal elements, so we propose simpler realization of the symmetric group QFT.
\end{abstract}

\tableofcontents

\section{Introduction}\label{chap: introduction}

Before diving into the main subject of this work, let us appreciate the quantum Fourier transform (QFT) as a fundamental tool for quantum algorithm design.
Famous examples include the algorithms of \textcite{Deutsch1985,DeutschJozsa1992,Shor1994,BernsteinVazirani1997}, where the abelian QFT plays a key role.
Another important, more general result about abelian QFTs was found by \textcite{Kitaev1995}, who implements the QFT of \emph{any} abelian group and uses it to solve the abelian hidden subgroup problem (HSP).
Even the humble Hadamard gate can be interpreted as the QFT for $\mathbb{Z}_2$.

For non-abelian QFTs, we know that they often fail to solve the HSP efficiently, as discussed by \textcite{RevModPhys.82.1}.
Even for the symmetric group, which is the focus of this work, \textcite{FourierNotForHSPSn2005} show that strong Fourier sampling does \emph{not} efficiently solve the HSP, despite previous hopeful expectations.
Also, no quantum algorithms for the related graph isomorphism problem have beaten the classical quasipolynomial-time algorithm by \textcite{10.1145/2897518.2897542} so far.
However, the QFT for the symmetric group $S_n$ still has many other important applications, which motivate this work.

One of them is computing representation-theoretic multiplicities.
Recently, \textcite{PRXQuantum.5.010329} used it to approximate Kronecker coefficients. However, they use the existence of an efficient QFT for $S_n$ more as an argument to discuss the complexity class of finding Kronecker coefficients.
\textcite{paperMultiplicities} use the QFT over $S_n$ to approximate Kostka numbers, Littlewood--Richardson coefficients, Kronecker coefficients, and certain plethysm coefficients efficiently.
While they argue that we cannot expect superpolynomial speedups for these problems, there is likely a large polynomial gap between the quantum and classical algorithms for plethysm coefficients, and possibly for Kronecker coefficients, in certain regimes.
As the QFT for $S_n$ plays such a central role in these algorithms, quantifying the polynomial gaps between the quantum and classical algorithms requires an implementation of the $S_n$-QFT that is as efficient as possible.

Another important application of the quantum Fourier transform over the symmetric group is in the quantum Schur transform algorithm proposed by \textcite{SchurTransformKrovi}.
\textcite{newschur} recently showed that when the local dimension $d$ is much larger than the number of copies $n$,
Krovi's approach outperforms the original Schur transform \cite{OGSchur}.
Moreover, they showed that the complexity of the $S_n$-QFT is currently the bottleneck for the overall Schur transform’s asymptotic gate count.

This leads to our main objective of this work: finding a low-level implementation of the $S_n$-QFT which is as efficient as possible.
A high-level idea on how to implement the QFT over the symmetric group $S_n$ was first proposed by \textcite{Beals} and later generalized to other groups by \textcite{GenQFT}.
This latter work finds the gate count for the $S_n$-QFT as $\widetilde{\mathcal{O}}(n^4)$.
The first explicit circuit for the $S_n$-QFT was derived by \textcite{KS2013} and uses $\widetilde{\mathcal{O}}(n^4)$ elementary operations, too.
At a high level, this circuit has the same inductive structure as the algorithm by \textcite{Beals}, but at a lower level, the work introduces a second induction relation which is completely new.
Three years later, \textcite{SnQFT_Speedup}
improve the circuit's gate count, claiming to need $\widetilde{\mathcal{O}}(n^3)$ elementary operations only.
However, some details that are key to their complexity claim have not been fully worked out in their paper.
In particular, their complexity analysis treats a certain multi-qubit gate as implementable in $\widetilde{\mathcal{O}}(n)$ elementary operations, while their proposed implementation method does not result in a linear-size decomposition into elementary one- and two-qubit gates when worked out in detail.
This seems to be a very non-trivial task.

To achieve our goal, we will study both approaches for implementing the $S_n$-QFT.
For the approach related to the work of \textcite{SnQFT_Speedup}, we derive how the second induction relation follows from Mackey theory, which is an application of one of the techniques for constructing QFTs described by \cite{cookbook}.
As a result, we not only correct the complexity claim to a more usual elementary-gate model, but also simplify the circuit and fix a few mistakes in the proposed implementation.
However, picking a suitable input encoding in the approach of \textcite{Beals} as generalized by \textcite{GenQFT} allows for a significant simplification, which surprisingly gives us the $S_n$-QFT algorithm with the best known asymptotic gate complexities compared to the literature. 

\subsection{Results and their relation to previous work}

\cref{def: Fourier transform} defines the unitary $\mathfrak{F}_n$ implementing the Fourier transform for the symmetric group.
\cref{thm: the big complexity theorem Beals,thm: the big complexity theorem KS} state the main results of this work, which we prove in \cref{sec: complexity}.
In both theorems, the complexities scale as $\mathrm{polylog}(1/\epsilon)$ with diamond norm error $\epsilon$.

\begin{restatable}[Beals algorithm]{theorem}{thmBealsResult}\label{thm: the big complexity theorem Beals}
For $n \geq 1$, the algorithm by \textcite{Beals} leads to a circuit implementing the quantum Fourier transform for the symmetric group $S_n$ with $\widetilde{\mathcal{O}}(n^{3})$ depth, $\widetilde{\mathcal{O}}(n^{3})$ one- and two-qubit gates and $\widetilde{\mathcal{O}}(n^{1.5})$ qubits.
\end{restatable}

For the Beals algorithm, we worked out the circuit in detail. 
By fixing a suitable choice of transversals for a chain of subgroups of the symmetric group, which is used to encode group elements, the resulting gate cancellations reduce the gate count from $\widetilde{\mathcal{O}}(n^4)$---as analyzed by \cite{GenQFT}---to $\widetilde{\mathcal{O}}(n^3)$.

\begin{restatable}[Kawano--Sekigawa algorithm]{theorem}{thmKSResult}
\label{thm: the big complexity theorem KS}
For $n \geq 1$, a simplified and corrected algorithm of \textcite{SnQFT_Speedup} gives a quantum circuit implementing the quantum Fourier transform for the symmetric group $S_n$ with depth of $\widetilde{\mathcal{O}}(n^{3})$, $\widetilde{\mathcal{O}}(n^{7/2})$ one- and two-qubit gates and $\widetilde{\mathcal{O}}(n^{3/2})$ qubits.
\end{restatable}

For the second implementation, note that the complexity seems slightly worse than the gate count of $\widetilde{\mathcal{O}}(n^3)$ claimed by \textcite{SnQFT_Speedup}.
This is the result of a difference in gate count model.
When we count one- and two-qubit gates, rather than the gates considered elementary by \textcite{SnQFT_Speedup}, the circuit uses $\widetilde{\mathcal{O}}(n^{7/2})$ gates. 

While both theorems focus on the complexities of two different $S_n$-QFT algorithms, their proofs involve a low-level decomposition of the algorithm to the level of one- and two-qubit gates and coherent arithmetic operations, which we consider implementable using a polylogarithmic number of gates in the input size.
This paves the way towards running the QFT on an actual quantum computer.

\section{Preliminaries and notation}\label{sec: preliminaries}
This section briefly introduces the background concepts central to this work.

To shorten notation and save horizontal space, we will use the following notation for matrix elements of a matrix $X$:
\begin{equation}
    [X]^{P}_{Q} := \bra{P} X \ket{Q}.
    \label{eq: matrix entry convention}
\end{equation}
This notation replaces the usual braket notation in this paper.

\subsection{The symmetric group and its irreducible representations}
For $n \in \mathbb{N}$, let $[n] := \{1,\dotsc,n\}$.
The \emph{symmetric group} $S_n$ consists of all bijections $\pi : [n] \to [n]$, with composition of functions as group operation.
We will use \emph{cycle notation} to denote group elements.
We denote a cycle $c \in S_n$ of length $k \leq n$ as $c = (c_0, \dots, c_{k-1})$, where all entries $c_i$ are disjoint elements of $[n]$,
and $c$ acts on $c_i \in [n]$ by $c(c_i) = c_{i+1 \bmod k}$, and trivially on all elements of $[n]$ that do not appear as entries in the cycle. 
Disjoint cycles commute, and any symmetric group element can be written as a product of disjoint cycles. 

Denote by $\widehat{S}_n$ a set indexing the irreducible representations (irreps) of $S_n$ over $\C$.
It is well-known that $\widehat{S}_n$ can be identified with \emph{integer partitions}.

\begin{definition}[Integer partition]
    For $n \in \mathbb{N}$, an \emph{integer partition} $\lambda$ is a weakly decreasing sequence of natural numbers summing to $n$.
    We denote this by $\lambda \pt n$.
\end{definition}

Hence, $\widehat{S}_n \defeq \{\lambda : \lambda \pt n\}$.
We can visualize an integer partition $\lambda = (\lambda_1, \dots, \lambda_k) \pt n$ by a \emph{Young diagram} consisting of $n$ \emph{boxes} or \emph{cells}, with $\lambda_1$ boxes in the first row, $\lambda_2$ boxes in the second row etc., and $\lambda_k$ boxes in the last row, where all rows are aligned to the left (see \cref{fig: bratteli diagram} for some examples).

Any representation of a finite group can be made unitary by an appropriate basis change.
For $\lambda \in \widehat{S}_n$, denote by $\Rep_\lambda : S_n \to \U(V_\lambda)$ the unitary irreducible representation of $S_n$ labeled by $\lambda$, where $V_\lambda$ is the corresponding simple $\C[S_n]$-module of dimension $d_\lambda := \dim V_\lambda$.
We will often refer to $V_\lambda$ or just $\lambda$ as ``representation'', with the map $\Rep_\lambda$ being implicit.

\subsection{Chain of subgroups and multiplicity-free restrictions for \texorpdfstring{$S_n$}{Sn}}\label{subsec: group setup}
In this work, we will extensively use the following chain of subgroups for $S_n$:
\begin{equation}\label{eq: Sn chain of subgroups}
    S_1 \subset S_2 \subset \dots \subset S_n.
\end{equation}
For $1 < m < n$, the subgroup $S_m \subset S_n$ stabilizes $\{m+1, \dots, n\} \subset [n]$.

An important property of this chain of subgroups is that all restrictions of irreducible representations along the chain are \emph{multiplicity-free}.
Namely, if we \emph{restrict} an $S_m$ irrep $V_\lambda$ to $S_{m-1}$, which we denote by $\Res^{S_m}_{S_{m-1}} V_\lambda$, it decomposes into a direct sum of $S_{m-1}$-irreps, each appearing with multiplicity at most one. 
Hence, $\lambda$ has a set of \emph{predecessors} $\N^{-}(\lambda) \subseteq \widehat{S}_{m-1}$ such that
\begin{equation}
    \Res^{S_m}_{S_{m-1}} V_\lambda \cong \bigoplus_{\kappa \in \N^-(\lambda)} V_\kappa.
    \label{eq: restriction}
\end{equation}
Similarly, we define the set $\N^+(\lambda) \subseteq \widehat{S}_{m+1}$ of \emph{successors} of $\lambda$ to contain all $S_{m+1}$-irreps appearing in the \emph{induced representation} $\Ind^{S_{m+1}}_{S_m}V_\lambda$. 
By Frobenius reciprocity, $\mu \in \N^+(\lambda)$ if and only if $\lambda \in \N^-(\mu)$.
Inducing an irrep one step along the chain also yields a multiplicity-free decomposition into irreps.

The \emph{branching rule} for $\lambda \in \widehat{S}_{m}$ gives the following explicit description of predecessors:
\begin{equation}
    \mathcal{N}^-(\lambda) = \{\kappa \in \widehat{S}_{m-1} \mid \lambda - \kappa \text{ is a standard basis vector}\}
    \label{eq: N-}
\end{equation}
where standard basis vectors are of the form $(0,\dotsc,0,1,0,\dotsc,0)$.
As Young diagrams, we say that $\kappa \in \mathcal{N}^-(\lambda)$ if the Young diagram $\kappa$ can be obtained by \emph{removing a box} from $\lambda$.
Similarly, $\mu \in \mathcal{N}^+(\lambda)$ if $\mu$ can be obtained by \emph{adding a box} to $\lambda$.

\subsection{Bratteli diagram for \texorpdfstring{$S_n$}{Sn}}\label{subsec: Bratteli diagram}
The \emph{Bratteli diagram} for the subgroup chain $S_1 \subset S_2 \subset \dots \subset S_n$ of $S_n$ is a directed graph with nodes $\bigcup_{m = 1}^{n}\widehat{S}_m$ arranged in $n$ \emph{levels}, with $\widehat{S}_m$ belonging to level $m$.
There is an edge $(\kappa, \lambda)$ from $\kappa \in \widehat{S}_{m-1}$ to $\lambda \in \widehat{S}_{m}$ if $\kappa$ appears in the restriction of $\lambda$ to $S_{m-1}$, or equivalently, if $\lambda$ appears in the induction of $\kappa$ to $S_{m}$.
\cref{fig: bratteli diagram} shows the Bratteli diagram for $S_5$, with integer partitions visualized using Young diagrams.

\begin{figure}[h]
    \centering
    \includegraphics[]{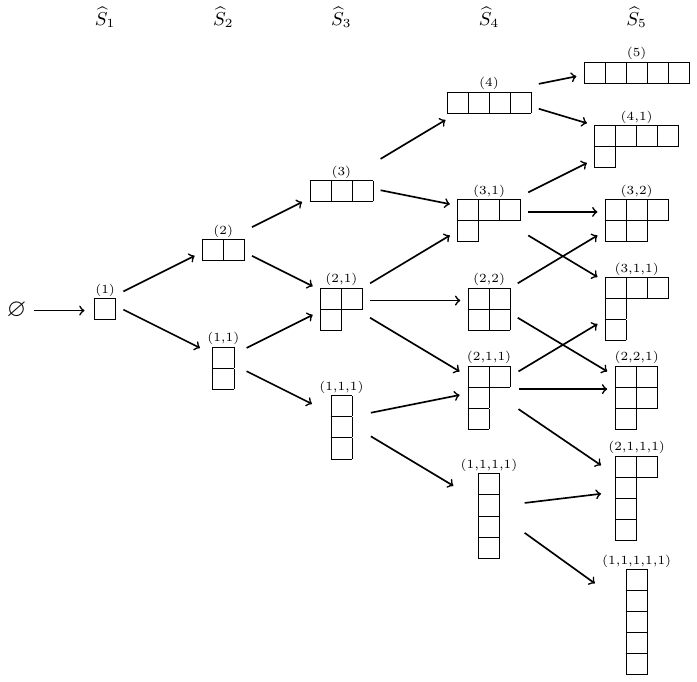}
    \caption{The Bratteli diagram for the symmetric group $S_5$. The nodes are shown as both integer partitions and their corresponding Young diagrams.}
    \label{fig: bratteli diagram}
\end{figure}

As a consequence of the multiplicity-free restriction along the chain of subgroups, the set of directed paths in the Bratteli diagram from $\varnothing \in \widehat{S}_0$ to any node $\lambda \in \widehat{S}_m$ indexes a basis for $V_\lambda$.
More formally, we define the set of all \emph{paths} to $\lambda \in \widehat{S}_m$ as
\begin{equation}
    \Path(\lambda) := \{(P_0 = \varnothing) \to P_1 \to P_2 \to \dots \to (P_m = \lambda) \mid \forall k \in [m] : P_k \in \widehat{S}_k, P_{k-1} \in \mathcal{N}^-(P_k)\}.
\end{equation}
We say that $\Path(\lambda)$ indexes the \emph{Gelfand--Tsetlin} basis of $V_\lambda$ with respect to the subgroup chain~\eqref{eq: Sn chain of subgroups} and write
\begin{equation}
    V_\lambda = \Span \{ \ket{P} : P \in \Path(\lambda)\}.
    \label{eq: P spans V}
\end{equation}
Note that $d_\lambda = \dim V_\lambda = |\Path(\lambda)|$.

To denote the extension of a path $P \in \Path(\lambda)$ to a node $\mu \in \mathcal{N}^+(\lambda)$, we write $P \to \mu$. 
If we shorten the path $P$ by one node, we write $P'$.
In particular, $P = P' \to \lambda$.

\subsection{Young--Yamanouchi basis for \texorpdfstring{$S_n$}{Sn}}\label{subsec: Young's orthogonal form}
Let $\lambda \in \widehat{S}_{n}$. 
Just like Young diagrams visualize integer partitions, the paths in $\Path(\lambda)$ indexing a basis of the simple $\C[S_n]$-module $V_\lambda$ can be identified with the \emph{standard Young tableaux} of shape $\lambda$.
\begin{definition}[Standard Young tableau]\label{def: SYT}
    For an integer partition $\lambda \pt n$, a \emph{standard Young tableau (SYT)} of shape $\lambda$ is a filling of the boxes of the Young diagram $\lambda$ with numbers $1$ to $n$, such that the entries strictly increase from left to right along rows and from top to bottom along columns.
\end{definition}
If we follow a path $P \in \Path(\lambda)$ from $\varnothing$ to $\lambda$ and add a box at each step accordingly, we get a Young diagram of shape $\lambda$ with $n$ boxes.
We can obtain the SYT corresponding to $P$ by inserting entry $m$ into the box added at the $m$-th step.
This gives a bijection between paths $\Path(\lambda)$ and SYTs of shape $\lambda$.

Recall that $S_n$ is generated by adjacent transpositions $\sigma_i \defeq (i, i{+}1)$ where $i \in [n-1]$.
\emph{Young's orthogonal form} \parencite[Section 3.4]{JamesKerber} defines the irreducible representations $\Rep_\lambda$ with respect to the Gelfand--Tsetlin basis indexed by paths in the Bratteli diagram.
For path $P \in \Path(\lambda)$, $\Rep_\lambda(\sigma_i)$ acts on basis vector $\ket{P}$ as
\begin{equation}\label{eq: young's orthogonal form}
    \Rep_\lambda(\sigma_i)\ket{P} = \frac{1}{r_P(i)}\ket{P} + \sqrt{1 - \frac{1}{r_P(i)^2}}\ket{\sigma_i P}.
\end{equation}
Here, $r_P(i)$ denotes the \emph{Manhattan distance} from the box with entry $i$ to the box with entry $i+1$ in the SYT of $P$, where going up or right counts as positive distance and otherwise as negative distance,
and $\sigma_i P \in \Path(\lambda)$ denotes the path corresponding to the SYT of $P$ with entries $i$ and $i{+}1$ swapped. 
By \cref{lemma: absolute manhattan distance of 1}, such a path exists if and only if $|r_P(i)| \neq 1$.

Young's orthogonal form is \emph{subgroup-adapted}, meaning that $S_n$ irreps evaluated on elements from subgroups earlier in the chain~\eqref{eq: Sn chain of subgroups} are block-diagonal.
More precisely, for any $\lambda \in \widehat{S}_n$, $P, Q \in \Path(\lambda)$ and subgroup element $\pi \in S_{n-1}$, the matrix entries of $\Rep_{\lambda}(\pi)$ satisfy
\begin{equation}\label{eq: subgroup adapted basis restriction version}
    [\Rep_{\lambda}(\pi)]^{P}_{Q} = \begin{cases}[\Rep_{P_{n-1}}(\pi)]^{P'}_{Q'} &\text{if } P_{n-1} = Q_{n-1},\\
    0 &\text{otherwise}.\end{cases} 
\end{equation}
Conversely, consider $\kappa \in \widehat{S}_{n{-}1}$ and $\lambda \in \widehat{S}_n$ such that $\kappa \in \N^-(\lambda)$. 
Then for any paths $P, Q \in \Path(\kappa)$ and $\pi \in S_{n-1}$ it holds that 
\begin{equation}\label{eq: subgroup adapted basis induction version}
    [\Rep_{\kappa}(\pi)]^{P}_{Q} = [\Rep_{\lambda}(\pi)]^{P \to \lambda}_{Q \to \lambda}.
\end{equation}
In terms of modules, subgroup-adapted means that \cref{eq: restriction} holds with equality:
\begin{equation}
    \Res^{S_{n}}_{S_{n-1}} V_\lambda = \bigoplus_{\kappa \in \N^-(\lambda)} V_\kappa.
\end{equation}

\subsection{Addable and removable cells, Manhattan distances and contents}\label{subsec: ACs RCs distances and contents}

For each box in the Young diagram of $\lambda \in \widehat{S}_n$ we can assign coordinates $(x,y) \in [n] \times [n]$.
We assign the upper-left box coordinates $(1,1)$. 
Traversing the rows to the right increases the first coordinate by one per box, and traversing columns from top to bottom increases the second coordinate by one per box.

The following definition relates to the notions of predecessors $\N^-(\lambda)$ and successors $\N^-(\lambda)$ of $\lambda$ introduced in \cref{subsec: group setup}.

\begin{definition}[Addable and removable cells]\label{def: AC and RC implicit}
    For $\lambda \in \widehat{S}_n$, denote by $\AC(\lambda) \subset [n{+}1] \times [n {+} 1]$ the \emph{set of addable cells} of $\lambda$, containing all cells $a = (x_a, y_a)$ that can be added to $\lambda$ to obtain a valid integer partition $\lambda + a \in \widehat{S}_{n+1}$.
    Similarly, when $n > 0$, denote by $\RC(\lambda) \subset [n]\times [n]$ the \emph{set of removable cells} of $\lambda$, containing all cells $r = (x_r, y_r)$ of $\lambda$ that can be removed from $\lambda$ to obtain a valid integer partition $\lambda - r \in \widehat{S}_{n-1}$.
\end{definition}

There is a bijection from $\AC({\lambda})$ to $\N^+(\lambda)$ given by $a \mapsto \lambda {+} a$ with inverse $\mu \mapsto \mu {\setminus} \lambda$, where $\mu {\setminus} \lambda$ denotes the unique cell that is in $\mu$ but not in $\lambda$.
Similarly, there is a bijection from $\RC(\lambda)$ to $\N^-(\lambda)$ given by $r \mapsto \lambda {-} r$ with inverse $\kappa \mapsto \lambda {\setminus} \kappa$.

For computing the Manhattan distance defined below \cref{eq: young's orthogonal form}, we will use the notion of content.
The \emph{content} of a cell $c = (x_c, y_c)$ is
\begin{equation}
    \cont(c) \defeq x_c - y_c.
\end{equation}
Note that the content is constant along the diagonals from top left to bottom right.
Now, consider a path $P \in \Path(\lambda)$, or equivalently an SYT, and 
denote by $\cont_P(i)$ the content of the box $i$ of the SYT.
Then, we can express the Manhattan distance for $i \in [n{-}1]$ as
\begin{equation}\label{eq: Manhattan distance with contents}
    r_P(i) = \cont_{P}(i{+}1) - \cont_{P}(i).
\end{equation}

We will later use the following identity \parencite[eq. 3.2.3]{Kerov1993} that expresses the ratio of two irrep dimensions in terms of contents.

\begin{lemma}\label{lemma: twiddle factor}
Let $n > 1$, $\lambda \in \widehat{S}_{n-1}$, $d_\lambda = \dim V_\lambda$, and $a \in \AC(\lambda)$. Then
\begin{equation}
    \frac{d_{\lambda+a}}{n d_\lambda} = \frac{\prod_{x \in \RC(\lambda)}(\mathrm{cont}(x)-\mathrm{cont}(a))}{\prod_{x \in \AC(\lambda)\setminus\{a\}}(\mathrm{cont}(x)-\mathrm{cont}(a))}.
\end{equation}
\end{lemma}

\subsection{Yamanouchi words}\label{subsec: Yamanouchi words}
Our quantum algorithm will need to store paths in the Bratteli diagram.
Since there are different ways of encoding a path, we need to choose the most suitable one.

Let $\Path_n$ denote the full set of paths to level $n$:
\begin{equation*}
    \Path_n \defeq \bigsqcup_{\lambda \in \widehat{S}_n} \Path(\lambda).
\end{equation*}
One way to encode $P = (P_1, \dots, P_n)\in \Path_n$ in a quantum register is to individually store each node label $P_m$, $m \in [n]$ in a subregister as $\ket{P_m}$.
We call this \emph{absolute path encoding}.

Another option is to use a relative path encoding. 
Recall from the symmetric group branching rule in \cref{eq: N-} that the Young diagrams $P_{m-1}$ and $P_m$ differ by one box $P_m {\setminus} P_{m-1} \in \AC(P_{m-1}) \cap \RC(P_m)$.
As any Young diagram has at most one addable or removable cell per row, we can identify addable cells with their row number.

\begin{definition}[Yamanouchi word]
    The \emph{Yamanouchi word} of path $P \in \Path_n$ is the string $p = (p_1, \dots, p_n) \in [n]^n$ where $p_m$ for $m \in [n]$ denotes the row occupied by the box $P_{m} {\setminus} P_{m-1}$. 
    Note that we always have $p_1 = 1$. 
\end{definition}

Let us show by induction on $n$ that the Yamanouchi word $p$ fully encodes the path $P$.
When $n = 1$, we always have $p = (p_1) = (1)$ and $P_1 = (1)$.
For the induction step, assume that $p' = (p_1, \dots, p_{n-1})$ encodes $P' = (P_0, P_1, \dots, P_{n-1})$. 
Then, we know $P_{n-1}$ by the induction hypothesis, and $p_n$ identifies an addable cell in $a \in \AC(P_{n-1})$ by its row number.
We recover the last node as $P_n = P_{n-1} + a$, and the full path is $P = P' \to P_n$.

For any encoding of paths, the minimal number of bits is the base 2 logarithm of the total number of paths.
For the symmetric group, \cref{cor: optimal encoding Sn paths} finds this minimal number of bits to scale as $\Theta(n \log n)$.
Hence, we consider any encoding of paths using $\mathcal{O}(n\log(n))$ bits to have an optimal asymptotic space complexity.
From \cref{lemma: Yamanouchi word optimal scaling}, we can show that the Yamanouchi word encoding has an optimal asymptotic space complexity for encoding paths in the $S_n$-Bratteli diagram.
\cref{lemma: Sn absolute path min qubits} shows that an absolute encoding of paths for $S_n$ never has optimal asymptotic space complexity.

\section{Fourier transform along a multiplicity-free subgroup chain}\label{sec: Fourier transforms}

In this section, we define the Fourier transform, and give a quantum circuit that implementing it.
This is a summary of chapter 3 of \cite{cookbook}, where the circuit is introduced as an implementation of the QFT for any finite group with a multiplicity-free subgroup chain.
This induction relation is the main ingredient of classical fast Fourier transforms,

and captures the structure of the QFT algorithms by \cite{Beals, GenQFT, SnQFT_Speedup}.

\subsection{Regular representations and the Fourier transform}\label{subsec: Sn regular reps}

Let us recap basic facts about Fourier transform $\F_n$ for the symmetric group $S_n$.The group algebra $\C[S_n]$ can be seen as a left and right $\C[S_n]$ module, corresponding to the \emph{left} and \emph{right regular representations} $L, R : S_n \to \C^{S_n \times S_n}$.
The left and right actions of $h \in S_n$ on basis vectors $\ket{g} \in \C[S_n]$, with $g \in S_n$, are defined by
\begin{align}
    L(h)\ket{g} &= \ket{hg} & 
     R(h)\ket{g} &= \ket{gh^{-1}}.
\end{align}
As the actions of the left and right regular representations commute, we can simultaneously decompose both representations into irreps.
We have
\begin{equation}\label{eq: Peter--Weyl duality Sn}
    \C[S_n] \simeq \bigoplus_{\lambda \in \widehat{S}_n} V_\lambda \otimes V_\lambda^*.
\end{equation}
Here, $V_\lambda^*$ denotes the dual representation of $V_\lambda$. 

Note that the intertwiners for the isomorphism of \cref{eq: Peter--Weyl duality Sn} depends on our choice of basis for the irreducible representations.
The following definition ensures that the Fourier transform $\F_n$ of $S_n$ is unitary.
Recall that we use row and column index conventions from \cref{eq: matrix entry convention}.

\begin{definition}\label{def: Fourier transform}
For the symmetric group $S_n$, the Fourier transform $\F_n$

is a matrix whose columns are indexed by group elements $g \in S_n$ and rows are indexed by triples $(\lambda,Q,R)$ where $\lambda \in \widehat{S}_n$ and $Q,R \in \Path(\lambda)$.
For simplicity, we will use $(Q,R)$ instead of $(\lambda,Q,R)$ and treat $\lambda$ as implicit in both $Q$ and $R$.
The corresponding matrix entry is then given by
\begin{equation}
    \label{eq: definition Fourier transform}
    \left[\F_{n}\right]^{Q,R}_{g} := \sqrt{\frac{d_\lambda}{n!}} \left[\Rep_{\lambda}(g)\right]^{Q}_{R}.
\end{equation}
    Here, $\Rep_{\lambda}$ is given by Young's orthogonal form from \cref{subsec: Young's orthogonal form}. 
\end{definition}

Recall that $(Q,R)$ ranges over all pairs of paths in the Bratteli diagram of $S_n$ (see \cref{subsec: Bratteli diagram}) leading to the same final node $\lambda \in \widehat{S}_n$.
Equivalently, we can range $(Q,R)$ over all standard Young tableaux of the same shape $\lambda$ (see \cref{def: SYT}).

\subsection{Transversals and induced representations}
Chapter 3  of \cite{cookbook} introduces a Fourier transform algorithm based on induced representations along a multiplicity free chain of subgroups. 
Applied to the symmetric group, this gives a recursive implementation of $\F_n$ from \cref{def: Fourier transform}

For a finite group $G$ with subgroup $H \subset G$, let $V$ be a $\C[H]$-module with a corresponding matrix representation $\Rep : H \to \U(V)$.
Inducing $V$ from $H$ to $G$ is a canonical way of turning it from a $\C[H]$-module into a $\C[G]$-module.
We denote the induced module by $\Ind^G_H V$, and the corresponding induced representation by $\Ind^G_H \Rep$.
To define induced representations with an explicit basis, we first introduce a few important concepts from group theory.

\begin{definition} [Transversal]
For a group $G$ with subgroup $H \subset G$, a \emph{(left) transversal} $\T$ is a subset of $G$ that contains exactly one representative of every (left) coset $gH = \{gh\mid h \in H\}$ of $H$ in $G$.
\end{definition}

For the rest of this section, fix a left transversal $\T$ indexing the left cosets of $H$ in $G$.

\begin{lemma}\label{lemma: decomposing group elements with transversal}
For all $g \in G$, there exists a unique decomposition $g = th$ with $h \in H$ and $t \in \T$.
\end{lemma}

\begin{proof}
As left cosets $gH = \{gh \mid h \in H\}$ partition $G$, we know that $g$ belongs to a unique left coset.
As $\T$ has exactly one representative per left coset, there is a unique $t \in \T$ so that $gH \cap \T = \{t\}$.
As $t$ and $g$ are in the same left coset, there exists an $h \in H$ such that $g = th$, which we find as $h = t^{-1}g$.
\end{proof}

\begin{definition}[Transversal representation]\label{def: transversal representation}
   The \emph{transversal representation} is defined by the $\C[G]$-module $\C[\T]$.
   Let $m_\T : G \to \T$ be the map that assigns the unique transversal element $t$ to each group element $g$ as in \cref{lemma: decomposing group elements with transversal}. 
   Then, the action of group elements $g \in G$ on basis vectors $\ket{t} \in \C[\T]$ with $t \in \T$ is given by 
   \begin{equation}\label{eq: transversal module action}
       g \cdot \ket{t} = \ket{m_\T(gt)}.
   \end{equation}
\end{definition}

To validate that this indeed defines a representation of $G$, we show that for all $g_1, g_2 \in G$ and $t \in \T$, we have that $g_1 \cdot (g_2 \cdot \ket{t}) = (g_1g_2)\cdot \ket{t}$. 
Using \cref{eq: transversal module action}, we get $g_1 \cdot (g_2 \cdot \ket{t}) = g_1 \cdot \ket{m(g_2t)} = \ket{m(g_1(m(g_2t))}$.
As $m(g_2t) \in g_2tH$, there exists a subgroup element $h \in H$ so that $m(g_2t) = g_2t h$.
Using that $m(gh) = m(g)$ for all $g \in G$ and $h \in H$, we conclude that $\ket{m(g_1(m(g_2t))} = \ket{m(g_1 g_2t h)} = \ket{m(g_1g_2t)} = (g_1 g_2) \cdot \ket{t}$. 

\begin{definition}[Induced representation]\label{def: induced rep} 
    For an $H$-representation $\Rep : H \to \U(V)$, define the \emph{induced representation} $\Ind ^G_H \Rep : G \to \U(\C[\T] \otimes V)$ by its action on basis vectors as
    \begin{equation}\label{eq: induced rep action}
        \Ind^G_H \Rep(g) \bigl(\ket{t} \otimes \ket{P}\bigr) := \ket{m_\T(gt)}\otimes \Rep(m_\T(g)^{-1}g) \ket{P},
    \end{equation}
    where $g \in G$, $t \in \T$ and $\ket{P} \in V$.
    Note that $\Rep(m_\T(g)^{-1}g)$ is well-defined since $m_\T(g)^{-1} g \in H$.
\end{definition}

Hence, the induced module $\Ind^G_H V$ consists of $[G : H]$ copies of $V$, and fixing a transversal $\T$ gives us an explicit labeling for these copies. 
As a matrix, the induced representation $\Ind^G_H \Rep(g)$ contains $[G : H]$ blocks $\Rep(m_{\T}(g)^{-1}g)$ that are permuted according to the transversal representation from \cref{def: transversal representation}.

\begin{definition}[Induction map]\label{def: IndMap}
Assume that the restrictions of all irreps of $G$ to $H$ are multiplicity-free.
For all $\lambda \in \widehat{H}$, define the linear map $\IndMap[\lambda] : \C[\T] \otimes V_\lambda \to \bigoplus_{\mu \in \N^+(\lambda)} V_\mu$ by its entries
\begin{equation}
[\IndMap[\lambda]]^{Q', \mu}_{t, P} = \sqrt{\frac{d{\mu}}{d_\lambda|\T|}}\left[\Rep_{\mu}(t)\right]^{Q}_{P \to \mu}.
\end{equation}
Here, we have indices $t \in \T$, $P \in \Path(\lambda)$ for basis vectors of $\C[\T] \otimes V_\lambda$ and $\mu \in  \N^+(\lambda)$ and $Q \in \Path(\mu)$.
We use path truncation notation $Q'$, so that $Q = Q' \to \mu$.
As in \cref{subsec: group setup}, the set of successors of $\N^+(\lambda)\subseteq \widehat{G}$ includes all irreps of $G$ that restrict to $\lambda$.  
\end{definition}

\begin{lemma}\label{lemma: IndMap basis transform}
For all $\lambda \in \widehat{H}$, we have that $\Ind^G_H V_\lambda \overset{\IndMap[\lambda]}{\simeq} \bigoplus_{\mu \in \N^+(\lambda)} V_\mu$
as $\C[G]$-modules.
\end{lemma}

Hence, $\IndMap[\lambda]$ is the basis transform decomposing induced irreducible representations of $H$ to $G$ into a direct sum of irreducible representations of $G$.
Equivalently, we say that for all $g \in G$, we have
\begin{equation}\label{eq: IndMap intertwines induced rep}
\IndMap[\lambda] \left( \Ind^{G}_{H}\Rep_\lambda (g)\right)\IndMapDag[\lambda] = \bigoplus_{\mu \in \N^+(\lambda)}\Rep_{\mu}(g) .
\end{equation}

\subsection{Induction along the chain of subgroups}\label{subsec: First induction}

For all $m \in [n]$, let $\T_m$ be a left transversal for $S_{m-1} \subset S_m$, where $\T_1 = {e}$.
Denote by $\IndMap[][m]$ the linear map 
\begin{equation}
    \IndMap[][m] \defeq \bigoplus_{\lambda \in \widehat{S}_{m-1}} \IndMap[\lambda],
\end{equation}
where we have $\IndMap[\lambda]$ as in \cref{def: IndMap}.
We will now introduce a quantum circuit implementing the $S_n$-QFT matrix $\F_n$ from \cref{def: Fourier transform} using these induction maps.
The proof for its correctness is given in \cite{cookbook}.

\begin{figure}[H]
\centering
\includegraphics[width=\textwidth]{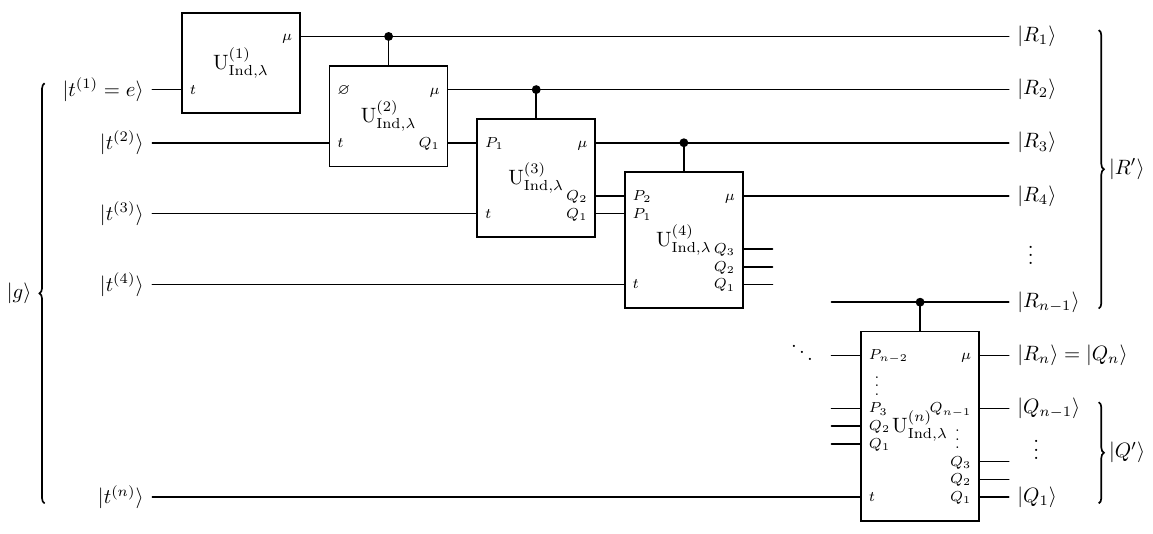}
\caption{Quantum circuit for the Fourier transform $\F_n$. Group elements $g \in S_n$ are encoded with transversal elements $t^{(m)} \in T_m$ as $g = t^{(n)}\dots t^{(3)}t^{(2)}$.For Fourier basis vectors, we let $P = P_1 \to P_2 \to \dots \to P_n$ and $Q = Q_1 \to Q_2 \to \dots \to Q_n$ encodes paths in the Bratteli diagram, where $P_m, Q_m \in \widehat{S}_m$.}
\label{fig: theorem first induction full circuit}
\end{figure}

\begin{theorem}\label{thm: first induction}
We can decompose the quantum Fourier transform for $G$ as a sequence of induction maps $\IndMap[][1], \IndMap[][2], \dots, \IndMap[][n]$ on the path register corresponding to the left regular representation.
Here, the irrep label is controlled by the endpoint of the other path.
\cref{fig: theorem first induction full circuit} shows the corresponding quantum circuit for the Fourier transform.
\end{theorem}

In this theorem, we encode group elements using a sequence of transversal registers.
This is possible by \cref{lemma: decomposing group elements with transversal}, which can be used to show that every group element $g \in S_{n}$ has a unique decomposition 
\begin{equation}
    g = t^{(n)}t^{(n-1)}\dots t^{(2)}t^{(1)}
\end{equation}
with $t^{(m)}\in \T_m$.

\subsection{Relative path encoding}\label{subsec: Fourier transform with relative paths}
The circuit of \cref{fig: theorem first induction full circuit} encodes paths by storing all nodes of the paths in separate registers.
For the symmetric group, we saw in \cref{subsec: Yamanouchi words} that such an absolute encoding does not have an optimal asymptotic space complexity.
However, encoding paths as Yamanouchi words does achieve optimal scaling.

To implement the Fourier transform as in \cref{thm: first induction} with paths encoded as Yamanouchi words, we need two more classical gates, based on the work of \cite{SnQFT_Speedup}.

First, consider a path $P \in \Path(\lambda)$ of length $n{-}1$ and a path $Q \in \Path(\mu)$ of length $n$, with $\lambda \in \N^-(\mu)$.
Recall that $P \to \mu$ denotes the extension of $P$ to $\mu$.
As Yamanouchi words, extending $p$ to the endpoint of $q$ involves adding a new entry to $o$.
We will denote this new entry by $q {\setminus} p$.
We will denote the corresponding gate mapping $\ket{p}\ket{0}\ket{q}$ to $\ket{p}\ket{q {\setminus} p}\ket{q}$ by $q {\setminus} p$.

Second, with $\lambda$, $\mu$, $p$ and $q$ as before, suppose we have a register $\ket{\lambda}$ and a register $\ket{q {\setminus} p}$.
Then, denote by $+\square$ the gate which converts $\lambda$ to $\mu$, controlled on the register $p {\setminus} q$.

Note that the input $\ket{t} \otimes \ket{p}$ of $\IndMap[][n]$ has exactly the same register structure as the output $\ket{q}$, as the subregister $\ket{q'} = \ket{q_1} \dots \ket{q_{n-1}}$ has the same structure as $\ket{p} = \ket{p_1} \dots \ket{p_{n-1}}$, and both $\ket{t}$ and $\ket{q_m}$ are $|\T_m|$-dimensional registers.
The information of the output $\ket{\mu}$ is implicitly stored in $\ket{q}$, but as this register is needed as a control register in the QFT circuit of \cref{thm: first induction}, and also needed when the QFT is used for weak Fourier sampling, we will not uncompute it.

At every layer of the circuit of \cref{thm: first induction}, we replace the block $\IndMap[\lambda][n]$ by the circuit in \cref{fig:IndMap_induction_step_rel_version}.
\begin{figure}[H]
\centering
\includegraphics[width=\textwidth]{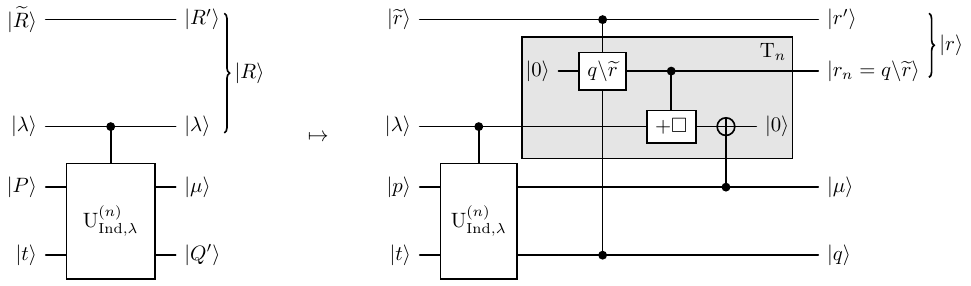}
\caption{Updated induction step for the Fourier transform with relative path encoding.
The output of the quantum Fourier transform for $S_{n-1}$ is stored in $\ket{\widetilde{r}}\ket{\lambda}\ket{p}$, with $\widetilde{r}$ and $p$ encoding path $\widetilde{R}, Q \in \Path(\lambda)$.
We have $\ket{t}$ a transversal register for $S_{n-1} \subset S_n$ like before.}
\label{fig:IndMap_induction_step_rel_version}
\end{figure}

With the relative encoding, we update the circuit of \cref{fig: theorem first induction full circuit} as in \cref{fig: theorem first induction full circuit relative paths}.
\begin{figure}[H]
\centering
\includegraphics[width=\textwidth]{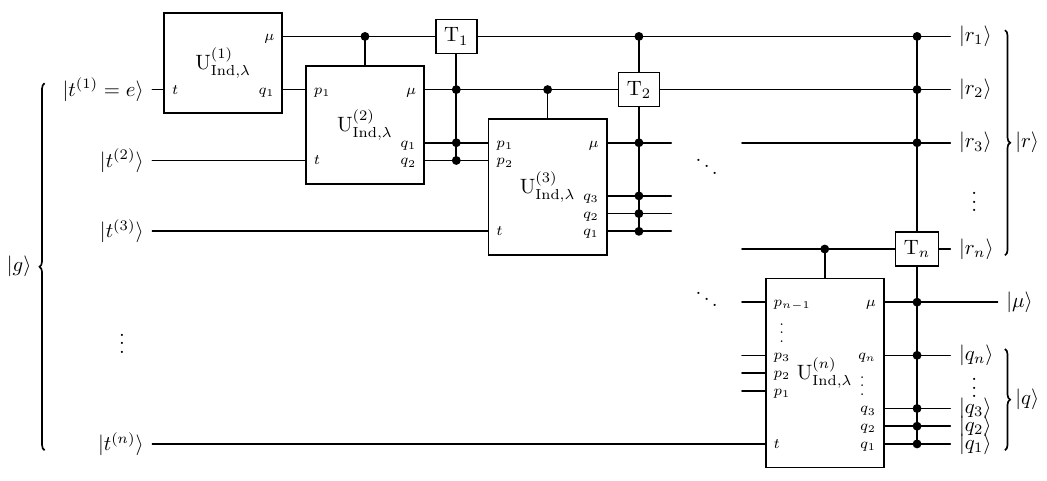}
\caption{Updated implementation of the Fourier transform for $S_n$.
The circuit is the same as the circuit in \cref{fig: theorem first induction full circuit}, but uses a Yamanouchi word encoding instead of storing the full paths.}
\label{fig: theorem first induction full circuit relative paths}
\end{figure}

\cite{SnQFT_Speedup} provide quantum algorithms for $q {\setminus} p$ and $+\square$, using only coherent arithmetic. 
Now, we will discuss two different approaches of the induction map, leading to two different QFT implementations. 
The first one uses a method related to the work of \cite{Beals}, as worked out in more detail by \cite{GenQFT}.
The second one is inspired on the algorithm by \cite{SnQFT_Speedup}, and uses Mackey theory to derive a second induction implementing the induction map.

\section{Constructing \texorpdfstring{$\IndMap$}{Uind} for \texorpdfstring{$S_n$}{Sn}: Beals approach}\label{sec: induction map Beals}

In this section, we describe how \cite{Beals} implements $\IndMap$ for constructing the QFT of the symmetric group.
This algorithm was generalized to a method for constructing general QFTs by \cite{GenQFT}.
We start by introducing the circuit implementing $\IndMap[]$ from \cref{def: IndMap} for any finite group $G$ with subgroup $H \subset G$, where irreps are unitary, with multiplicity-free and subgroup adapted restrictions.
Let $n \defeq [G : H]$, and fix a left transversal $\T = \{t_1, \dots, t_n\}$.
The algorithm implementing $\IndMap[]$ treats transversal elements one by one, replacing them by a special state $\ket{\star}$. 
At the end of the circuit, the transversal register will be entirely in the state $\ket{\star}$, so that it can be uncomputed.

\subsection{General approach: Looping through transversal elements}

In our algorithm, we have a transversal register, a path-register in relative encoding, an irrep label register storing the end-point of the path explicitly, and a control register with the end-point $\lambda$ of the incoming path.
Hence, the registers are the same as in \cref{fig: theorem first induction full circuit relative paths}, but with an added end-point register for the path.
Let us now introduce the three gates which we will use in our construction. 

\begin{definition}\label{def: U gate Beals}
    Let $\lambda \in \widehat{H}$ and $P \in \Path(\lambda)$, so that $\ket{P}$ is a basis vector for the simple $\C[H]$-module $V_\lambda$.
    The gate $\mathrm{U}_\uparrow$ acts on $\ket{P}$ by extending it to a superposition of successors of $\lambda$ as
    \begin{align}
        \mathrm{U}_\uparrow \ket{\star}\ket{P}\ket{\lambda}\ket{\lambda} &= \sum_{\mu \in \N^+(\lambda)} \sqrt{\frac{d_{\mu}}{n d_\lambda}}\ket{\star}\ket{P \to \mu}\ket{\mu}\ket{\lambda},
    \end{align}
    and acts as identity when the transversal register is not $\ket{\star}$.
\end{definition}

\begin{definition}\label{def: V gate beals}
    The gate $\mathrm{V}_k$ acts as an $\mathrm{X}$-gate on the subspace of the transversal register spanned by $\ket{t_k}$ and $\ket{\star}$ when the irrep label register is equal to the control register.
    Consider $\lambda \in \widehat{H}$ and $P \in \Path(\lambda)$.
    Then, $\mathrm{V}_k$ acts as
    \begin{align}
        \mathrm{V}_k \ket{\star}\ket{P}\ket{\lambda}\ket{\lambda} &= \ket{t_k}\ket{P}\ket{\lambda}\ket{\lambda} \\
        \mathrm{V}_k \ket{t_k}\ket{P}\ket{\lambda}\ket{\lambda} &=  \ket{\star}\ket{P}\ket{\lambda}\ket{\lambda}.
    \end{align}
    When the irrep label register and the control register, corresponding to $\ket{\lambda}\ket{\lambda}$ in the equations above, are not equal, the gate $\mathrm{V}_k$ acts as identity.
\end{definition}

\begin{definition}\label{def: irrep gate Beals}
    The gate $\mathrm{R}(t_k)$ acts as the irrep of the transversal element $t_k$ on the path register when the transversal register is in the state $\ket{\star}$. 
    Consider $\mu \in \widehat{G}$ and $P \in \Path(\mu)$.
    Then, $\mathrm{R}(t_k)$ acts as
    \begin{equation*}
        \mathrm{R}(t_k)\ket{\star}\ket{P}\ket{\mu} = \ket{\star}\left(\Rep_\mu(t_k)\ket{P}\right)\ket{\mu}
    \end{equation*}
    When the transversal register is not in the state $\ket{\star}$, the gate $\mathrm{R}(t_k)$ acts as identity.
\end{definition}

\begin{theorem}
    The induction transform can be decomposed as
    \begin{equation}\label{eq: Beals circuit as an equation}
        \IndMap = \prod_{k=1}^n \mathrm{R}(t_k) \mathrm{U}_{\uparrow} \mathrm{V}_k \mathrm{U}_{\uparrow}^\dagger \mathrm{R}^\dagger(t_k).
    \end{equation}
    We use the convention where $k=1$ acts first and $k = n$ acts last.
    \cref{fig: Beals circuit} shows the corresponding quantum circuit implementing $\IndMap$.
\end{theorem}

\begin{figure}[H]
    \centering
    \includegraphics[]{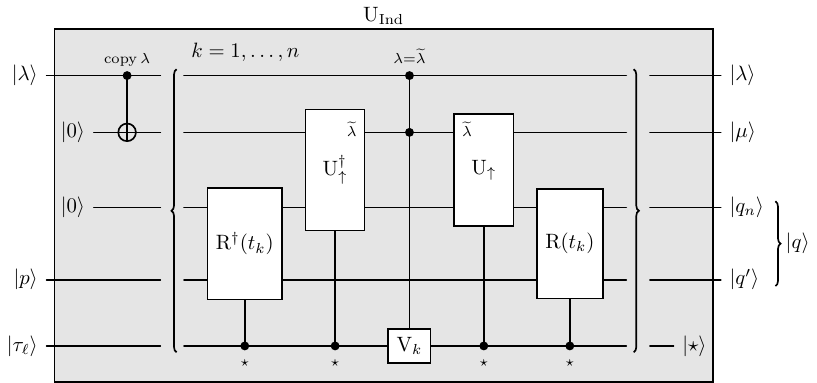}
    \caption{Implementation for $\IndMap[][n]$ of \cite{Beals}, later generalized by \cite{GenQFT}.}
    \label{fig: Beals circuit}
\end{figure}

The proof is given in chapter 4 of \cite{cookbook}.

\subsection{Implementation for \texorpdfstring{$S_n$}{Sn}}\label{subsec: Beals for Sn}
To use the Beals algorithm for the implementation of the $S_n$-QFT, we will first look at some general simplifications.

First, when $k=1$, the gate $\mathrm{R}^\dagger(t_1)$ and $\mathrm{U}_{\uparrow}$ will always act trivially, and can therefore be neglected.
Similarly, when $k = n$, the gate $\mathrm{U}_{\uparrow}^\dagger$ and $\mathrm{R}(t_n)$ act trivially.
Lastly, at the end of step $k$ and the start of step $k{+}1$ for $k \in [n{-}1]$, we can contract 
\begin{equation}
    \mathrm{R}^\dagger(t_k)\mathrm{R}(t_k) = \mathrm{R}(t_k^{-1}t_k),
\end{equation}
using group-homomorphism and that irreps are unitary.
If we fix our transversals $\T_n$ for $S_{n-1} \subset S_n$ as
\begin{equation}
    \T_n \defeq \{t^{(n)}_i \defeq (1, \dots, n)^{i} \mid i \in [n]\},
\end{equation}
we find that 
\begin{equation}
    \mathrm{R}^\dagger(t_k)\mathrm{R}(t_{k-1}) = \mathrm{R}((1, \dots, n)^{-i})) = \mathrm{R}(\sigma_{n-1}\dots\sigma_1) = \mathrm{R}(\sigma_{n-1})\dots\mathrm{R}(\sigma_1).
\end{equation}
We update \cref{fig: Beals circuit} as in \cref{fig: simplified Beals circuit for Sn}.

\begin{figure}[H]
    \centering
    \includegraphics[]{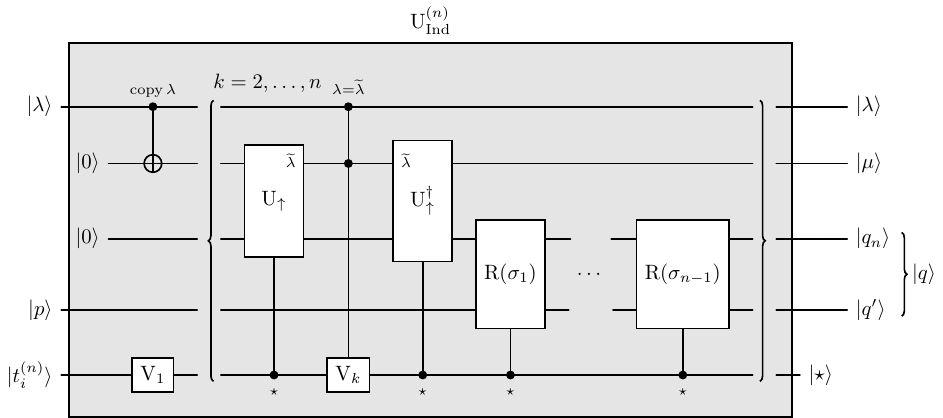}
    \caption{Beals' implementation of $\IndMap[]$ for $S_{n}$, with simplifications.}
    \label{fig: simplified Beals circuit for Sn}
\end{figure}
\cref{app: low level circuits Beals} includes descriptions of low-level implementations of the gates used in \cref{fig: simplified Beals circuit for Sn}. 
There, it becomes clear that the low-level implementation of $\mathrm{R}(\sigma_{k-1})$ requires access to $x$-coordinates of all cells in the SYT corresponding to the path, besides the Yamanouchi-word register $\ket{p}$.
This is not shown in \cref{fig: simplified Beals circuit for Sn}, but the corresponding extra register has the same structure as a Yamanouchi word register.

\subsection{Complexity}
For the gates in \cref{fig: simplified Beals circuit for Sn}, \cref{app: low level circuits Beals} finds the following complexities as in \cref{tab:Beals complexities table}.

\begin{table}[H]
\centering
\begin{tabular}{l|ccc}
\textbf{Gate} & \textbf{Space} & \textbf{Gatecount}  & \textbf{Depth} \\ \hline
$\mathrm{V}_k $    &       $\mathcal{O}(n)$                          &   $\mathcal{O}(n)$     &         $\widetilde{\mathcal{O}}(1)$ \\
$\mathrm{U}_\uparrow$ &     $\widetilde{\mathcal{O}}(n)$          &  $\widetilde{\mathcal{O}}(n)$ & $\widetilde{\mathcal{O}}(\sqrt{n})$                \\
$\mathrm{R}(\sigma_k) $       &     $\widetilde{\mathcal{O}}(n)$           &                   $\widetilde{\mathcal{O}}(1)$ & $\widetilde{\mathcal{O}}(1)$            
\end{tabular}
\caption{Complexities of the gates used in Beals' implementation of $\IndMap[][n]$.}
\label{tab:Beals complexities table}
\end{table}

\begin{lemma}\label{lemma: IndMap Beals complexity}
    For $n \geq 1$, the circuit in \cref{fig: simplified Beals circuit for Sn} implementing $\IndMap[][n]$ uses $\widetilde{\mathcal{O}}(n)$ qubits, $\widetilde{\mathcal{O}}(n^2)$ elementary gates and has depth $\widetilde{\mathcal{O}}(n^{2})$.
\end{lemma}

\begin{proof}
    First, consider the space complexity. 
    During the initialization, $V_1$ is not controlled, and only needs a polylog ($\widetilde{\mathcal{O}}(1)$) number of ancillas. 
    The space complexity is dominated by the $\mathcal{O}(n)$-qubit shape registers (\cref{lemma: shape register qubit count}), and the gate count of the initialization step is $\mathcal{O}(n)$.
    A single loop also uses at most $\widetilde{\mathcal{O}}(n)$ ancillary qubits at a time, as can be seen from \cref{tab:Beals complexities table}.
    Hence, $\IndMap[][n]$ can be implemented using $\widetilde{\mathcal{O}}(n)$ qubits. 

    Second, for the size and depth, the initialization step uses $\mathcal{O}(n)$ parallel gates for copying $\ket{\lambda}$, which is a circuit of constant depth. 
    The same is true for $\mathrm{V}_1$. 
    Hence, the initialization step uses $\widetilde{\mathcal{O}}(n)$ gates and has a depth of $\widetilde{\mathcal{O}}(1)$.
    From \cref{tab:Beals complexities table}, we know that a single loop uses $\widetilde{\mathcal{O}}(n)$ gates and has the same depth.
    As there are $n{-}1$ loops in total, the gate count and depth of $\IndMap[][n]$ is $\widetilde{\mathcal{O}}(n^2)$.
\end{proof}

\section{Constructing \texorpdfstring{$\IndMap$}{Uind} for \texorpdfstring{$S_n$}{Sn}: Mackey approach}\label{sec: induction map Mackey}

In this section, we will introduce a simplified and corrected implementation of the algorithm introduced by \cite{SnQFT_Speedup}.
This paper implements $\IndMap[\lambda][n]$---referred to as $\mathrm{H}_{n,\lambda}$ by Kawano and Sekigawa---recursively, using a sequence of four gates. 
First, a gate named $\mathrm{K}_n$ acts on the incoming path $\ket{p}$ as the left regular representation of group elements conditioned on the transversal register.
Second, the gate $\IndMap[][n{-}1]$ is used recursively, wrapped inside a few classical gates.
Third, a classical gate named $\mathrm{P}_n$ relabels basis vectors. 
Lastly, a gate $\mathrm{A}_n$, referred to as the \emph{embedding operation}, is at the core of relating $\IndMap[][n{-}1]$ to $\IndMap[][n]$.

However, there are a few issues.
First, the described basis after recursively applying $\IndMap[][n{-}1]$ does not appear to span the codomain in its definition.
While $\IndMap[\lambda]$ has all paths to successors of $\lambda$ in its image, \cite{SnQFT_Speedup} only consider the subset of these paths going through $\lambda$, and do not mention the other paths which we will refer to as \emph{detours}.
Second, their complexity analysis treats the gate $\mathrm{A}_n$ as implementable using $\mathcal{O}(n)$ elementary operations, using an argument based on matrix decomposition methods.
However, this matrix decomposition algorithm needs to be carried out coherently to obtain a polynomial size circuit, and the complexity analysis does mention the related computational overhead.
Unfortunately, the size of $\mathcal{O}(n)$ for implementing $\mathrm{A}_n$ does not hold in the usual gate count model.

By taking a step back, we will show how the second induction relation follows from Mackey theory. 
This allows us to fix a new choice of transversals which removes the need for gate $\mathrm{K}_n$ entirely.
Second, by correctly identifying the basis after the recursion, including detours, we show that we also do not need the classical relabeling gate $\mathrm{P}_n$.
Lastly, by explicitly providing a low-level implementation of $\mathrm{A}_n$, we show that the that the asymptotic gate count of the embedding operation is actually $\mathcal{O}(n^{3/2})$ instead of $\mathcal{O}(n)$, while the complexity $\mathcal{O}(n)$ only holds for the depth of the circuit.

\subsection{Mackey theory: Background}

This section summarizes the techniques from \cite{cookbook}, simplified and adapted to the context of the Symmetric group. 
First, let introduce Mackey theory in the way that we will use it for finding the second induction relation for implementing $\IndMap[][n]$.

\begin{definition}
    For a group $G$ and subgroups $H, K \subset G$, denote by $HgK \defeq \{hgk\mid h \in H, k \in K\}$ the \emph{$(H,K)$-double coset of $g$}.
    Denote by $H \backslash G / K = \{HgK\mid g \in G\}$ the set of all $(H, K)$-double cosets.
    If $H = K$, we call $H g H$ the \emph{$H$-double coset of $g$}.
\end{definition}

\begin{lemma}\label{lemma: double cosets partition group}
    For $H$, $K$ and $G$ as before, the set of double cosets $H \backslash G / K$ partitions $G$.
\end{lemma}

\begin{proof}
    To show that $(H, K)$-double cosets partition $G$, we will prove that for $g_1, g_2 \in G$, the binary relation $g_1 \sim g_2$ if $H g_1 K = H g_2 K$ is an equivalence relation.
    Note that $g_1 \sim g_2$ if and only if there exists a $h \in H$ and $k \in K$ such that $g_1 = h g_2 k$.
    \begin{itemize}
        \item Symmetry: if $g_1 \sim g_2$ with $h \in H$ and $k \in K$, then $g_2 \sim g_1$ with $h^{-1} \in H$ and $k^{-1} \in K$;
        \item Reflexivity: use $e \in H$ and $e \in K$ to show that $g \sim g$;
        \item Transitivity: if $g_1 \sim g_2$ with $h_1 \in H$ and $k_1 \in K$, and $g_2 \sim g_3$ with  $h_2 \in H$ and $k_2 \in K$, then $g_1 \sim g_3$ with $h_1 h_2 \in H$ and $k_2 k_1 \in K$.\qedhere
    \end{itemize}
\end{proof}

\begin{lemma}\label{lemma: Mackey subgroups}
    For $H$, $K$ and $G$ as before, let $\Omega$ be a set of $(H, K)$-double coset representatives. 
    Hence, for all $g \in G$, we have $|\Omega \cap HgK| = 1$.
    Then, for any $\omega \in \Omega$, the set $H_\omega \defeq H \cap \omega K \omega^{-1} \subseteq H$ is a subgroup.
\end{lemma}

\begin{proof}
    Since $K$ is a subgroup of $G$, we also know that $\omega K \omega^{-1}$ is a subgroup of $G$. 
    As $H_\omega$ is defined as an intersection of subgroups, and as it is a subset of $H$ by definition, we conclude that $H_\omega$ is a subgroup of $H$.
\end{proof}

\begin{lemma}\label{lemma: Mackey transversals}
    Let $H$, $K$, $G$ and $\Omega$ as before. 
    For all $\omega \in \Omega$, let $\T_\omega$ be a left transversal for $H_\omega \subseteq H$. 
    Then, the set 
    \begin{equation}\label{eq: transversal from Mackey subgroup transversals}
        \T \defeq \bigsqcup_{\omega \in \Omega} \T_\omega \omega
    \end{equation}
    is a left transversal for $K \subset G$.
\end{lemma}

\begin{proof}
    Let $g \in G$. 
    To show that $\T$ is a left transversal for $K$ in $G$, we show that $|gK \cap \T| = 1$.
    Note that by \cref{lemma: double cosets partition group}, there exist a $h \in H$ and $k \in K$ and a unique $\omega \in \Omega$ such that $g = h \omega k$. 
    Also, note that by \cref{lemma: decomposing group elements with transversal}, we have that $h = t' (\omega k' \omega^{-1})$ for a unique $\omega k' \omega^{-1} \in H_\omega$ and $t' \in \T_\omega$.
    We rewrite $g = t' \omega k' \omega^{-1} \omega k = t' \omega k'k$.
    It follows that $t' \omega \in g K$. 
    By definition of $\T$ in \cref{eq: transversal from Mackey subgroup transversals}, note that $t' \omega \in \T$ so that $|gK \cap \T| \geq 1$.

    While $\omega$ is unique by definition of $\Omega$, we still need to show uniqueness of $t'$. 
    Suppose that for a different choice of $h$ and $k$, we find that $h = t'' \omega k''\omega^{-1}$ with a possibly different $t'' \in \T_\omega$.
    With the same analysis as before, we find that that $t'' \omega \in g K$, which implies that $t''\omega K = t' \omega K$.
    We rewrite that $t''^{-1} t' \in \omega K \omega^{-1}$. 
    As $t'', t' \in \T_\omega \subseteq H$, we know that $t''^{-1} t' \in H$ so that $t''^{-1} t' \in H_\omega$.
    Hence, $t' H_\omega = t'' H_\omega$, which means that $t' = t''$ by definition of $\T_\omega$ as a transversal, showing uniqueness of $t'$.
    We conclude that $|gK \cap \T| = |\{t'\omega\}| = 1$.
\end{proof}

\begin{definition}[Mackey transform]\label{def: Mackey transform}
    Let $H$, $K$, $G$ and $\Omega$ as before.
    Let $\T = \sqcup_{\omega \in \Omega} \T_\omega \omega $ be as in \cref{lemma: Mackey transversals}.
    Define the \emph{Mackey transform} as the linear map $\C[\T] \to \bigoplus_{\omega \in \Omega} \C[\T_\omega]$ by its action on basis vectors $\ket{t} \in \C[\T]$ as 
    \begin{equation}
        \mathrm{M} \ket{t} = \ket{t\omega^{-1}},
    \end{equation}
    where the vector on the right is in the $\omega$-summand, with $\omega \in \Omega$ such that $t \in H \omega K$.
\end{definition}

Note that $t\omega^{-1} \in \T_\omega$.

\begin{lemma}[Mackey decomposition]\label{lemma: Mackey}
    Let $H$, $K$, $G$ and $\Omega$ as before, and pick a transversal $\T = \sqcup_{\omega \in \Omega} \T_\omega \omega $ as in \cref{lemma: Mackey transversals} to fix a basis for induced representations from $K$ to $G$ and from all $H_\omega$ to $H$.
    For any $\C[K]$-module $V$, with corresponding matrix representation $\Rep : K \to \U(V)$, we have 
    \begin{equation}
        \Res^G_H \Ind^G_K V \overset{\mathrm{M} \otimes \mathrm{I}_{V}}{\simeq} \bigoplus_{\omega \in \Omega} \Ind_{H_\omega}^H V^{(\omega)},
    \end{equation}
    where $V^{(\omega)}$ denotes a ``twisted'' $\C[H_\omega]$-module.
    The action of $h \in H_\omega$ on a vector $\ket{P} \in V^{(\omega)}$ is
    \begin{equation}
        h (\ket{P}) = \Rep(\omega^{-1}h \omega)\ket{P}.
    \end{equation}
    As $h = \omega k \omega^{-1}$ for some $k \in K$, note that $\omega^{-1}h \omega = k$, which means that we can recycle the definition of $V$ as $\C[K]$-module to define the twisted $\C[H_\omega]$-module $V^{(\omega)}$.
\end{lemma}

\subsection{Mackey theory and transversals for \texorpdfstring{$S_n$}{Sn}}

Now, we apply the theory from \cite{cookbook} to the symmetric group.
For this section, fix a level $n > 1$ on the chain of subgroups from \cref{eq: Sn chain of subgroups}.

\begin{lemma}\label{lemma: Sn-1 double cosets of Sn}
    There are two $S_{n-1}$-double cosets of $S_n$, and we can use $\Omega_n \defeq \{e, \sigma_{n-1}\}$ as a set of double coset representatives. 
    We find that the corresponding subgroups as in \cref{lemma: Mackey subgroups} are $H_e = S_{n-1}$ and $H_{\sigma_{n-1}} = S_{n-2}$.
\end{lemma}

The proof is given in \cref{app: Sn-1 double cosets of Sn}.
Now, we can use \cref{lemma: Mackey transversals} to we find a suitable set of transversals along the chain of subgroups of $S_n$.

\begin{proposition}[Transversals for along the chain of subgroups of the symmetric group]\label{prop: Mackey transversals for Sn}
    Denote by $\T_n$ left transversal for $S_{n-1} \subset S_n$, and let $\T_1 = \{e\}$.
    Then, \cref{lemma: Mackey transversals} together with $\Omega_n$ from \cref{lemma: Sn-1 double cosets of Sn} as a set of double coset representatives fixes
    \begin{equation}\label{eq: transversals explicit form}
        \T_n \defeq \{t_{i}^{(n)} \defeq (i, \dots, n) \mid i \in [n]\},
    \end{equation}
    where $t_{n}^{(n)} = e$. 
    With this indexing of transversal elements, $i < n$ implies that $t_{i}^{(n)} = t_{i}^{(n-1)} \sigma_{n-1}$.
\end{proposition}

\begin{proof}
    Recall that \cref{lemma: Mackey transversals} states that $\T_n = \bigsqcup_{\omega \in \Omega_n} \T_\omega \omega = \T_e e \sqcup T_{\sigma_{n-1}} \sigma_{n-1}$ is a left transversal for $S_{n-1} \in S_n$, where $\T_\omega$ is a left transversal for $H_\omega \subseteq S_{n-1}$.
    From \cref{lemma: Sn-1 double cosets of Sn}, we know that $H_e = S_{n-1}$, which has $\T_e = \{e\}$ as a most canonical choice of transversal, and that $H_{\sigma_{n-1}} = S_{n-2}$, so that we can inductively use $\T_{\sigma_{n-1}} = \T_{n-1}$.
    We find that
    \begin{equation}
        \T_n = \{e\} \sqcup T_{n-1} \sigma_{n-1}.
    \end{equation}
    We will now prove the explicit form of $\T_n$ by induction on $n$
    
    For the base case of $n = 1$, we have $\T_1 = \{e\}$ by definition, and there is no $\sigma_{n-1}$. 
    For $n > 1$, the induction hypothesis states that $\T_{n-1} \defeq \{t_{i}^{(n-1)} \defeq (i, \dots, n{-}1) \mid i \in [n{-}1]\}$.
    We find 
    \begin{align}
        \T_n &= \{e\} \sqcup \{t_{i}^{(n-1)} = (i, \dots, n{-}1) \mid i \in [n{-}1]\} \sigma_{n-1}\\
        &= \{t_{i}^{(n-1)}\sigma_{n-1} = (i, \dots, n{-}1)\sigma_{n-1} \mid i \in [n{-}1]\} \sqcup \{e\}\\
        &= \{t_{i}^{(n-1)}\sigma_{n-1} = (i, \dots, n) \mid i \in [n{-}1]\} \sqcup \{e\}.
    \end{align}
    If we define $t_{i}^{(n)} \defeq t_{i}^{(n-1)}\sigma_{n-1} = (i, \dots, n)$ for $i < n$, and $t_{n}^{(n)} \defeq e$, we find back the explicit form as 
    \begin{equation}
        \T_n = \{t_{i}^{(n)} = (i, \dots, n) \mid i \in [n]\}.
    \end{equation}\qedhere
\end{proof}

Next, we will introduce two $\C[S_{n-1}]$-module isomorphisms, both decomposing the module $\Res_{S_{n-1}}^{S_n}\Ind_{S_{n-1}}^{S_n} V_\lambda$ as a direct sum of simple $S_{n-1}$-modules, with different multiplicity spaces.
For the first isomorphism, which is a corollary of \cref{lemma: IndMap basis transform}, define the following space.

\begin{definition}
    Fix $\lambda \in \widehat {S}_{n-1}$. 
    For all $\widetilde{\lambda} \in \widehat{S}_{n-1}$, define the vector space 
    \begin{equation}
        M^{+}_{\lambda, \widetilde{\lambda}} \defeq \C^{\N^+(\lambda) \cap \N^+(\widetilde{\lambda})}.
    \end{equation}
\end{definition}
Clearly, if $\lambda = \widetilde{\lambda}$, we have that $M^{+}_{\lambda, \widetilde{\lambda}} = \C^{\N^+(\lambda)}$.
Otherwise, if $\lambda \neq \widetilde{\lambda}$, \cref{lemma: diamond} implies that $M^{+}_{\lambda, \widetilde{\lambda}}$ is zero- or one-dimensional. 

\begin{corollary}\label{cor: indmap restricted to Sn-1}
    Fix $\lambda \in \widehat{S}_{n-1}$.
    Then, we have the following isomorphism of $\C[S_{n-1}]$-modules:
    \begin{equation}\label{eq: isomorhism indmap restricted to Sn-1}
        \Res^{S_{n}}_{S_{n-1}} \Ind^{S_{n}}_{S_{n-1}} V_\lambda \overset{\IndMap[\lambda][n]}{\simeq} \bigoplus_{\widetilde{\lambda} \in \widehat{S}_{n-1}} V_{\widetilde{\lambda}} \otimes M^{+}_{\lambda, \widetilde{\lambda}}.
    \end{equation}
\end{corollary}

\begin{proof}
    In \cref{lemma: IndMap basis transform}, set $H = S_{n-1}$ and $G = S_{n}$, so that
    \begin{equation}
        \Ind^{S_{n}}_{S_{n-1}} V_\lambda \overset{\IndMap[\lambda][n]}{\simeq} \bigoplus_{\mu \in \N^+(\lambda)} V_\mu.
    \end{equation}
    As restriction does not change the representation spaces, restricting both sides back to $S_{n-1}$ gives
    \begin{align}
        \Res^{S_{n}}_{S_{n-1}} \Ind^{S_{n}}_{S_{n-1}} V_\lambda \overset{\IndMap[\lambda][n]}{\simeq} \Res^{S_{n}}_{S_{n-1}}\bigoplus_{\mu \in \N^+(\lambda)} V_\mu = \bigoplus_{\mu \in \N^+(\lambda)}  \Res^{S_{n}}_{S_{n-1}} V_\mu
    \end{align}
    Using that all $\mu \in \widehat{S}_n$ are specified in a subgroup-adapted basis with respect to $S_{n-1}$, we rewrite
    \begin{align}
        \Res^{S_{n}}_{S_{n-1}} \Ind^{S_{n}}_{S_{n-1}} V_\lambda \overset{\IndMap[\lambda][n]}{\simeq} \bigoplus_{\mu \in \N^+(\lambda)} \bigoplus_{\widetilde{\lambda} \in \N^-(\mu)} V_{\widetilde{\lambda}} \otimes \C^{\{\mu\}}.
    \end{align}
    Here, we keep the one-dimensional space $\C^{\{\mu\}}$ to make it explicit in which summand a state lives. 
    Note that all spaces $V_{\widetilde{\lambda}} \otimes \C^{\{\mu\}}$ in the big double direct sum on the RHS bijectively correspond to an element of the set
    \begin{align}
        &\{(\widetilde{\lambda}, \mu) \in \widehat{S}_{n-1} \times \widehat{S}_n \mid \mu \in \N^+(\lambda) \text{ and } \widetilde{\lambda} \in \N^-(\mu)\}\\ &= \{(\widetilde{\lambda}, \mu) \in \widehat{S}_{n-1} \times \widehat{S}_n \mid \mu \in \N^+(\lambda) \text{ and } \mu \in \N^+(\widetilde{\lambda})\} && \text{(by Frobenius reciprocity)}\\
        &= \{(\widetilde{\lambda}, \mu) \in \widehat{S}_{n-1} \times \widehat{S}_n \mid \mu \in \N^+(\lambda) \cap \N^+(\widetilde{\lambda})\}. &&\text{(by definition of intersection)}
    \end{align}
    We continue rewriting
    \begin{align}
        \Res^{S_{n}}_{S_{n-1}} \Ind^{S_{n}}_{S_{n-1}} V_\lambda \overset{\IndMap[\lambda][n]}{\simeq} \bigoplus_{\displaystyle \overset{(\widetilde{\lambda}, \mu) \in \widehat{S}_{n-1} \times \widehat{S}_n}{\scriptstyle \mu \in \N^+(\lambda) \cap \N^+(\widetilde{\lambda})}} V_{\widetilde{\lambda}} \otimes \C^{\{\mu\}} = \bigoplus_{\widetilde{\lambda} \in \widehat{S}_{n-1}} V_{\widetilde{\lambda}} \otimes \bigoplus_{\mu \in \N^+(\lambda) \cap \N^+(\widetilde{\lambda})} \C^{\{\mu\}}.
    \end{align}
    Finally, using that $M^+_{\lambda, \widetilde{\lambda}} = \bigoplus_{\mu \in \N^+(\lambda) \cap \N^+(\widetilde{\lambda})} \C^{\{\mu\}}$, we find back \cref{eq: isomorhism indmap restricted to Sn-1}.
\end{proof}

For the second isomorphism, which is a corollary of \cref{lemma: Mackey}, consider the following vector space.
\begin{definition}
    Fix $\lambda \in \widehat{S}_{n-1}$. 
    For all $\widetilde{\lambda} \in \widehat{S}_{n-1}$, define the short-hand notation
    \begin{align}
        M^{-}_{\lambda,\widetilde{\lambda}} &\defeq \begin{cases}
            \C^{\{\star\}} \oplus \C^{\N^-(\lambda)} &\text{if } \lambda = \widetilde{\lambda},\\
            \C^{\N^-(\lambda) \cap \N^-(\widetilde{\lambda})} &\text{if } \lambda \neq \widetilde{\lambda}.
        \end{cases}
    \end{align}
\end{definition}

Using $\Omega_n$ from \cref{lemma: Sn-1 double cosets of Sn} as set of $S_{n-1}$-double coset representatives of $S_n$ and transversal $\T_n$ for $S_{n-1}\subset S_n$ as in \cref{prop: Mackey transversals for Sn}, we find the Mackey transform from \cref{def: Mackey transform} with $H = K = S_{n-1}$ and $G = S_n$ as in the following definition

\begin{definition}[Mackey transform for $S_{n-1}$-double cosets of $S_n$]
    Define the linear map $\mathrm{M}_n : \C [\T_n] \to \C^{\{\star\}} \oplus \C[\T_{n-1}]$ to act on basis vectors $\ket{t^{(n)}_i}$ of $\C[\T_n]$ as 
    \begin{equation}
        \mathrm{M}_n \ket{t^{(n)}_i} = 
        \begin{cases}
            \ket{t^{(n-1)}_i} &\text{if } i < n,\\
            \ket{\star} &\text{if } i = n.
        \end{cases}
    \end{equation}
    Where we introduced a special symbol $\star$ for labeling the trivial double coset.
\end{definition}

Note that if we encode transversal elements by there index, and $\star$ by $n$, the implementation of $\mathrm{M}_n$ is trivial.

\begin{corollary}\label{cor: mackey for Sn}
    For $\lambda \in \widehat{S}_{n-1}$, we have
    \begin{equation}\label{eq: Mackey for Sn}
        \Res_{S_{n-1}}^{S_n}\Ind_{S_{n-1}}^{S_n} V_\lambda \overset{\left(\mathrm{I}_{d_\lambda} \oplus \bigoplus_{\kappa \in \N^-(\lambda)}\IndMap[\kappa ][n-1]\right) \left(\mathrm{M}_n \otimes \mathrm{I}_{d_\lambda}\right)}{\simeq} \bigoplus_{\widetilde{\lambda} \in \widehat{S}_{n-1}} V_{\widetilde{\lambda}} \otimes M^{-}_{\lambda, \widetilde{\lambda}}
    \end{equation}
\end{corollary}

\begin{proof}
    The isomorphism consists of a composition of two maps.
    The first map is given by $\mathrm{M}_n \otimes \mathrm{I}_{d_\lambda}$, which exactly implements the Mackey decomposition of \cref{lemma: Mackey}, with $G = S_n$ and $H = K = S_{n-1}$.
    Explicitly, we get 
    \begin{equation}
        \Res_{S_{n-1}}^{S_n}\Ind_{S_{n-1}}^{S_n} V_\lambda \overset{\mathrm{M}_n \otimes \mathrm{I}_{d_\lambda}}{\simeq} \Ind^{S_{n-1}}_{H_e} V_\lambda^{(e)} \oplus \Ind^{S_{n-1}}_{H_{\sigma_{n-1}}} V_\lambda^{(\sigma_{n-1})}.
    \end{equation}
    Recall from \cref{lemma: Sn-1 double cosets of Sn} that $H_e = S_{n-1}$ and $H_{\sigma_{n-1}} = S_{n-2}$. 
    Conjugating elements of $S_{n-1}$ by the identity $e$, and conjugating elements of $S_{n-2}$ by $\sigma_{n-1}$ corresponds to an identity map.
    Hence, the action of subgroup elements on the twisted module is the same as on the the original simple module $V_{\lambda}$.
    We simplify
    \begin{equation}
        \Res_{S_{n-1}}^{S_n}\Ind_{S_{n-1}}^{S_n} V_\lambda \overset{\mathrm{M}_n \otimes \mathrm{I}_{d_\lambda}}{\simeq} \C^{\{\star\}} \otimes V_\lambda \oplus \Ind^{S_{n-1}}_{S_{n-2}} \Res^{S_{n-1}}_{S_{n-2}}V_\lambda.
    \end{equation}
    Using that we picked an $S_{n-2}$-adapted basis for $V_{\lambda}$, we find that 
    \begin{equation}
        \Res_{S_{n-1}}^{S_n}\Ind_{S_{n-1}}^{S_n} V_\lambda \overset{\mathrm{M}_n \otimes \mathrm{I}_{d_\lambda}}{\simeq} \C^{\{\star\}} \otimes V_\lambda \oplus  \Ind^{S_{n-1}}_{S_{n-2}} \bigoplus_{\kappa \in \N^-(\lambda)} V_\kappa \otimes \C^{\{\kappa\}}.
    \end{equation}
    Just like in the proof of \cref{cor: indmap restricted to Sn-1}, we keep the one dimensional space $\C^{\{\kappa\}}$ to explicitly label the summand to which each state belongs. 
    As we do not change the set indexing basis vectors for the RHS, but only the ordering, we say that 
    \begin{equation}
        \Res_{S_{n-1}}^{S_n}\Ind_{S_{n-1}}^{S_n} V_\lambda \overset{\mathrm{M}_n \otimes \mathrm{I}_{d_\lambda}}{\simeq} \C^{\{\star\}} \otimes V_\lambda \oplus \bigoplus_{\kappa \in \N^-(\lambda)} \Ind^{S_{n-1}}_{S_{n-2}} V_\kappa \otimes \C^{\{\kappa\}}.
    \end{equation}

    The second isomorphism in the composition is the map $\mathrm{I}_{d_\lambda} \oplus \bigoplus_{\kappa \in \N^-(\lambda)}\IndMap[\kappa ][n-1]$. 
    To match the ordering of wires in the final quantum algorithm, we define $\mathrm{I}_{d_\lambda}$ to swap the order of $\C^{\{\star\}}$ and $V_\lambda$ in the tensor product before the direct sum, which does not change the mathematical structure, as $\C^{\{\star\}}$ is one-dimensional.
    In other words, it implicitly includes the canonical tensor swap isomorphism.
    By \cref{lemma: IndMap basis transform}, we know that $\IndMap[\kappa ][n-1]$ decomposes each induced irrep $\Ind^{S_{n-1}}_{S_{n-2}} V_\kappa$ into a direct sum of $S_{n-1}$-irreps. 
    We find that $\mathrm{I}_{d_\lambda} \oplus \bigoplus_{\kappa \in \N^-(\lambda)}\IndMap[\kappa ][n-1]$ implements the isomorphism
    \begin{align}
        \C^{\{\star\}} \otimes V_\lambda \oplus \bigoplus_{\kappa \in \N^-(\lambda)} \Ind^{S_{n-1}}_{S_{n-2}} V_\kappa \otimes \C^{\{\kappa\}} \simeq V_\lambda \otimes \C^{\{\star\}}  \oplus \left(\bigoplus_{\kappa \in \N^-(\lambda)} \bigoplus_{\widetilde{\lambda} \in \N^+(\kappa)} V_{\widetilde{\lambda}} \otimes \C^{\{\kappa\}}\right). 
    \end{align}
    For the part right of the direct sum on the RHS, every space $V_{\widetilde{\lambda}} \otimes \C^{\{\kappa\}}$ corresponds bijectively to an element of the set 
    \begin{align}
        &\{(\widetilde{\lambda}, \kappa) \in \widehat{S}_{n-1} \times \widehat{S}_{n-2} \mid \kappa \in \N^-(\lambda)  \text{ and } \widetilde{\lambda} \in \N^+(\kappa)\}\\
        &= \{(\widetilde{\lambda}, \kappa) \in \widehat{S}_{n-1} \times \widehat{S}_{n-2} \mid \kappa \in \N^-(\lambda) \text{ and } \kappa \in \N^-(\widetilde{\lambda})\} &&\text{(by Frobenius reciprocity)}\\
        &= \{(\widetilde{\lambda}, \kappa) \in \widehat{S}_{n-1} \times \widehat{S}_{n-2} \mid \kappa \in \N^-(\lambda) \cap \N^-(\widetilde{\lambda})\} &&\text{(by definition of intersection)}.
    \end{align}
    We rewrite 
    \begin{align}\label{eq: last eq in proof of cor mackey for Sn}
        \Res_{S_{n-1}}^{S_n}\Ind_{S_{n-1}}^{S_n} V_\lambda \simeq V_\lambda \otimes \C^{\{\star\}} \oplus \left(\bigoplus_{\widetilde{\lambda} \in \widehat{S}_{n-1}}  V_{\widetilde{\lambda}} \otimes \bigoplus_{\kappa \in \N^-(\lambda) \cap \N^-(\widetilde{\lambda})} \C^{\{\kappa\}}\right).
    \end{align}
    When $\widetilde{\lambda} = \lambda$, note that $\N^-(\lambda) \cap \N^-(\widetilde{\lambda}) = \N^-(\lambda)$, and that there is one extra copy of $V_{\lambda}$ with label $\star$ coming. 
    Without changing basis vector labels, but only changing the order of basis vectors, the RHS of \cref{eq: last eq in proof of cor mackey for Sn} is exactly the RHS of \cref{eq: Mackey for Sn}.
    To summarize, we found that 
    \begin{align}
        \Res_{S_{n-1}}^{S_n}\Ind_{S_{n-1}}^{S_n} V_\lambda &\overset{\mathrm{M}_n \otimes \mathrm{I}_{d_\lambda}}{\simeq}   \C^{\{\star\}} \otimes V_\lambda \oplus \bigoplus_{\kappa \in \N^-(\lambda)} \Ind^{S_{n-1}}_{S_{n-2}} V_\kappa \otimes \C^{\{\kappa\}} \notag \\
        &\overset{\mathrm{I}_{d_\lambda} \oplus \bigoplus_{\kappa \in \N^-(\lambda)}\IndMap[\kappa ][n-1]}{\simeq} \bigoplus_{\widetilde{\lambda} \in \widehat{S}_{n-1}} V_{\widetilde{\lambda}} \otimes M^{-}_{\lambda, \widetilde{\lambda}} 
    \end{align}
    which concludes our proof.
\end{proof}

\begin{lemma}\label{lemma: commutative diagram for embedding operation}
    Let $\mathrm{A}$ be an intertwiner implementing the $\C[S_{n-1}]$-module isomorphism 
    \begin{equation}\label{eq: isomorphism for embedding operation}
        \bigoplus_{\widetilde{\lambda} \in \widehat{S}_{n-1}} V_{\widetilde{\lambda}} \otimes M^-_{\lambda, \widetilde{\lambda}} \simeq \bigoplus_{\widetilde{\lambda} \in \widehat{S}_{n-1}} V_{\widetilde{\lambda}} \otimes M^+_{\lambda, \widetilde{\lambda}}
    \end{equation}
    Then, $\mathrm{A}$ is of the form
    \begin{equation}\label{eq: block structure A}
        \mathrm{A} = \bigoplus_{\widetilde{\lambda} \in \widehat{S}_{n-1}} \mathrm{I}_{d_{\widetilde{\lambda}}} \otimes \mathrm{A}_{\lambda, \widetilde{\lambda}} ,
    \end{equation}
    with invertible matrices $\mathrm{A}_{\lambda, \widetilde{\lambda}} : M^-_{\lambda, \widetilde{\lambda}} \to M^+_{\lambda, \widetilde{\lambda}}$.
\end{lemma}

\begin{proof}
    First, note that an isomorphism $\mathrm{A}$ must exist, as both the LHS and RHS of \cref{eq: isomorphism for embedding operation} are isomorphic to $\Res^{S_n}_{S_{n-1}}\Ind^{S_n}_{S_{n-1}} V_\lambda$ by \cref{cor: indmap restricted to Sn-1,cor: mackey for Sn}, and module-isomorphism is transitive. 
    As both sides fully decompose $\Res^{S_n}_{S_{n-1}}\Ind^{S_n}_{S_{n-1}} V_\lambda$ into simple modules, we know that the multiplicity spaces of $V_{\widetilde{\lambda}}$ on both sides must be of equal dimension. 
    Moreover, the action of group elements on the $V_{\widetilde{\lambda}}$ part on both sides must be the equal.
    Therefore, the intertwiner can only act non-trivially on the multiplicity spaces, and must be of the form of \cref{eq: block structure A}.
\end{proof}

The following commutative diagram summarizes \cref{cor: indmap restricted to Sn-1,cor: mackey for Sn,lemma: commutative diagram for embedding operation}.
\begin{figure}[H]
    \centering
    \begin{tikzcd}[column sep=5.5cm, row sep=1.2cm, nodes={inner sep=3pt}]
        \Res^{S_n}_{S_{n-1}}\Ind^{S_n}_{S_{n-1}}V_\lambda
            \arrow[r, "\mathrm{U}_{\mathrm{Ind},\lambda}^{(n)}"]
            \arrow[d, "\mathrm{M}_n \otimes \mathrm{I}_{d_\lambda}"']
        &
        \displaystyle
        \bigoplus_{\widetilde{\lambda} \in \widehat{S}_{n-1}}
        V_{\widetilde{\lambda}} \otimes M^{+}_{\lambda, \widetilde{\lambda}}
        \\
        \left(\C^{\{\star\}} \otimes V_\lambda \right) \oplus \left(\bigoplus_{\kappa\in \N^-(\lambda)} \Ind^{S_{n-1}}_{S_{n-2}} V_\kappa \right) \arrow[r, "\mathrm{I}_{d_\lambda} \oplus \left(\bigoplus_{\kappa \in \N^-(\lambda)}\mathrm{U}_{\mathrm{Ind}, \kappa}^{(n-1)}\right)"']
        &
        \bigoplus_{\widetilde{\lambda}
        \in \widehat{S}_{n-1}}
        V_{\widetilde{\lambda}} \otimes M^{-}_{\lambda, \widetilde{\lambda}}
        \arrow[u, "\mathrm{A}"']
    \end{tikzcd}
    \label{fig: comm diagram for Sn-1 modules}
    \caption{Commutative diagram of $\C[S_{n-1}]$-modules, where all morphisms are isomorphisms. 
    The upper path corresponds to \cref{cor: indmap restricted to Sn-1}. 
    The two arrows from top left to bottom right implement \cref{cor: mackey for Sn}. 
    The map $\mathrm{A}$ must be as in \cref{lemma: commutative diagram for embedding operation} to make the diagram commute.}
\end{figure}

We will now extend this commutative diagram of $\C[S_{n-1}]$-modules to a diagram of $\C[S_n]$-modules.
Then, the path with the Mackey relabeling, the map $\mathrm{I}_{d_\lambda} \oplus \bigoplus_{\kappa \in \N^-(\lambda)} \IndMap[\kappa][n-1]$ and matrix $\mathrm{A}$ gives us an implementation for $\IndMap[\lambda][n]$. 
Here, the map $\mathrm{A}$ becomes exactly the \emph{embedding operation} in the algorithm of \cite{SnQFT_Speedup}. 

\subsection{The embedding operation}

Again, fix some level $n > 1$ on the chain of subgroups.
To turn the diagram of $\C[S_{n-1}]$-modules in \cref{fig: comm diagram for Sn-1 modules} into a diagram for $\C[S_n]$-modules, we use the following lemma.

\begin{lemma}\label{lemma: extended module isomorphism}
    For a group $G$ with subgroup $H \subset G$, let $V$ be a $\C[G]$-module and let $W$ be a $\C[H]$-module. 
    Let $\phi : \Res^G_H V \to W$ be a $\C[H]$-module isomorphism.
    Then, 
    \begin{enumerate}
        \item there exists a $\C[G]$-module structure on $W$;
        \item the $\C[G]$-module structure on $W$ extends its $\C[H]$-module structure;
        \item $\phi$ is a $\C[G]$-module isomorphism for $V$ and the $\C[G]$-module $W$.
    \end{enumerate}
\end{lemma}

\begin{definition}[Embedding operation]\label{def: embedding operation}
    For $\lambda, \widetilde{\lambda} \in \widehat{S}_{n-1}$, define the orthogonal matrix
    \begin{equation}
        [\mathrm{A}_{\lambda,\widetilde{\lambda}}]^{\mu}_{\kappa} \defeq 
        \begin{cases}
            \sqrt{\frac{d_\mu}{n d_\lambda}} &\text{if } \kappa = \star,\\
            \sqrt{\frac{d_\mu}{n d_\lambda}\frac{(n{-}1) d_\kappa}{d_{\widetilde{\lambda}}}}
            \frac{1}{\cont(\mu {\setminus} \lambda) - \cont(\lambda {\setminus} \kappa)} &\text{if } \lambda = \widetilde{\lambda} \text{ and } \kappa \neq \star,\\
            1 &\text{if } \lambda \neq \widetilde{\lambda}.
        \end{cases}
    \end{equation}
    Here, we have indices $\kappa$ labeling basis vectors of $M^-_{\lambda, \widetilde{\lambda}}$ and $\mu$ labeling basis vectors of $M^+_{\lambda, \widetilde{\lambda}}$.
    For the contents, note that $\mu$ defines an addable cell $\mu {\setminus} \lambda \in \AC(\lambda)$, and $\kappa$ defines a removable cell $\lambda {\setminus} \kappa \in \RC(\lambda)$ if $\kappa \neq \star$.
    Define the \emph{embedding operation} $\mathrm{A}$ for $S_{n-1} \subset S_n$ as
    \begin{equation}
        \mathrm{A} \defeq \bigoplus_{\widetilde{\lambda}} \mathrm{I}_{d_{\widetilde{\lambda}}} \otimes \mathrm{A}_{\lambda, \widetilde{\lambda}} 
    \end{equation}
\end{definition}

\begin{theorem}\label{thm: second induction}
    With $\mathrm{A}$ be as in \cref{def: embedding operation}, the diagram from \cref{fig: comm diagram for Sn-1 modules}, without restriction to $S_{n-1}$, commutes as a diagram of $\C[S_{n}]$-modules.
\end{theorem}
 
The proofs of \cref{lemma: extended module isomorphism} and \cref{thm: second induction} are given in \cref{app: Mackey proofs}.
As the $\C[S_n]$ modules extend the $\C[S_{n-1}]$-module structure by \cref{lemma: extended module isomorphism}, the embedding operation $\mathrm{A}$ should have the same structure as in \cref{lemma: commutative diagram for embedding operation}.
We find the entries for the blocks $\mathrm{A}_{\lambda, \widetilde{\lambda}}$ of $\mathrm{A}$ by requiring that the linear maps corresponding to both paths are equal, hence
\begin{equation}
    \IndMap[\lambda][n] = \left(\bigoplus_{\widetilde{\lambda} \in \widehat{S}_{n-1}} \mathrm{I}_{d_{\widetilde{\lambda}}} \otimes \mathrm{A}_{\lambda, \widetilde{\lambda}}\right) \cdot
    \left(\mathrm{I}_{d_{\lambda}} \oplus \bigoplus_{\kappa \in N^-(\lambda)}\IndMap[\kappa][n-1]\right) \cdot
    \left(\mathrm{M}_n \otimes \mathrm{I}_{d_\lambda}\right).
\end{equation}
As all matrices in this expression are unitary, and even orthogonal, we have found a decomposition of $\IndMap[\lambda][n]$ as a product of unitary matrices with corresponding to quantum gates.

To turn this theorem into a quantum algorithm, let us look at the three maps in the decomposition in order of application.
We already saw that the Mackey relabeling gate is identity in our chosen encoding. 
For the map $\mathrm{I}_{d_{\lambda}} \oplus \bigoplus_{\kappa \in \N^-(\lambda)}\IndMap[\kappa][n-1]$, we need to slightly modify our definition of $\IndMap$.
While \cref{def: IndMap} of $\IndMap$ corresponds to submatrix after the direct sum, the action on the part before the direct sum is not defined yet.
This corresponds to the transversal register being in the state $\ket{\star}$, encoded as $\ket{n}$, which is an ``invalid'' input, as $\IndMap[][n-1]$ expects state $\ket{i}$ with $i < n$ in the transversal register. 

\begin{definition}[Updated induction map for the QFT of $S_n$]\label{def: IndMap updated for special cases}
    For $\lambda \in \widehat{S}_{n-1}$, we extend the gate $\IndMap[\lambda][n]$ from \cref{def: IndMap} to have an extra \emph{flag register} $\ket{j}$ at the output, and to accept ``invalid'' transversal labels $i$ with $i > n$ at the input.
    \begin{itemize}
        \item If $i \leq n$, then the gate acts as in \cref{def: IndMap}, and sets $j = 1$, so that we can uncompute $j$ easily in the context of \cref{fig: theorem first induction full circuit relative paths};
        \item If $i > n$, which corresponds to an ``invalid'' input:
        \begin{itemize}
            \item Set output path $q' = p$, $q_n = \star$;
            \item Set $\mu = \lambda$. 
            Now, $\mu$ has only $n{-}1$ boxes instead of $n$ boxes;
            \item Set $j = i - n + 1$.
        \end{itemize}
    \end{itemize} 
\end{definition}

Lastly, we have defined the blocks $\mathrm{A}_{\lambda, \widetilde{\lambda}}$ to act on a space with a basis indexed by shared predecessors $\kappa \in \N^-(\lambda) \cap \N^-(\widetilde{\lambda})$ or the special symbol $\star$, and map it to a space with basis indexed by shared succesors $\mu \in \N^+(\lambda) \cap \N^+(\widetilde{\lambda})$.
However, in the circuit, we only have access to $\lambda$ and $\widetilde{\lambda}$, while the information about $\kappa$ should be extracted from Yamanouchi word entry $p_{n-1}$, and the information about $\mu$ should be stored implicitly as a Yamanouchi word entry $q_n$ at the output. 
In the circuit, we therefore define the embedding operation to act on Yamanouchi word entries, instead of entire irrep labels. 

\begin{proposition}\label{prop: second induction induction step circuit}
    The circuit in \cref{fig: second induction induction step} implements $\IndMap[\lambda][n]$ from \cref{def: IndMap updated for special cases}. 
\end{proposition}

\begin{figure}[H]
    \centering
    \includegraphics[]{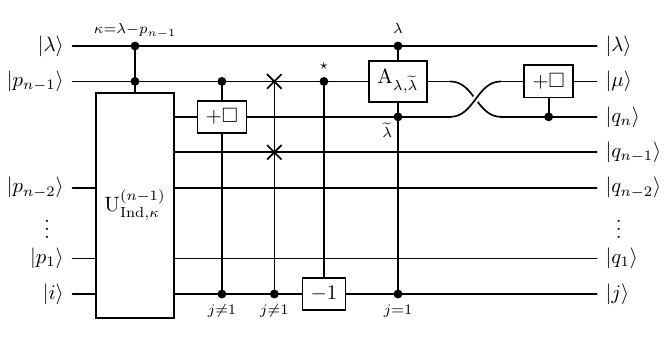}
    \caption{Recursive implentation of $\IndMap[\lambda][n]$ derived from Mackey theory. 
    The extra register $\ket{j}$ at the output is $\ket{1}$ when $i$ encodes an element of $\mathrm{T}_n$, and $\ket{i - n +1}$ otherwise, when $i > n$. }
    \label{fig: second induction induction step}
\end{figure}

\begin{proof}
    When $i \leq n$, this circuit corresponds to the alternative path in the commutative diagram of \cref{thm: second induction}, where we decompose $\IndMap[\lambda][n]$ into a product of unitaries which are Mackey relabeling, $\mathrm{I}_{d_{\lambda}} \oplus \bigoplus_{\kappa \in \N^-(\lambda)}\IndMap[\kappa][n-1]$ and the embedding operation $\mathrm{A}$.
    At the end of the circuit, the embedding operation does not act on the truncated length $n{-}1$ path, which corresponds to the identity in \cref{eq: block structure A}. 
    Also, note that if we encode transversal elements of the transversals defined in \cref{prop: Mackey transversals for Sn} by their indices, the Mackey relabeling map $\mathrm{M}_n$ has a trivial implementation.
    For the rest of the circuit, we separate the proof into cases depending on the value of $i$.
    
    \subsubsection*{Case 1: \texorpdfstring{$i = n$}{i is n}}
    Consider a basis vector $\ket{n}\ket{p}$ with control register $\ket{\lambda}$, where $n$ encodes $t^{(n)}_n = e$, and $p$ encodes a path to $\lambda$.
    Let $\kappa \in \N^-(\lambda)$ be the endpoint of $p'$.
    The Mackey relabeling from \cref{thm: second induction} should change the transversal label to $\star$, but as we encode $\star$ with $n$, we act trivially.
    
    The next map from \cref{thm: second induction} acts trivially, too.
    In our chosen encoding, we should map $\ket{\star}\ket{p} \mapsto \ket{1}\ket{q' = p}\ket{\star}\ket{\widetilde{\lambda} = \lambda}$, where $\ket{1}$ is the flag register.
    In the circuit, the gate $\IndMap[\kappa][n-1]$ sees $i = n$ as an ``invalid input''.
    Setting the flag register to $j = n - (n{-}1) + 1 = 2$, the state becomes $\ket{2}\ket{p'}\ket{\star}\ket{\kappa}\ket{p_{n-1}}$.
    After adding the box back to $\ket{\kappa}$ and swapping $\ket{\star}$ and $\ket{p_{n-1}}$, we get the state $\ket{2}\ket{p}\ket{\lambda}\ket{\star}$. 
    Now, we subtract $1$ from the flag register, which ensures that we have our desired $\ket{1}$ at the flag output. 
    The final gate of \cref{thm: second induction} is the embedding operation $\mathrm{A}$.
    Afterwards, we add box $q_n$ to $\widetilde{\lambda}$ to obtain $\mu$, so that we stay in our chosen encoding.

    \subsubsection*{Case 2: \texorpdfstring{$i < n$}{i lt n}}
    Now, consider a basis vector $\ket{i}\ket{p}$, with control register $\ket{\lambda}$, where $i < n$ encodes $t^{(n)}_i = (i, \dots, n)$, and $p$ encodes a path to $\lambda$.
    In this case, the Mackey relabeling should first map $t^{(n)}_{i}$ to $t^{(n-1)}_i$ according to \cref{prop: Mackey transversals for Sn}. 
    As these transversal elements are encoded using the same state $\ket{i}$, the Mackey step is again trivial as a quantum gate. 
    Then, we should act by $\bigoplus_{\kappa \in \N^-(\lambda)} \IndMap[\kappa][n-1]$, which is a submatrix of $\IndMap[\kappa][n-1]$.
    As $i < n$ encodes a ``valid'' input state of $\IndMap[][n-1]$, the flag register is set to $\ket{1}$.
    As the three controlled gates are now switched of, we end up in our desired state $\ket{1}(\IndMap[\kappa][n-1]\ket{p'})\ket{p_{n-1}}$, where $p_{n-1}$ uniquely identifies $\kappa \in \N^-(\lambda)$ using that $\lambda$ is stored in the control register. 
    Finally, the embedding operation and mapping $\ket{\widetilde{\lambda}}$ to $\ket{\mu}$ finishes the implementation of the second path in \cref{thm: second induction} as a quantum gate.

    \subsubsection*{Case 3: \texorpdfstring{$i > n$}{i gt n}}
    In this case, where inputs are ``invalid'', we do not have to consider \cref{thm: second induction}, but we only need to validate that our circuit acts as described in this proposition.
    Consider a basis state $\ket{i}\ket{p}$ at the input with $i > n$ and $p$ encoding an element of $\Path(\lambda)$.
    First, $\IndMap[\kappa][n-1]$ sets the flag register to $i - (n{-}1) + 1 = i - n + 2$, acts as identity on the path and outputs $\kappa$ in the irrep label register, so that we obtain the state $\ket{i - n + 2}\ket{p'}\ket{\star}\ket{\kappa}\ket{p_{n-1}}$. 
    After the three controlled gates which are all switched on, we obtain the state $\ket{i - n + 1}\ket{p}\ket{\lambda}\ket{\star}$.
    As the embedding operation is switched off, and $+\square$ acts trivially when the control register is $\star$, our final state is $\ket{i - n + 1}\ket{p}\ket{\star}\ket{\lambda}$, which matches the definition of $\IndMap[][n]$.
\end{proof}

\begin{theorem}\label{thm: second induction circuit}
    For $n \geq 1$, the quantum circuit in \cref{fig:2nd induction cycle version} implements $\IndMap[n]$.
\end{theorem}
\begin{figure}[H]
    \centering
    \includegraphics[]{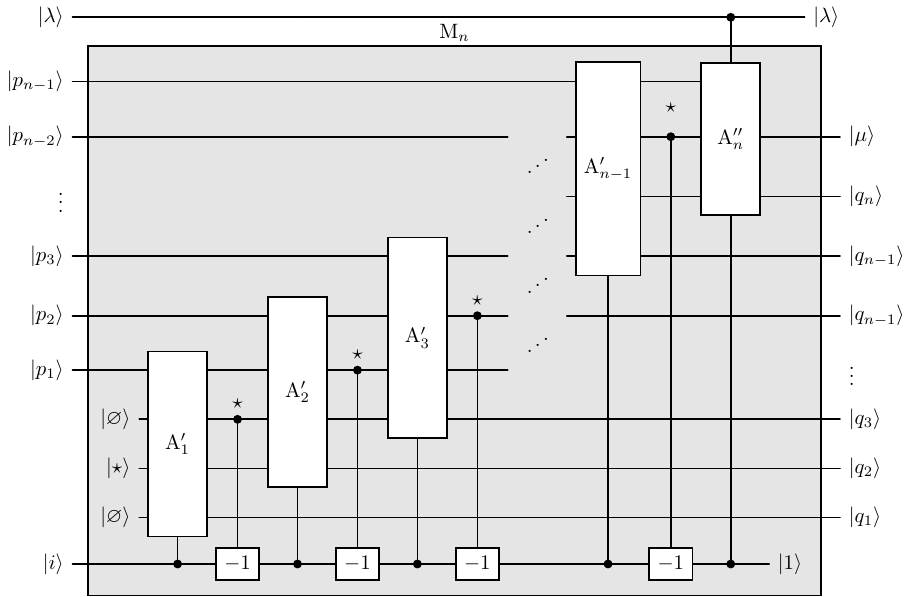}
    \caption{Implementation of $\IndMap[][n]$ derived using Mackey theory. 
    This circuit is a correction and simplification of the circuit by \cite{SnQFT_Speedup}.
    The gates $\mathrm{A}_m'$ and $\mathrm{A}_m''$ are the embedding operations for $S_{m-1} \subset S_m$ wrapped in a small classical circuit, as defined in \cref{fig: A matrix wrapper definitions for second induction}.}
    \label{fig:2nd induction cycle version}
\end{figure}

\begin{figure}[h]
    \centering
    \begin{subfigure}{0.57\textwidth}
        \centering
        \includegraphics[width=\linewidth]{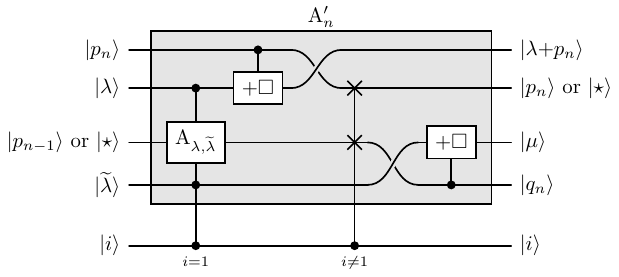}
        \caption{Definition of the blocks used for all steps in the Mackey construction of $\IndMap[][n]$, except for the last step. }
    \end{subfigure}
    \hfill
    \begin{subfigure}{0.42\textwidth}
        \centering
        \includegraphics[width=\linewidth]{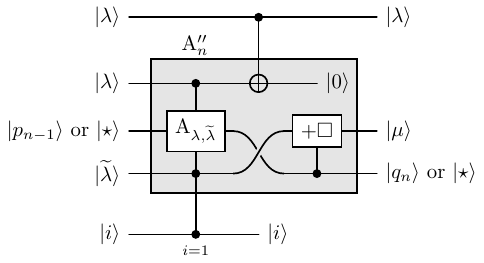}
        \caption{Definition of the last block in the Mackey construction of $\IndMap[][n]$.}
    \end{subfigure}
    \caption{Definitions of the classical wrapper circuits for the embedding operation $\mathrm{A}_{\lambda, \widetilde{\lambda}}$ used in \cref{fig:2nd induction cycle version}.}
    \label{fig: A matrix wrapper definitions for second induction}
\end{figure}

We prove this theorem by induction on $n$, where $\cref{prop: second induction induction step circuit}$ forms the key for proving the induction step.
The full proof can be found in \cref{app: proof second induction circuit}.

\subsection{Recovering Kawano--Sekigawa}

\begin{proposition}\label{prop: induction map with different transversals}
    Let $\T'_n = \{\tau_j \mid j \in [n]\}$ be any transversal for $S_{n-1} \subseteq S_n$.
    Then, there is always a unique re-indexing $f : [n] \to [n]$ such that $\tau_{f(i)} S_{n-1} = t_i S_{n-1}$, where we have $t_{i}^{(n)}$ from \cref{prop: Mackey transversals for Sn}, and unique subgroup elements $h^{(n)}_i \in S_{n-1}$ such that for all $i \in [n]$, it holds that
    \begin{equation}\label{eq: other choice of transversals}
        \tau_{f(i)}^{(n)} = t_{i}^{(n)} h^{(n)}_i.
    \end{equation}
    Denote by $\IndMapDash[\lambda][n]$ the induction map using the new transversal $\T_n'$, and by $\IndMap[\lambda][n]$ the induction map using the transversals from \cref{prop: Mackey transversals for Sn} which can by implemented as in \cref{thm: second induction circuit}.
    Then, we can implement $\IndMapDash[\lambda][n]$ by first applying $f^{-1}$ on the transversal register coherently, and then apply irreducible representations of $h^{(n)}_{i}$ in the Path register, followed by applying the circuit for $\IndMap[\lambda][n]$ as in \cref{thm: second induction circuit}.
\end{proposition}

The proof is given in \cref{app: proof induction map with different transversals}. 
As it is possible to write any subgroup element---such as $h^{n}_i$---as a product of transversal elements $t^{(n-1)}_{i_{n-1}} \dots t^{(1)}_{i_1}$, and each $t^{(m)}_{i}$ can be decomposed as $t^{(m)}_{i} = \sigma_{i} \dots \sigma_{m-1}$, we can use $\mathrm{O}(n^2)$ gates $\mathrm{L}(\sigma_{k-1})$ from \cref{subsec: Beals for Sn} to apply irreps of the correcting subgroup elements to the path register.
However, for an arbitrary choice of transversals, compiling the control conditions for these gates might be hard.  
Still, for most applications, it can be assumed that there is structure in the choice of transversals giving a closed form expression for the control conditions.
Then, we use that $\mathrm{L}(\sigma_{k-1})$ can be implemented using $\widetilde{\mathcal{O}}(1)$ elementary gates, so that the corrections can be implemented with $\widetilde{\mathcal{O}}(n^2)$, contributing $\widetilde{\mathcal{O}}(n^3)$ extra gates for the entire QFT.
We will later see that this does not dominate the complexity. 
Alternatively, the correction could also be applied classically outside of the QFT, with similar implications for the gatecount.

In \cite{SnQFT_Speedup}, the transversal for $S_{n-1} \subset S_n$ is defined as 
\begin{equation}
    \T^{(KS)}_n \defeq \{\tau^{(n)}_i = (1, \dots, n)^{i} \mid i \in [n]\}.
\end{equation}
We can show that $\tau^{(n)}_{i}S_{n-1} = t^{(n)}_i S_{n-1}$, which implies that our re-indexing bijection $f$ is identity, and that we can simply ignore it. 
We find the subgroup element for the corrections as 
\begin{equation}
    h^{(n)}_i = (t^{(n)}_{i})^{-1} \tau^{(n)}_{f(i)} = (t^{(n)}_{i})^{-1} \tau^{(n)}_{i} = (1, \dots, n{-}1)^{i - 1},
\end{equation}
where $h^{(n)}_1 = h^{(n)}_n = e$.

Kawano and Sekigawa include the gates which we intepret as corrections in their second induction relation as a gate named $\mathrm{K}_n$ and subtraction of $1$ modulo $n$ at each layer, where $n =2$ acts as the base case. 
$K_n$ applies irreps of subgroup elements on the path register. 
The figure showing the implementation of $K_n$ seems to apply irreps of the inverse group elements compared to the definition of $K_n$ in the text. 
If we replace the subtraction by a subtraction of $1$ mod $n{-}1$ instead of mod $n$, working out the full recursion seems to give a circuit for $\tau^{(n)}_i = (1, \dots, n)^{-i}$ so that $f(i) = -i \bmod n$, which corresponds to \cref{prop: induction map with different transversals}.
Still, this does not imply the need for addition of $1$ mod $n$ after using $\IndMap[][n-1]$ inductively, and their classical gate $P_n$.

Another interpretation may use that $h^{(n)}_i = \tau^{(n-1)}_{i=1}$ for $i \leq n{-}1$, leading to a slightly different second induction relation.
This might be more in line with their embedding of $\IndMap[][n-1]$ into the larger space when it is used recursively. 
This extension of the space is what we interpret as the identity channel in the Mackey decomposition, which we marked with the special symbol $\star$. 
While we solved the embedding by defining the action on ``invalid'' inputs and adding a flag register at the output, their gate $P_n$ seems to achieve the same. 
Most importantly, the definition of the entries and the combinatorics for labeling the input and output of both our embedding operation and the one of \cite{SnQFT_Speedup} seem to perfectly match.

It is hard to reproduce the exact embedding of $\IndMap[][n-1]$ in the construction of $\IndMap[][n]$ by Kawano and Sekigawa and the implementation of $P_n$.
One of the reasons is that they do not explicitly mention the possibility of detours at the output of $\IndMap[][n-1]$, but somehow assume that all paths at the output of this gate go to the original endpoint $\lambda$, instead of any endpoint $\widetilde{\lambda}$ sharing predecessors $\kappa$ with $\lambda$. 
Our solution with the flag register gets rid of the gate $P_n$ entirely, and is a much simpler implementation of the same induction relation as a quantum algorithm. 

\subsection{Implementation of \texorpdfstring{$\mathrm{A}_{\lambda, \widetilde{\lambda}}$}{A}}

Lastly, \cite{SnQFT_Speedup} claim that $\mathrm{A}_{\lambda, \widetilde{\lambda}}$ can be implemented using $\mathcal{O}(n)$ elementary operations.
Their main argument uses that the maximal dimension of the blocks $\mathrm{A}_{\lambda, \widetilde{\lambda}}$ scales as $\mathcal{O}(\sqrt{n})$, and can therefore be implemented using a matrix decomposition method, requiring $\mathcal{O}(n)$ rotations. 
However, computing the $\mathcal{O}(n)$ rotation angles for such a decomposition needs to be carried out coherently, as classical pre-computation for all $\lambda \in \widehat{S}_{n-1}$ would require analyzing $|\widehat{S}_{n-1}| = \mathcal{O}(e^{\sqrt{n}})$ different matrices, which is not efficient.
\cite{SnQFT_Speedup} do not account for the related computational overhead in their complexity analysis. 

\begin{wrapfigure}{r}{4.73cm}
    \begin{center}
        \includegraphics[width = 4.53cm]{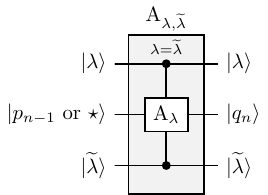}
        \caption{Implementation of $\mathrm{A}_{\lambda, \widetilde{\lambda}}$ using $\mathrm{A}_{\lambda}$.}
        \label{fig: A detour handling circuit}
    \end{center}
\end{wrapfigure}

The upper bound to the dimension of the blocks is given more precisely by \cref{lemma: number of ACs}.
First, let us take care of the detours. 
By \cref{thm: second induction}, we see that $\mathrm{A}_{\lambda, \widetilde{\lambda}}$ is one-dimensional and acts as a bijection if $\lambda \neq \widetilde{\lambda}$.
Moreover, \cref{lemma: diamond} shows that the cell related to Yamanouchi word entry $p_{n-1}$ is exactly the same as $q_n$, so $\mathrm{A}_{\lambda, \widetilde{\lambda}}$ should act trivially.
Denote by $\mathrm{A}_{\lambda} \defeq \mathrm{A}_{\lambda, \lambda}$.
Then, the circuit in \cref{fig: A detour handling circuit} implements $\mathrm{A}_{\lambda, \widetilde{\lambda}}$ using classical gates and $\mathrm{A}_{\lambda}$.

To implement $\mathrm{A}_{\lambda}$, we will run a coherent matrix decomposition algorithm.
An extensive analysis and low-level implementation can be found in \cref{app: low level embedding operation}. 
In summary, we compute the blocks $\mathrm{A}_{\lambda}$ coherently, and decompose it into a product of $\mathcal{O}(n)$ Givens rotations.
After computing each angle, we update the two corresponding columns of the register storing the matrix.
This matrix update step requires $\mathcal{O}(\sqrt{n})$ coherent arithmetic operations for updating all the entries, leading to a total size of $\mathcal{O}(n^{1.5})$ elementary gates for implementing the embedding operation. 
While the depth can be reduced to $\mathcal{O}(n)$ by updating matrix entries simultaneously, the coherent matrix decomposition forms the bottleneck for the size of the $S_n$-QFT, which is not included in the analysis of \cite{SnQFT_Speedup}.

\subsection{Complexity}
\begin{lemma}\label{lemma: IndMap Mackey complexity}
    For $n \geq 1$, the circuit in \cref{fig:2nd induction cycle version} implementing $\IndMap[][n]$ uses $\widetilde{\mathcal{O}}(n^{1.5})$ qubits, $\widetilde{\mathcal{O}}(n^{2.5})$ elementary gates and has depth $\widetilde{\mathcal{O}}(n^{2})$.
\end{lemma}

\begin{proof}
    First, consider the space complexity. 
    We have at most two shape registers $\ket{\lambda}$ and $\ket{\mu}$ at a time, which together need $\widetilde{\mathcal{O}}(n)$ qubits. 
    For the Yamanouchi-word register, we need $\widetilde{\mathcal{O}}(n)$ qubits as well. 
    In both $A_m'$ and $A_{m}''$, with $1 \leq m \leq n$, the space complexity is dominated by $\mathrm{A}_{\lambda, \widetilde{\lambda}}$, which need at most $\widetilde{\mathcal{O}}(n^{3/2})$ qubits.
    For all other gates, the space complexity is constant.
    Hence, the total space complexity is $\widetilde{\mathcal{O}}(n^{3/2})$.

    Second, for the size and depth, note that both $A_m'$ and $A_{m}''$, with $1 \leq m \leq n$, contain a single gate $\mathrm{A}_{\lambda, \widetilde{\lambda}}$, with size $\widetilde{\mathcal{O}}(m^{3/2})$ and depth $\widetilde{\mathcal{O}}(m)$.
    The gate $+\square$ has size and depth $\widetilde{\mathcal{O}}(m)$.
    The big swap gate in $\mathrm{A}_m'$ can be implemented using $\lceil\log_2(m)\rceil$ parallel swap gates, while uncomputing the shape register in $\mathrm{A}_m''$ uses $\widetilde{\mathcal{O}}(m)$ parallel CNOT-gates with depth $\mathcal{O}(1)$. 
    Hence, $\mathrm{A}_{\lambda, \widetilde{\lambda}}$ again dominates all complexities in both $A_m'$ and $A_{m}''$, giving size $\widetilde{\mathcal{O}}(m^{3/2})$ and depth $\widetilde{\mathcal{O}}(m)$.
    As the circuit uses $A_m'$ for $m = 1, \dots, n{-1}$ and $A_n''$, the total size becomes $\sum_{m = 1}^{n} \widetilde{\mathcal{O}}(n^{3/2}) = \widetilde{\mathcal{O}}(n^{5/2})$ and the depth becomes $\sum_{m = 1}^{n} \widetilde{\mathcal{O}}(n) = \widetilde{\mathcal{O}}(n^{2})$.
\end{proof}

\section{Complexity}\label{sec: complexity}
In this chapter, we will prove the main results of this work, \cref{thm: the big complexity theorem Beals,thm: the big complexity theorem KS}.

\thmBealsResult*

\begin{proof}
    Implementing the $S_n$-QFT using the circuit of \cref{fig: theorem first induction full circuit relative paths}, we alternatingly use $\mathrm{T}_m$ and $\IndMap[m]$ once for all $m = 1, \dots, n$.
    \cite{SnQFT_Speedup} propose an implemetation for $\mathrm{T}_m$ with $\widetilde{\mathcal{O}}(m)$ gates and depth and $\widetilde{\mathcal{O}}(1)$ ancillary registers. 
    The gate $\IndMap[][m]$---using the circuit in \cref{fig: simplified Beals circuit for Sn}, which is a simplified version of the algorithm proposed by \cite{Beals}---has a gate count of $\widetilde{\mathcal{O}}(m^2)$, depth of $\widetilde{\mathcal{O}}(m^2)$ and uses $\widetilde{\mathcal{O}}(m)$ ancillary registers by \cref{lemma: IndMap Beals complexity}.
    We find the total depth of the QFT as 
    \begin{equation}
        \sum_{m=1}^{n} \underbrace{\mathcal{O}(m)}_{\mathcal{T}_m} + \underbrace{\mathcal{O}(m^2)}_{\IndMap[][m]} = \sum_{m=1}^{n} \mathcal{O}(m^2) = \mathcal{O}(n^3).
    \end{equation}
    We find the gate count as 
    \begin{equation}
        \sum_{m=1}^{n} \underbrace{\mathcal{O}(m)}_{\mathcal{T}_m} + \underbrace{\mathcal{O}(m^2)}_{\IndMap[][m]} = \sum_{m=1}^{n} \mathcal{O}(m^2) = \mathcal{O}(n^3).
    \end{equation}
    Inside $\IndMap[][m]$, the gate $\mathrm{U}_{\uparrow}$ uses the most ancillary qubits.
    
    This concludes our proof.
\end{proof}
The scaling with target diamond error $\epsilon$ scales polylogarithmically with $\epsilon^{-1}$ and is therefore hidden in the notation $\widetilde{\mathcal{O}}$. 
To motivate this, note that our only approximations come from $\Rep(\sigma_k)$ and $\mathrm{U}_\uparrow$, of which the first one uses $1$ controlled rotation and the second one uses $\mathcal{O}(\sqrt{m})$ controlled rotations.
More precisely, these controlled rotations are $\mathrm{R}_y(\theta)$ gates with the rotation angle $\theta$ stored in a quantum float register. 
When $\theta$ has a maximal numerical error of $\delta$, the corresponding rotation will have a diamond error of at most $\delta$, too.
In total, the $\IndMap[][m]$ circuit uses $\mathcal{O}(m^{1.5})$ of these rotations, making the whole QFT require $\mathcal{O}(n^{2.5})$ rotations, upper bounding the total diamond error by $\epsilon= \mathcal{O}(n^{2.5})\delta$, where $\delta$ is the maximal numerical error of the rotation angles.
As the numerical error of the coherent arithmetic operations goes down exponentially with the number of extra qubits $k$, we say that $\delta \propto 2^{-k}$.
Hence, for a given maximal diamond error $\epsilon$, we require $k = \mathcal{O}\left(\log_2\left(\frac{n^{2.5}}{\epsilon}\right)\right)$ extra qubits.
As the size and depth of coherent arithmetic operations scale polynomially with the number of qubits, the polylogarithmic dependence holds for the depth and gatecount as well.

\thmKSResult*

\begin{proof}
    Just like \cref{thm: the big complexity theorem Beals}, we implement the QFT using \cref{fig: theorem first induction full circuit relative paths}, but now with the implementation of $\IndMap[][m]$ from \cref{thm: second induction}, with the circuit in \cref{fig:2nd induction cycle version}. 
    Then, $\IndMap[m]$ has a gate count of $\widetilde{\mathcal{O}}(m^{2.5})$, depth of $\widetilde{\mathcal{O}}(m^2)$ and uses $\widetilde{\mathcal{O}}(m^{1.5})$ ancillary registers by \cref{lemma: IndMap Beals complexity}.
    We find the total depth of the QFT as 
    \begin{equation}
        \sum_{m=1}^{n} \underbrace{\mathcal{O}(m)}_{\mathcal{T}_m} + \underbrace{\mathcal{O}(m^2)}_{\IndMap[][m]} = \sum_{m=1}^{n} \mathcal{O}(m^2) = \mathcal{O}(n^3).
    \end{equation}
    We find the gate count as 
    \begin{equation}
        \sum_{m=1}^{n} \underbrace{\mathcal{O}(m)}_{\mathcal{T}_m} + \underbrace{\mathcal{O}(m^{2.5})}_{\IndMap[][m]} = \sum_{m=1}^{n} \mathcal{O}(m^2) = \mathcal{O}(n^{3.5}).
    \end{equation}
    The group elements at the input and the Yamanouchi words at the output can be encoded with $\widetilde{\mathcal{O}}(n)$ qubits, by \cref{lemma: Yamanouchi word optimal scaling}. 
    The irrep label register $\ket{\mu}$ at the output requires $\mathcal{O}(n)$ qubits by \cref{lemma: shape register qubit count}. 
    In the circuit of \cref{fig: theorem first induction full circuit relative paths}, we need at most two irrep label registers at a time, so the space complexity is dominated by the embedding operation, requiring $\widetilde{\mathcal{O}}(n^{1.5})$ ancillary qubits.
    This concludes our proof.
\end{proof}
Now, the only approximations come from the implementation of the embedding operations, which we decompose as a sequence of $\mathcal{O}(m)$ controlled $\Rep_y(\theta)$ rotations at level $m$ in \cref{fig:2nd induction cycle version}.
In this implementation, we can show that the QFT uses a total of $\mathcal{O}(n^3)$ of such rotations. 
As \cref{alg: angles}, based on Algorithm 5.1.3 of Golub and Van Loan~\cite{MatrixComputations}, computes the Givens rotations in a numerically stable manner, the required arithmetic precision again grows only logarithmically with the inverse target error.

\section*{Acknowledgements}

This paper originated from the master's thesis of \cite{masterthesis}.
CB would like to thank Peter Bruin for supervision during the master thesis project. 
CB was supported by a QDNL CAT-1 Phase 3 grant.
DG and MO were supported by National Growth Fund grant (NGF.1623.23.025) ``Qudits in theory and experiment''.

\newpage
\appendix

\section{Low level circuits}
This section will elaborate on the implementation and complexity of some of the gates used in our algorithm

\subsection{Coherent Yamanouchi word encoding} 
\begin{lemma}
    For a subgroup chain level $n \geq 1$, the implementation of the gate $T_n$ from \cref{fig: theorem first induction full circuit relative paths} as given by \cite{SnQFT_Speedup} has depth $\widetilde{\mathcal{O}}(n)$, uses $\widetilde{\mathcal{O}}(n)$ gates and $\widetilde{\mathcal{O}}(1)$ ancillary qubits.
\end{lemma}

\begin{proof}
    The gate $q {\setminus} \widetilde{r}$ uses $n-1$ coherent additions and $n-2$ coherent subtractions on an $n$-dimensional register, hence requiring $\widetilde{\mathcal{O}}(n)$ gates and the same depth.
    The gate $+\square$ uses $n$ coherent additions of $1$, controlled by the $n$-dimensional register. 
    The 
\end{proof}

\subsection{Beals' implementation}\label{app: low level circuits Beals}

In this section, we discuss the implementation of $\mathrm{U}_{\uparrow}$, $\mathrm{V}_{k}$ and $\Rep(\sigma_{k-1})$ from \cref{def: U gate Beals,def: V gate beals,def: irrep gate Beals} used for implementing $\IndMap[][]$ with the circuit from \cref{fig: simplified Beals circuit for Sn}.

\subsubsection{Creating a superposition of successors}

We will now discuss the implementation of $\mathrm{U}_{\uparrow}$ at level $n$, where we turn a fresh level-$n$ qudit $\ket{0}$ into a register with a superposition of successor labels for $\lambda \in \widehat{S}_{n-1}$, and a shape register $\ket{\lambda}$ into a superposition of successors as 
\begin{equation*}
    \ket{0}\ket{\lambda}\ket{\lambda} \mapsto \sum_{a \in \AC(\lambda)} \sqrt{\frac{d_{\lambda + a}}{n d_{\lambda}}} \ket{a}\ket{\lambda + a}\ket{\lambda}, 
\end{equation*}
where we encode $\ket{a}$ by its Yamanouchi word entry. 
We will first compute this Yamanouchi word entry, after which we can use the gate $+\square$ from \cref{subsec: Fourier transform with relative paths} to update the shape register. 
Note that the last register $\ket{\lambda}$ acts as a control.

For initializing the Yamanouchi word register, we apply the following steps:
\begin{enumerate}
    \item Conditioned on control register $\ket{\lambda}$, compute the row numbers and contents of all addable and removable cells of $\lambda$, and store it in dedicated quantum registers.
    This is discussed further in \cref{subsec: rows and contents of RC AC}.
    \item For all addable cells $\mathrm{a} \in \AC(\lambda)$, compute $\sqrt{\frac{d_{\lambda + a}}{n d_{\lambda}}}$.
    This ratio is given by \cref{lemma: twiddle factor}.
    \item The isometry $\ket{0} \mapsto \sum_{a \in \AC(\lambda)} \sqrt{\frac{d_{\lambda + a}}{n d_{\lambda}}} \ket{a}$ can be written as a product of $\mathcal{O}(\sqrt{n})$ Givens rotations. 
    Compute the rotation angles coherently using the one-column specialization of \cref{alg: angles}.
    \item Apply the resulting Givens rotations in order to the fresh Yamanouchi-word register $\ket{0}$.
\end{enumerate}

For these steps, we find asymptotic complexities as summarized in \cref{tab: U uparrow complexities}

\begin{table}[H]
\centering
\begin{tabular}{l|ccc}
\textbf{Step} & \textbf{Space} & \textbf{Gatecount} & \textbf{Depth} \\ \hline
1. Rows, contents    &        $\widetilde{\mathcal{O}}(n)$                          &    $\widetilde{\mathcal{O}}(n)$   & $\widetilde{\mathcal{O}}(1)$              \\
2. Dimension ratios &     $\widetilde{\mathcal{O}}(\sqrt{n})$&                            $\widetilde{\mathcal{O}}(n)$ &     $\widetilde{\mathcal{O}}(\sqrt{n})$  \\
3. Coherent Givens decomposition       &    $\widetilde{\mathcal{O}}(n)$   & $\widetilde{\mathcal{O}}(n)$   &               $\widetilde{\mathcal{O}}(\sqrt{n})$ \\
4. Applying Givens rotations   &    -   & $\widetilde{\mathcal{O}}(\sqrt{n})$   &             $\widetilde{\mathcal{O}}(\sqrt{n})$ \\
\end{tabular}
\caption{Complexities of the different steps used in the implementation of $U_\uparrow$ at level $n$.}
\label{tab: U uparrow complexities}
\end{table}

Hence, our implementation of $U_\uparrow$ at level $n$ has a qubit count and gate count of $\widetilde{\mathcal{O}}(n)$ and a depth of $\widetilde{\mathcal{O}}(\sqrt{n})$. 

\subsubsection{CNOT with intricate controls: \texorpdfstring{$V_k$}{Vk}}
Suppose that we use an ancillary qubit for flagging the special state $\ket{\star}$, so that $\ket{t_m}$ is encoded as $\ket{m}\ket{0}$ and $\ket{\star}$ as $\ket{0}\ket{1}$. 
If we flip the flag register when the transversal is $\ket{k}$, then apply transversal $\mathrm{X}$-gates uncomputing $\ket{k}$ to $\ket{0}$ controlled on the flag being $\ket{1}$, and then flip the flag register again when the transversal is $\ket{k}$, we implemented $V_k$ without control.
The valid input states are mapped as follows:
\begin{align}
    m \neq k : \ket{t_m}\ket{0} &\mapsto \ket{t_m}\ket{0} & \ket{t_k}\ket{0} &\mapsto \ket{0}\ket{1} & \ket{0}\ket{1} &\mapsto \ket{t_k}\ket{0}.
\end{align}

To turn this into $\mathrm{V}_k$, we only need to add a control on the two irrep label registers $\ket{\lambda}$, $\ket{\widetilde{\lambda}}$ being equal. 
As these registers use $\mathcal{O}(n)$ qubits by \cref{lemma: shape register qubit count}, such a control cal be implemented with $\mathcal{O}$ elementary gates, $\mathcal{O}(\log(n))$ depth and $\mathcal{O}(n)$ ancillas.

\subsubsection{Applying irreps of generators on a path register}

For the gate $\Rep(\sigma_k)$, we can decompose group elements of transversals into a product of adjacent transposition, for which the irreps are very sparse in Young's orthogonal form. 
Then, we only need to implement $R(\sigma_i)$ for $i \in [n-1]$. 
Fix $i \in [n-1]$, $\lambda \in \widehat{S}_{n-1}$, and let $P \in \Path(\lambda)$.
Consider the corresponding SYT of $P$, and denote by $\sigma_i P $ the diagram obtained by swapping boxes with entries $i$ and $i+1$.
While this does not necessarily have to be a valid SYT, we can still encode it in our quantum registers. 

\begin{lemma}
    The circuit in \cref{fig: circuit for irreps of generators on path register} implements $R(\sigma_i)$ for $i \in [n]$ using $\widetilde{\mathcal{O}}(1)$ elementary operations. 
    When $i = 1$, we can replace the entire circuit by a Pauli-$\mathrm{Z}$ on $\ket{u_2}$.
\end{lemma}

\begin{figure}[H]
    \centering
    \includegraphics[width=\textwidth]{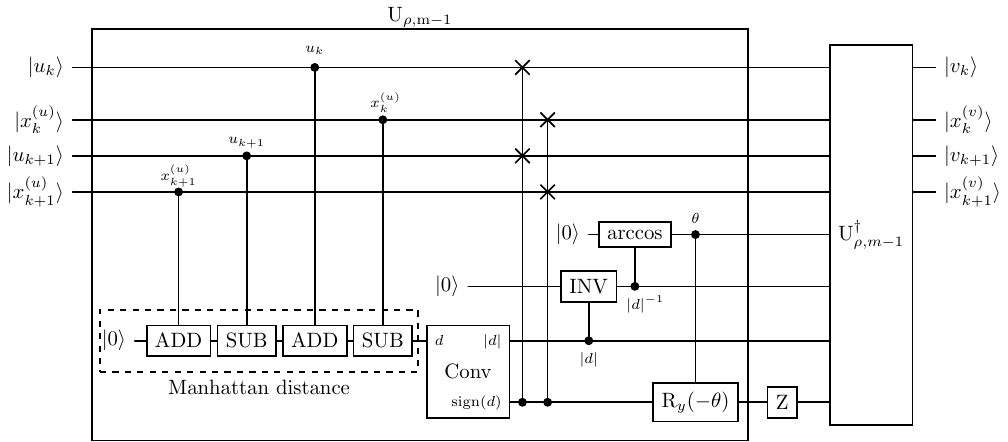}
    \caption{Circuit for applying irreps of generators on path register from the thesis}
    \label{fig: circuit for irreps of generators on path register}
\end{figure}

\begin{proof}
    Consider an incoming path $\ket{U} = \ket{c_1}\ket{c_2}\dots\ket{c_n}$, where $\ket{c_i} = \ket{u_m}\ket{x^{u}_m}$, so that $u_m$ is the Yamanouchi word entry and $x^{u}_m$ the column index of the corresponding cell with entry $m$ in the SYT.
    Then, Young's orthogonal shows that $R(\sigma_i)$ must act on $\ket{U}$ as
    \begin{equation}
        R(\sigma_i) \ket{U} = d_i^{-1}\ket{U} + \sqrt{1- d_i^{-2}}\ket{\sigma_i U},
    \end{equation}
    where in our encoding, we have that $\ket{\sigma_i U} = \ket{c_1}\dots\ket{c_{i-1}}\ket{c_{i+1}}\ket{c_i}\ket{c_{i+2}} \dots \ket{c_n}$. 
    Hence, we obtain $\ket{\sigma_i U}$ by swapping the $i$-th and $i+1$-th register of $\ket{U}$.

    In the algorithm, we distinguish the following steps:
    \begin{enumerate}
        \item Compute the Manhattan distance using \cref{eq: Manhattan distance with contents};
        \item Convert signed integer $d_i$ from two's complement encoding to an absolute value with sign bit;
        \item Compute $\arccos(|d_i|^{-1})$ using coherent arithmetic;
        \item Changing $\ket{U}$ to $\ket{\sigma_i U}$ if $d_i < 0$, using the sign bit as a control bit;
        \item Applying a continuous $Y$-rotation on the sign bit, with angle $-2\arccos(|d_i|^{-1})$ controlled by the corresponding quantum register, followed by a Pauli-$\mathrm{Z}$;
        \item Uncompute steps 4 to 1.
    \end{enumerate}

    We will now look at how the algorithm transforms an incoming state $\ket{U}$.
    Step 1 to 3 act as $\ket{U} \mapsto \ket{d_i}\ket{U} \mapsto \ket{\sign(d_i)}\ket{|d_i|}\ket{U} \mapsto \ket{\sign(d_i)}\ket{\arccos{|d_i|^{-1}}}\ket{garbage}\ket{U}$.

    For analyzing steps 4, 5 and the entire uncomputation, we distinguish two cases. 
    First, we consider $d_i \geq 0$, and second, we consider $d_i < 0$.
    When $d_i \geq 0$, we have $\ket{\sign(d_i)}\ket{|d_i|} = \ket{0}\ket{d_i}$.
    We find
    \begin{align}
        \ket{0}\dots \ket{U} &\mapsto \ket{0}\dots \ket{U}  &&\text{controlled swaps}\\
        &\mapsto |d_i|^{-1}(\ket{0}\dots \ket{U})  - \sqrt{1- |d_i|^{-2}}(\ket{1}\dots \ket{U}) &&\mathrm{R}_y(-2\arccos(|d_i|^{-1})\\
        &\mapsto |d_i|^{-1}(\ket{0}\dots \ket{U})  + \sqrt{1- |d_i|^{-2}}(\ket{1}\dots \ket{U}) &&\text{Pauli-}\mathrm{Z}\\
        &\mapsto |d_i|^{-1}(\ket{0}\dots \ket{U})  + \sqrt{1- |d_i|^{-2}}(\ket{1}\dots \ket{\sigma_i U}) &&\text{controlled swaps}\\
        &\mapsto |d_i|^{-1}(\ket{0}\ket{d_i}\ket{U})  + \sqrt{1- |d_i|^{-2}}(\ket{1}\ket{d_i}\ket{\sigma_i U})&&\text{uncomputing step 3}\\
        &\mapsto |d_i|^{-1}(\ket{d_i}\ket{U})  + \sqrt{1- |d_i|^{-2}}(\ket{-d_i}\ket{\sigma_i U})&&\text{uncomputing step 2}\\
        &= d_i^{-1}(\ket{d_i}\ket{U})  + \sqrt{1- d_i^{-2}}(\ket{-d_i}\ket{\sigma_i U}) \\
        &\mapsto d_i^{-1}\ket{U}  + \sqrt{1- d_i^{-2}}\ket{\sigma_i U}.&&\text{uncomputing step 1}\\
    \end{align}
    In the last step, we used that the Manhattan distance from box $i$ to box $i{+}1$ differs from $U$ and $\sigma_i U$ by a sign. 

    When $d_i < 0$, we have $\ket{\sign(d_i)}\ket{|d_i|} = \ket{0}\ket{-d_i}$ and find
    \begin{align}
        \ket{1}\dots \ket{U} &\mapsto \ket{1}\dots \ket{\sigma_i U}  &&\text{controlled swaps}\\
        &\mapsto |d_i|^{-1}(\ket{1}\dots \ket{\sigma_i U}) + \sqrt{1- |d_i|^{-2}}(\ket{0}\dots \ket{\sigma_i U}) &&\mathrm{R}_y(-2\arccos(|d_i|^{-1})\\
        &\mapsto -|d_i|^{-1}(\ket{1}\dots \ket{\sigma_i U})  + \sqrt{1- |d_i|^{-2}}(\ket{0}\dots \ket{\sigma_i U}) &&\text{Pauli-}\mathrm{Z}\\
        &\mapsto -|d_i|^{-1}(\ket{1}\dots \ket{U})  + \sqrt{1- |d_i|^{-2}}(\ket{0}\dots \ket{\sigma_i U}) &&\text{controlled swaps}\\
        &\mapsto -|d_i|^{-1}(\ket{1}\ket{|d_i|}\ket{U})  + \sqrt{1- |d_i|^{-2}}(\ket{0}\ket{|d_i|}\ket{\sigma_i U})&&\text{uncomputing step 3}\\
        &\mapsto -|d_i|^{-1}(\ket{-|d_i|}\ket{U})  + \sqrt{1- |d_i|^{-2}}(\ket{|d_i|}\ket{\sigma_i U})&&\text{uncomputing step 2}\\
        & = d_i^{-1}(\ket{d_i}\ket{U})  + \sqrt{1- d_i^{-2}}(\ket{-d_i}\ket{\sigma_i U})\\
        &\mapsto d_i^{-1}\ket{U}  + \sqrt{1- d_i^{-2}}\ket{\sigma_i U}.&&\text{uncomputing step 1}\\
    \end{align}

    When $i = 1$, note that the box with entry $2$ is always right next to box $1$, or right below is. 
    In the first case, we have $d_1 = 1$, and in the second case, we have $d_1 = -1$.
    Working out Young's orthogonal form, we find that $R(\sigma_1)$ must act trivially when $\ket{u_2} = \ket{1}$, and as phase flip when $\ket{u_2} = \ket{2}$.
    This can be implemented with a Pauli-$\mathrm{Z}$ gate with some classical gates depending on our chosen encoding into qubits.
\end{proof}

\subsection{Low-level description of the embedding operation}\label{app: low level embedding operation}

The following lemma gives an asymptotic scaling of the maximal dimension of $\mathrm{A}_\lambda$.

\begin{restatable}{lemma}{addableCellsBoundLemma}\label{lemma: number of ACs}
    For all $n \in \mathbb{N}, n \geq 3$, and all $\lambda \in \widehat{S}_{n-1}$, we have 
    \begin{equation}
        |\AC(\lambda)| \leq \Big\lfloor \frac{1}{2}\left(1+ \sqrt{8n-7}\right)\Big\rfloor.
    \end{equation}
\end{restatable}

Hence, $\dim{\mathrm{A_{\lambda}}} \leq \Big\lfloor \frac{1}{2}\left(1+ \sqrt{8n-7}\right)\Big\rfloor$ for $\lambda \in \widehat{S}_{n}$.

Moreover, we know the following about the determinant of $\mathrm{A}_\lambda$
\begin{restatable}{lemma}{AdetLemma}\label{lemma: determinant of A dash m lambda}
    For $n \geq 2$ and $\lambda \in \widehat{S}_{n-1}$, we have
    \begin{equation}
    \det(\mathrm{A}_{\lambda}) = \begin{cases}
        1 &\text{if } \dim(\mathrm{A}_{\lambda}) \text{ is odd},\\
        -1 &\text{if } \dim(\mathrm{A}_{\lambda}) \text{ is even}.
    \end{cases}
    \end{equation}
\end{restatable}

This gives us an easy way to write $\mathrm{A}_{\lambda}$ as the product of a possible reflection and a special orthogonal matrix.

As we found no extra structure in $\mathrm{A}_\lambda$ allowing for an easy decomposition into elementary operations, we propose to decompose its associated special orthogonal matrix into Givens rotations coherently.
Let $O \in SO(d)$ act on a space with basis $\{\ket{k}\}_{k \in [d]}$. 
Then, we can decompose $O$ as 
\begin{equation*}
    O =G_{d{-}1, d}(\theta_{d{-}1,d})\dots G_{2,d}(\theta_{2,d})\dots G_{2,3}(\theta_{2,3}) G_{1,d}(\theta_{1,d})\dots G_{1,3}(\theta_{1,3})G_{1,2}(\theta_{1,2}),
\end{equation*}
where Givens rotation $G_{i,j}(\theta)$, with $i <j$, acts only on the subspace spanned by $\ket{i}, \ket{j}$ as
\begin{equation*}
    G_{i,j}(\theta) = 
    \begin{pmatrix}
        \cos(\theta) & -\sin(\theta)\\
        \sin(\theta) & \cos(\theta
    \end{pmatrix}.
\end{equation*}
As a quantum gate, we can implement $G_{i,j}(\theta)$ by first mapping $\ket{i}$ to $\ket{0}$ and $\ket{j}$ to $\ket{1}$, then applying a $\mathrm{R}_y(2\theta)$-rotation on the least significant bit controlled on all other bits being zero, and then applying the state swapping operation in reverse. 

To find the angles of the Givens rotations, note that 
\begin{equation}
    OG_{1,2}(-\theta_{1,2})G_{1,3}(-\theta_{1,3})\dots G_{1,d}(-\theta_{1,d})G_{2,3}(-\theta_{2,3})\dots G_{2,d}(-\theta_{2,d})\dots G_{d{-}1, d}(-\theta_{d{-}1,d}) = \mathrm{I}.
\end{equation}
More precisely, each $G_{i,j}$ sets the entry $O^i_j$ to zero, following rows from main diagonal to the right, and going from top row to the bottom row. 
Hence, we can calculate the angle setting the correct entry to zero, update the matrix and continue with computin the next angle. 
If $O'$ is the matrix after updating it using all previously computed angles, then $[O'G_{i,j}(-\theta_{i,j})]^i_j = 0$ gives us 
\begin{equation}
    \sum_{k \in [d]}[O']^i_k [G_{i,j}(-\theta_{i,j})]^k_j = [O']^i_j[G_{i,j}(-\theta_{i,j})]^j_j + [O']^i_i[G_{i,j}(-\theta_{i,j})]^i_j = \cos(\theta_{i,j})[O']^i_j  \sin(\theta_{i,j})[O']^i_i = 0
\end{equation}
We find that 
\begin{equation*}
    \theta = \arctan\left(\frac{[O']^i_j}{[O']^i_i}\right) 
\end{equation*}

We find all angles in a numerically stable manner using the following algorithm, based on Algorithm (5.1.3) of \cite{MatrixComputations}.
With numerically stable, we mean that the round-off error of the angles scales at most linearly with the initial round-off error in the entries of $\mathrm{O}$. 
\begin{algorithmic}\label{alg: angles}
    \State 
    \For{$i = 1, \dots d{-}1$}
        \For{$j = i{+}1, \dots d$}
            \If{$O^{i}_{j} = 0$}
                \State $c_{i,j} \gets 1$, $s_{i,j} \gets 0$, $\theta_{i,j} \gets 0$
            \ElsIf{$|O^i_j| > |O^i_j|$}
                \State $\tau \gets O^i_i / O^i_j$
                \State $s_{i,j} \gets 1 / \sqrt{1 + \tau^2}$
                \State $c_{i,j} \gets s\tau$
                \State $\theta_{i,j} \gets \sign(O^i_j) \arccos(\sign(O^{i}_{j})c_{i,j})$
            \Else
                \State $\tau \gets O^i_j / O^i_i$
                \State $c_{i,j} \gets 1 / \sqrt{1 + \tau^2}$
                \State $s_{i,j} \gets c_{i,j}\tau$
                \State $\theta_{i,j} \gets \arcsin(s)$
            \EndIf
            \State $O \gets OG_{i,j}(-\theta_{i,j})$, using $c_{i,j}$ and $s_{i,j}$ as entries of $G_{i,j}(-\theta_{i,j})$.
        \EndFor
    \EndFor
\end{algorithmic}

Our implementation of $\mathrm{A}_{\lambda}$ consists of the following steps.
\begin{enumerate}
    \item Conditioned on $\ket{\lambda}$, compute the row indices of the addable cells of $\lambda$, together with their contents, and compute the contents of all removable cells, too. 
    \item $\mathrm{U}_{\text{dimensions}}$: conditioned on the content registers, compute $\mathrm{\sqrt{\frac{d_{\lambda + a}}{n d_{\lambda }}}}$ for all addable cells $a \in \AC(\lambda)$ \cref{lemma: twiddle factor}. 
    When when $|\AC(\lambda)|$ is not maximal, use zero padding.
    Note that this register is the first row of $\mathrm{A}_{\lambda}$.
    Similarly, compute $\mathrm{\sqrt{\frac{(n-1)d_{\lambda - r}}{d_{\lambda}}}}$ for all removable cells $r \in \RC(\lambda)$, where we additionally use \cref{lemma: AC and RC of lambda - r}.
    \item $\mathrm{U}_{\mathrm{A}_{\lambda}}$:  compute the rest of matrix $\mathrm{A}_{\lambda}$ coherently using the previously computed dimension ratios and the content registers. 
    \item Run the Givens decomposition method coherently to compute all rotation angles.
    During each step, we only update 2 columns of the matrix register, and store the old register as garbage to be uncomputed later.
    \item Add $1$ mod $m$ to the register storing $\ket{p_{n-1}}$ or $\ket{\star}$, to map $\star$ to $1$ and the rest of the removable cell labels to the index of the addable cell one row below it. 
    Note that the $2$-dimensional subspace correct subspace for each Givens rotation is now stored in the addable cell row register which we coputed in the first step.
    Apply the Givens rotations on the correct subpace, using $\mathrm{R}_y(\theta)$ gates with the rotation angle $\theta$ conditioned on the rotation angle registers.
\end{enumerate}
Of course, we finish by running steps $1$ to $4$ backwards to uncompute all ancillary registers. 
We will now elaborate on the implementation of the different subroutines of the circuit, and their asymptotic space, size and depth.

\subsubsection{Computing row indices and contents of addable and removable cells}\label{subsec: rows and contents of RC AC}
The goal of this step is to prepare 3 $e_n$-wire registers, which we will name $\ket{\mathrm{rows}}$, $\ket{\cont\,\AC}$ and $\ket{\cont\,\RC}$.
We use that $\lambda$ has an addable cell at row $i$ if $\lambda_{i-1} > \lambda_{i}$. 
Also, we always have an addable cell at row $1$. 
We start by preparing a flag register which becomes $1$ at indices where $\lambda_{i-1}$ and $\lambda_i$ are different, and stays $0$ otherwise.
Next to it, we store the row index, and the contents of the associated addable and removable cells. 
We do this with a coherent implementation of the following algorithm:
\begin{algorithmic}
\State Prepare flag, row, AC content and RC content registers $(f_i, r_i, c^{(a)}_i, c^{(r)}_i)$ with $i = 1, \dots, n$.
\State Set $(f_1, r_1) = (1, 1, \lambda_1, \star)$ and the rest to $0$.
\For{$i = 2, \dots n$}
    \If{$\lambda_{i-1} > \lambda_i$}
        \State $f_i \gets 1$
        \State $r_i \gets i$
        \State $c^{(a)}_i \gets \lambda_i - i + 1$
        \State $c^{(r)}_i \gets \lambda_{i-1} - i + 1$
    \EndIf 
\EndFor
\end{algorithmic}

The next step is to sort these registers, so that the first $e_n$ entries contain the information corresponding to addable end removable cells. 
For sorting, we only require the flag registers and row registers $(f_i, r_i)$. 
We say that 
$(f_i, r_i) < (f_j, r_j)$ if $f_i > f_j$ or $f_i = f_j$ and $r_i < r_j$. 
It is possible to construct a constant time quantum gate which compares two registers and swaps them if they are not ordered. 
Using a sorting network,
we can perform this sorting using $\mathcal{O}(n \log^2(n))$ of these compare-and-swap gates, and with a depth of $\mathcal{O}(\log^2(n))$. 

Afterwards, we can copy the first $e_n$ row indices and contents into our final registers if the flag register is 1, and uncompute the previously computed registers if minimizing garbage is a priority. 

For the complexity, computing the flag registers with the extra information creates $\widetilde{\mathcal{O}}(n)$ new registers, and uses $\widetilde{\mathcal{O}}(n)$ gates. 
Note that by running the algorithm for odd and even indices separately, we can implement this step in constant depth. 
We have already seen that the sorting step has size $\widetilde{\mathcal{O}}(n)$ and depths $\widetilde{\mathcal{O}}(1)$. 
Lastly, copying the $e_n = \mathcal{O}(\sqrt{n})$ registers can be performed with $\mathcal{O}(\sqrt{n})$ gates and in constant depth.

\subsubsection{Summary}

\begin{table}[H]
\centering
\begin{tabular}{lccc}
\textbf{Step}               & \textbf{Space}      & \textbf{Gatecount} & \textbf{Depth} \\ \hline
1. Rows, contents           & $\widetilde{\mathcal{O}}(n)$ & $\widetilde{\mathcal{O}}(n)$ & $\widetilde{\mathcal{O}}(1)$\\
2. Dimension ratios         & $\widetilde{\mathcal{O}}(\sqrt{n})$  & $\widetilde{\mathcal{O}}(n)$ & $\widetilde{\mathcal{O}}(\sqrt{n})$          \\
3. Matrix $\mathrm{A}_{\lambda}$ & $\widetilde{\mathcal{O}}(n)$  & $\widetilde{\mathcal{O}}(n)$ & $\leq \widetilde{\mathcal{O}}(n)$\\
4. Coherent Givens decomposition & $\widetilde{\mathcal{O}}(n^{3/2})$ & $\widetilde{\mathcal{O}}(n^{3/2})$ & $\widetilde{\mathcal{O}}(n)$ \\
5. Applying Givens rotations  & - & $\widetilde{\mathcal{O}}(n)$ & $\widetilde{\mathcal{O}}(n)$ \\
 \hline\textbf{Total} & $\widetilde{\mathcal{O}}(n^{3/2})$ & $\widetilde{\mathcal{O}}(n^{3/2})$ & $\widetilde{\mathcal{O}}(n)$ 
\end{tabular}
\end{table}

\section{Proofs}

\subsection{Proofs for \texorpdfstring{\cref{sec: preliminaries}}{preliminaries}}

\begin{lemma}\label{lemma: absolute manhattan distance of 1} For $\lambda \in \widehat{S}_n$ with $n > 1$, let $P \in \mathcal{P}(\lambda)$ and $k \in [n{-}1]$.
Then the absolute Manhattan distance satisfies $|r_P(k)| = 1$ if and only if $\sigma_k P \notin \Path(\lambda)$.
\end{lemma}

\begin{proof}
In the proof, we consider $P$ as an SYT.
First, assume that $|r_P(k)| = 1$.  
This means that the box $k{+}1$ is either directly below $k$ or directly to the right of it. 
In both cases, swapping boxes $k$ and box $k+1$ does not result in a valid SYT, meaning that $\sigma_k P \notin \Path(\lambda)$.

Second, assume that $\sigma_k P \notin \Path(\lambda)$. 
Then, box $k$ and $k{+}1$ of $P$ must be adjacent, meaning that horizontally or vertically, there are no boxes in between them. 
We will show this by contradiction. 
Suppose that box $k$ and $k{+}1$ of $P$ do not ``touch''. 
By definition of a standard Young tableau, the boxes left and on top of both $k$ and $k{+}1$ should have entries smaller than $k$, and the boxes below and to the right of them should have entries larger than $k{+}1$. 
In this case, swapping $k$ and $k{+}1$ to obtain $\sigma_k P$ will still result in a valid SYT, meaning that $Q^{k\leftrightarrow k{+}1} \in \mathcal{P}(\widehat{S}_n)$, which is a contradiction.
Hence, $k$ and $k{+}1$ are adjacent in $P$, which implies that their absolute Manhattan distance is $|r_P(k)|= 1$, which concludes our proof.
\end{proof}

\begin{lemma}\label{lemma: optimal path encoding size}
    Consider a multiplicity-free chain of subgroups $\{e\} = G_1 \subset G_2 \subset G_3 \subset \dots$ such that $G_n$ is finite for all $n \in \mathbb{N}$. 
    Then, the optimal number of qubits $B_{\mathrm{opt}}$ of any encoding of paths in the Bratteli diagram satisfies $B^{(G_n)}_{\mathrm{opt}} = \Theta(\log|G_n|)$.
\end{lemma}

\begin{proof}
    Note that for any encoding, we need at least $B^{(G_n)}_{\mathrm{opt}} = \lceil \log_2|\Path_n|\rceil$ qubits. 
    To find the asymptotic scaling, we start by finding an upper bound for the number of paths $|\Path_n|$.
    For all $\lambda \in \widehat{G}_n$, we have $d_\lambda \geq 1$, so that $d_\lambda \leq d_\lambda^2$.
    We find that 
    \begin{equation}
        \sum_{\lambda \in \widehat{G}_n} d_\lambda \leq \sum_{\lambda \in \widehat{G}_n}d_\lambda^2 = |G_n|.
    \end{equation}
    The last equality can be proven by counting the dimensions of the spaces on both sides of \cref{eq: Peter--Weyl duality Sn}, where $S_n$ is replaced by $G_n$.
    Note that the upper bound is tight if and only if $G_n$ is abelian. 
    
    For finding a lower bound of $|\Path_n|$, note that 
    \begin{align}
        |\Path_n|^2 = \Bigg(\sum_{\lambda \in \widehat{G}_n} d_\lambda\Bigg)^2 = \sum_{\lambda \in \widehat{G}_n} d_\lambda^2 + 2\sum_{\lambda, \mu \in \widehat{G}_n, \lambda < \mu} d_\lambda d_\mu \geq \sum_{\lambda \in \widehat{G}_n} d_\lambda^2 = |G_n|.
    \end{align} 
    Taking the square root of both sides, we find that $\sqrt{|G_n|} \leq |\Path_n|$.

    Putting them together, we find that $\sqrt{|G_n|} \leq |\Path_n| \leq |G_n|$.
    Taking the logarithm and ceiling to get $B^{(G_n)}_{\mathrm{opt}} = \lceil \log_2 |\Path_n|\rceil$, we find 
    \begin{equation}
        \frac{1}{2} \log_2 |G_n| \leq B^{(G_n)}_{\mathrm{opt}} \leq \log_2|G_n| + 1, 
    \end{equation}
    so that $B^{(G_n)}_{\mathrm{opt}} = \Theta(\log|G_n|)$.\qedhere
\end{proof}

\begin{corollary}\label{cor: optimal encoding Sn paths}
    For $S_n$, any register encoding of elements of $\Path_n$ optimally requires $B^{(S_n)}_{\mathrm{opt}} = \Theta (n \log n)$ qubits.
\end{corollary}

\begin{proof}
    Using \cref{lemma: optimal path encoding size} and Stirling's approximation stating that $n! \sim \sqrt{2\pi n}\left(n/e\right)^n$, we find that 
    \begin{align}
        B^{(S_n)}_{\mathrm{opt}} &= \Theta(\log_2|S_n|) \\
        &= \Theta(\log_2(n!)) \\
        &= \Theta\left(\log_2\left(\Theta\left(\sqrt{2\pi n}\left(n/e\right)^n\right)\right)\right)\\
        &= \Theta\left(\log_2\left(\sqrt{2\pi n}\left(n/e\right)^n\right)\right) &&\text{(using that } \lim_{n\to\infty}\sqrt{2\pi n}\left(n/e\right)^n = \infty\text{)}\\
        &=\Theta(n \log_2(n/e) + \log_2(2\pi n)/2)\\
        &= \Theta(n \log_2n).
    \end{align}\qedhere
\end{proof}

\begin{lemma}\label{lemma: Sn absolute path min qubits}
    For the symmetric group, the minimal number of qubits $B^{(S_n)}_{\mathrm{abs}, \min}$ for an absolute path encoding of paths of length $n$ in the Bratteli diagram satisfies
    \begin{equation}
        B^{(S_n)}_{\mathrm{abs}, \min} = \Theta(n^{3/2}).
    \end{equation}
\end{lemma}

\begin{proof}
    First, consider a single register $\ket{P_m}$, where $P_m$ can be any element of $\widehat{S}_m$.
    Most efficiently, we need $\lceil \log_2|\widehat{S}_m| \rceil$ qubits. 
    Equation 1.41 of \cite{HardyRamanujan1917} states that the number of integer partitions of $m$ can approximated as $|\widehat{S}_n| \sim \frac{1}{4n\sqrt{3}}e^{\pi\sqrt{\frac{2n}{3}}}$ when $n$ goes to infinity. 
    This implies that $|\widehat{S}_n| = \Theta\left(\frac{1}{n} e^{\pi\sqrt{2n/3}}\right)$.
    We find 
    \begin{align}
        \log_2|\widehat{S}_n| &= \log_2\left(\Theta\left( \frac{1}{n} e^{\pi\sqrt{2n/3}}\right)\right)\\
        &= \Theta\left( \log_2\left(\frac{1}{n} e^{\pi\sqrt{2n/3}}\right)\right)&&\text{(using that }\lim_{n \to \infty}\frac{1}{n} e^{\pi\sqrt{2n/3}} = \infty)\\
        &= \Theta\left( \frac{\pi\sqrt{2n}}{\ln(2)\sqrt{3}} - \log_2 n\right)\\
        &= \Theta(\sqrt{n}).
    \end{align}
    This implies that the minimal number of qubits for $\ket{P_m}$ is $\lceil \log_2|\widehat{S}_m|\rceil = \Theta(\sqrt{m})$, too.
    For the entire path register, we find that
    \begin{equation}
        B_{\mathrm{abs}, \min} = \sum_{m = 1}^n \Theta(\sqrt{m}) = \Theta(n^{3/2}).\qedhere
    \end{equation}
\end{proof}

By \cref{cor: optimal encoding Sn paths}, this implies that the absolute path encoding for $S_n$ does not have an optimal asymptotic space complexity. 
Moreover, in a realistic encoding, the canonical encoding of integer partitions 

\begin{lemma}\label{lemma: shape register qubit count}
    Let $\ket{\lambda}$ be a register storing integer partitions $\lambda$ of at most $n$ using a separate register $\ket{\lambda_k}$ for each row $k = 1, \dots, n$, possibly with zero padding.
    For each row register, we use exactly enough qubits to store any valid integer partition.
    Then, we require $\mathcal{O}(n)$ qubits. 
\end{lemma}

\begin{proof}
    First, let us find the number of qubits for each row register $\ket{\lambda_k}$.
    As $\lambda$ must be a valid integer partition, we require that $\lambda_1, \dots,\lambda_k \geq \lambda_{k}$.
    This implies that 
    \begin{equation}
        k \lambda_k \leq \sum_{i=1}^k \lambda_i \leq n.
    \end{equation}
    Dividing by $k$, and using that $\lambda_k \geq 0$, we find that $0 \leq \lambda_k \leq \frac{n}{k}$.
    Hence, $\lambda_k$ can take at most $\frac{n}{k}+1$ values, so $\ket{\lambda_k}$ can be implemented using at most
    \begin{equation}
        \left\lceil \log_2\left(\frac{n}{k}+1\right)\right\rceil \leq \log_2\left(\frac{n}{k}+1\right)+1
    \end{equation}
    qubits.
    For the total number of qubits of $\ket{\lambda}$, we find that we need no more than 
    \begin{align}
        \sum_{k=1}^n \left(\log_2\left( \frac{n}{k} + 1\right) + 1\right)
        &= n + \sum_{k=1}^n \left(\log_2(n+k) - \log_2(k)\right)\\
        &= n + \log_2\left(\frac{(2n)!}{(n!)^2}\right)\\
        &= n + \log_2\binom{2n}{n}\\
        &\leq n + 2n\\
        &= \mathcal{O}(n),
    \end{align}
    where we used $\binom{2n}{n} \leq 2^{2n}$.
\end{proof}

Hence, if we store length-$n$ paths by their absolute encoding in their most natural encoding as in \cref{lemma: shape register qubit count}, we even need $\mathcal{O}(n^2)$ qubits, instead of the optimal $\mathcal{O}(n^{1.5})$ qubits from \cref{lemma: Sn absolute path min qubits}.

\begin{lemma}\label{lemma: Yamanouchi word optimal scaling}
    The Yamanouchi word encoding of a path of length $n$ in the Bratteli diagram for $S_n$ requires $B^{(S_n)}_{\mathrm{Ymn}} = \mathcal{O}(n\log n)$ qubits, which is optimal in the sense of \cref{lemma: optimal path encoding size}.
\end{lemma} 

\begin{proof}
    First, note that $1 \leq q_m \leq m$, because $Q_{m-1}$ has at most $m-1$ boxes, and therefore any addable cell in $\AC(Q_{m-1})$ has at least row index $1$ and at most row index $m$.
    We can store $q_m$ in an $m$-dimensional qudit, implemented with $\lceil \log_2(m)\rceil$ qubits. 
    For the total number of qubits, we find that 
    \begin{equation*}
        \sum_{m = 1}^n \lceil \log_2(m) \rceil \leq n +  \sum_{m = 1}^n \log_2(m) = n + \log_2(n!) = \mathcal{O}(n \log n). 
    \end{equation*}
    In the last step, we used Stirling's approximation, stating that $n! = \Theta(\sqrt{2 \pi n} \left(\frac{n}{e}\right)^n)$, which implies that $\log_2(n!) = \mathcal{O}(n \log n)$. 
    
    By \cref{lemma: optimal path encoding size}, we need $\Theta (\log | S_n|) = \Theta (\log (n!)) = \Theta(n \log n)$ qubits for a path encoding with optimal space complexity, showing that Yamanouchi word encoding has an optimal scaling.
\end{proof}

Combining \cref{cor: optimal encoding Sn paths} and \cref{lemma: Yamanouchi word optimal scaling}, we see that Yamanouchi words encode of paths with optimal asymptotic space complexity.
Even if we add registers for the $x$-coordinates, the number of qubits multiplies by 2, and therefore the encoding still has an optimal space complexity.

\subsection{Proofs for \texorpdfstring{\cref{sec: induction map Mackey}}{Mackey induction map}}\label{app: Mackey proofs}

\subsubsection{Proof of \texorpdfstring{\cref{lemma: Sn-1 double cosets of Sn}}{Sn-1 double cosets of Sn}}\label{app: Sn-1 double cosets of Sn}
\begin{proof}
     Clearly, for any group, the identity element results in a double coset which is just the subgroup. 
     Also, it is clear that $H_e = S_{n-1} \cap e S_{n-1} e^{-1} = S_{n-1} \cap S_{n-1} = S_{n-1}$. 
     Note that $\sigma_{n-1} \notin S_{n-1}$, so that we can use this element as the representative of another $S_{n-1}$-double coset of $S_n$. 
     For $H_{\sigma_{n-1}} = S_{n-1} \cap \sigma_{n-1}S_{n-1}\sigma_{n-1}^{-1}$, first note that $\sigma_{n-1}$ commutes with all elements of $S_{n-2}$, so that $S_{n-2} \subseteq H_{\sigma_{n-1}}$. 
     Second, let $g \in H_{\sigma_{n-1}}$, so that $g = \sigma_{n-1} h \sigma_{n-1}$ for some $h \in S_{n-1}$, and $g \in S_{n-1}$. 
     The last condition implies that $g (n) = n$. 
     Combining this with first condition gives $g (n) = (\sigma_{n-1} h \sigma_{n-1})(n) = (\sigma_{n-1}h)(n{-}1)= \sigma_{n-1}(h(n{-}1)) = n$. 
     Equivalently, we say that $h(n{-}1) = \sigma_{n-1}^{-1}(n) = (n{-}1)$, hence $h \in \mathrm{Stab}(n-1) \cap S_{n-1} = S_{n-2}$, showing that $H_{\sigma_{n-1}} \subseteq S_{n-2}$.
     We conclude that $H_{\sigma_{n-1}} = S_{n-2}$.

     Finally, to show that $\Omega_n = \{e, \sigma_{n-1}\}$ is a complete set of $S_{n-1}$-double coset representatives, it suffices to show that $|\bigoplus_{\omega \in \Omega_n} S_{n-1}\omega  S_{n-1}| = |S_{n}|$, as the double cosets partition $S_n$ by \cref{lemma: double cosets partition group}. 
     The double coset formula from \cref{lemma: double coset formula} gives $|S_{n-1}\sigma_{n-1} S_{n-1}| = |S_{n-1}|^2 / H_{\sigma_{n-1}} = ((n-1)!)^2 / (n-2)! = (n-1)(n-1)!$.
     Hence, $|\bigoplus_{\omega \in \Omega_n} S_{n-1}\omega  S_{n-1}| = |S_{n-1}| + (n-1)(n-1)! = (n-1)! + (n-1)(n-1)! = n (n-1)! = n! = |S_{n}|$, which concludes our proof.
\end{proof}

In our proof, we use the following formula for the size of a double coset:
\begin{lemma}[Double coset formula]\label{lemma: double coset formula}
    Let $H, K \subseteq G$ be subgroups of a finite group $G$. 
    For $g \in G$, we have 
    \begin{equation}
        |HgK| = \frac{|H||K|}{|H_g|},
    \end{equation}
    where we defined subgroup $H_g = H \cap g K g^{-1}$.
\end{lemma}

\begin{proof}
    Consider the surjective map $\Phi : H \times K \to H g K$, $\Phi(h,k) = hgk$. 
    Let $g' \in  H g K$. 
    As $\Phi$ is surjective, the preimage $\Phi^{-1}(g')$ is non-empty, and we can fix an element $(h, k) \in \Phi^{-1}(g')$. 
    We will now construct a bijection $H_g \to \Phi^{-1}(g')$ to show that $|\Phi^{-1}(g')| = |H_g|$, where we define 
    \begin{equation}
        x \mapsto (hx, (g^{-1}xg)^{-1}k).
    \end{equation}
    First, we verify that $hx \in H$, because $h \in H$ and $x \in H_g \subseteq H$. 
    Second, we verify that $(g^{-1}xg)^{-1}k \in K$, because $k \in K$ and $x \in g K g^{-1}$ so that $(g^{-1}xg)^{-1} \in K$.
    Lastly, we verify that $(hx, (g^{-1}xg)^{-1}k) \in \Phi^{-1}(g')$, because
    \begin{equation}
        \Phi(hx, (g^{-1}xg)^{-1}k) = (hx) g ((g^{-1}x g)^{-1}k) = hx g g^{-1}x^{-1}gk = hxx^{-1}gk = hgk = g'.
    \end{equation}

    To show that the map is injective, the $H$-part simply left-multiplies $x$ with $h$, which is injective. 
    The $K$-part can be seen as a composition of three injective maps, where we first assigne the unique element $g^{-1}xg \in K$ to each $x \in H_g$, then takes the inverse, and finally right multiply with $k$.

    To show that the map is surjective, consider $(h', k') \in \Phi^{-1}(g')$. 
    Then, we can always define $x = h^{-1}h' \in H_g$ such that $x \mapsto (h', k')$ 
    First, to show that $x \in H_g = H \cap g K g^{-1}$, we clearly have that $x \in H$. 
    For showing that $x \in g K g^{-1}$, recall that $(h', k') \in \Phi^{-1}(g')$ means that $h'gk' = g' = hgk$.
    Using that $h' = hx$, we rewrite this so that $hxgk' = hgk$.
    Then, we find that $x = gkk'^{-1}g^{-1} \in g K g^{-1}$.
    Second, note that 
    \begin{align}
    x &\mapsto (hx, (g^{-1}xg)^{-1}k)\\
    &= (h(h^{-1}h', (g^{-1}(gkk'^{-1}g^{-1})g)^{-1}k)\\
    &= (h', g^{-1}gk'k^{-1}g^{-1}gk)\\
    &= (h', k'k^{-1}k) = (h', k').
    \end{align}
    
    Finally, we know that the domain of $\Phi$ has cardinality $|H||K|$, which should be equal to $\sum_{g' \in HgK}|\Phi^{-1}(g')|$.
    We find that 
    \begin{equation}
        |H||K| = \sum_{g' \in HgK} |\Phi^{-1}(g')| = \sum_{g' \in HgK} |H_g| = |HgK| |H_g|.
    \end{equation}
    Dividing both sides by $|H_g|$ gives us the double coset formula.
\end{proof}

In this section, we will use the following lemmas.
Note that an integer partition $\lambda = (\lambda_1, \dots, \lambda_t)$ has an addable cell at row $k$ if $\lambda_{k-1} > \lambda_k$, where we consider $\lambda_0 = \infty$ and $\lambda_{t+1} = 0$.
Similarly, it has a removable cell at row $k$ if $\lambda_{k} > \lambda_{k+1}$.

\begin{lemma}\label{lemma: AC of lambda - r}
For $n \in \mathbb{N}_{>0}$, it holds for all $\lambda \pt n$ and for all removable cells $r \in \RC(\lambda)$that $\AC(\lambda - r) \setminus \{r\} \subseteq \AC(\lambda)$.
\end{lemma}
\begin{proof}
    Let $y$ be the row number of removable cell $r \in \RC(\lambda)$.
    Then, we have that $\lambda - r = (\lambda_1, \dots, \lambda_{y-1}, \lambda_{y}-1, \lambda_{y+1}, \dots, \lambda_t)$.
    Let $a \in \RC(\lambda - r) \setminus\{r\}$, and denote by $k$ its row number. 
    As $a$ is an addable cell of $\lambda - r$ and not equal to $r$, we know that $(\lambda-r)_{k-1} > (\lambda - r)_k$ and that $k \neq y$.
    If $k < y$ or $k > y+1$, then $(\lambda-r)_{k-1} = \lambda_{k-1}$ and $(\lambda-r)_k = \lambda_{k}$, so that $(\lambda-r)_k > \lambda_{k}$ meaning that $a \in \AC(\lambda)$.
    Otherwise, if $k = y + 1$, then $\lambda - r$ has an addable cell at row $k$ and column $(\lambda - r)_{k} + 1 = (\lambda-r)_{y + 1} + 1 = \lambda_{y+1} + 1$.
    As $r$ is the removable cell of $\lambda$ at row $y$, we know that $\lambda_{y} > \lambda_{y+1}$, so that $\lambda$ has a removable cell at row $y+1 = k$ and column $\lambda_{y+1} + 1$ as well, so that $a \in \AC(\lambda)$.
    We have shown that all elements of $\AC(\lambda - r) \setminus \{r\}$ are also an element of $\AC(\lambda)$ and conclude that $\AC(\lambda - r) \setminus \{r\} \subseteq \AC(\lambda)$.  
\end{proof}
\begin{lemma}\label{lemma: RC of lambda + a}
For $n \in \mathbb{N}_{\geq 0}$, it holds for all $\lambda \pt n$ and for all addable cells $a \in \AC(\lambda)$ that $\RC(\lambda + a) \setminus \{a\} \subseteq \RC(\lambda)$.
\end{lemma}
\begin{proof}
    Suppose that $a$ is an addable cell at row $y$, so that $\lambda + a = (\lambda_1, \dots, \lambda_{y-1}, \lambda_{y} + 1, \lambda_{y + 1}, \dots, \lambda_t)$.
    Let $r \in \RC(\lambda + a) \setminus \{a\}$ have row index $k \neq y$.
    Hence, we know that $(\lambda_y + a)_{k} > (\lambda_y + a)_{k+1}$.
    If $k < y-1$ or $k > y$, it holds that $(\lambda_y + a)_{k} = \lambda_k$ and $(\lambda_y + a)_{k+1} = \lambda_{k+1}$, so that $\lambda_k > \lambda_{k+1}$ and $r \in \RC(\lambda)$.
    If $k = y - 1$, then $(\lambda + a)_{k} = (\lambda + a)_{y-1} = \lambda_{y-1} = \lambda_k$ and $(\lambda + a)_{k+1} = (\lambda + a)_y = \lambda_y + 1 = \lambda_{k} + 1$.
    As $\lambda$ has an addable cell at row $y$, we know that $\lambda_{y-1} > \lambda_y$, and therefore $r \in \RC(\lambda)$. 
    We have shown that all elements of $\RC(\lambda + a) \setminus \{a\}$ are also an element of $\RC(\lambda)$ and conclude that $\RC(\lambda + a) \setminus \{a\} \subseteq \RC(\lambda)$.  
\end{proof}

\begin{lemma}\label{lemma: AC RC cardinalities}
    For $\lambda \pt n$, we have that 
    \begin{equation*}
        |\AC(\lambda)| = |\RC(\lambda)| + 1.
    \end{equation*}
\end{lemma}
\begin{proof}
    Let $r \in \RC(\lambda)$, and denote by $y$ the row number of $r$.
    Hence, $\lambda_{y} > \lambda_{y+1}$, which means that there is an addable cell at row $y+1$.
    We conclude that any removable cell corresponds to an addable cell in the next row. 
    The only row that cannot be reached in this way is row $1$. 
    As there is always an addable cell at row $1$, we conclude that $|\AC(\lambda)| = |\RC(\lambda)| + 1$.
\end{proof}

\begin{wrapfigure}{r}{3.68cm}
    \begin{center}
        \includegraphics[width=3.67cm]{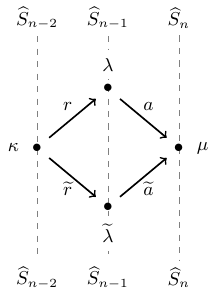}
        \caption{Subgraph of the Bratteli diagram for $S_n$ corresponding to the diamond lemma.}
        \label{fig: diamond lemma}
    \end{center}
\end{wrapfigure}

\begin{lemma}[Diamond lemma]\label{lemma: diamond}
    Let $n \in \mathbb{N}, n \geq 2$. 
    Let $\lambda, \widetilde{\lambda} \pt n{-}1$. 
    Then, we have that 
    \begin{equation}
        \exists \kappa \pt n{-}2: \kappa \to \lambda, \widetilde{\lambda} \quad \Longleftrightarrow \quad  \exists \mu \pt n : \lambda, \widetilde{\lambda} \to \mu.
    \end{equation}
    Moreover, if $\lambda \neq \widetilde{\lambda}$, then $\kappa$ and $\mu$ are unique if they exist.
\end{lemma}
In this lemma, we consider $\varnothing \pt 0$ as a node, too.
\cref{fig: diamond lemma} shows the corresponding subgraph of the Bratteli diagram when $\lambda \neq \widetilde{\lambda}$.

\begin{proof}
In this proof, we will identify $\kappa \pt n{-}2$ with a removable cell $r$ in $\RC(\lambda)$ and a removable cell $\widetilde{r}$ in $\RC(\widetilde{\lambda})$, so that $\kappa = \lambda - r = \widetilde{\lambda} - \widetilde{r}$.
Similarly, define addable cells $a \in \AC(\lambda)$ and $\widetilde{a} \in \AC(\widetilde{\lambda})$ such that $\mu = \lambda + a = \widetilde{\lambda} + \widetilde{a}$.
These cells are also indicated in \cref{fig: diamond lemma}.

First, consider $\lambda = \widetilde{\lambda}$.
Then, $\RC(\lambda) = \RC(\widetilde{\lambda})$ and $\AC(\lambda) = \AC(\widetilde \lambda)$.
As $\lambda \pt n{-}1$ with $n \geq 2$, we know that $\lambda$ one cell, so that $|\RC(\lambda)| \geq 1$.
By \cref{lemma: AC RC cardinalities}, this implies that $|\AC(\lambda)| \geq 2$. 
As we can always pick an $r \in \RC(\lambda)$ and $a \in \AC(\lambda)$, we can always find $\kappa$ and $\mu$ as in the lemma, and both sides of the logical equivalence are always true.

When $\lambda \neq \widetilde{\lambda}$, it must hold that $r \neq \widetilde{r}$ and that $a \neq \widetilde{a}$. 
First, assume that there exists a $\kappa \pt n{-}2$ such that $\kappa \to \lambda, \widetilde{\lambda}$.
This uniquely identifies $r$ and $\widetilde{r}$. 
As $\lambda = \kappa + r$, we know that $r \in \AC(\lambda)$.
By \cref{lemma: AC of lambda - r}, we know that $\AC(\kappa) \setminus \{\widetilde{r}\} \subseteq \AC(\widetilde{\lambda})$.
Therefore, $r \in \AC(\widetilde{\lambda})$.
By symmetry, it also holds that $\widetilde{r} \in \AC(\lambda)$.

Integer partitions can be uniquely identified by their set of cells.
We see that $\Cells(\lambda + \widetilde{r}) = \Cells(\lambda) \cup \{\widetilde{r}\} = \Cells(\kappa + r) \cup \{\widetilde{r}\} = \Cells(\kappa) \cup \{r, \widetilde{r}\} = \Cells(\kappa + \widetilde{r}) \cup \{r\} = \Cells(\widetilde{\lambda}) \cup \{r\} = \Cells(\widetilde{\lambda} + r)$.
We may define $\mu = \lambda + \widetilde{r} = \widetilde{\lambda} + r$, which gives us the unique addable cells $a = \widetilde{r}$ and $\widetilde{a} = r$.

Second, assume that there exists a $\mu \pt n$ such that $\lambda, \widetilde{\lambda} \to \mu$.
This uniquely identifies $a$ and $\widetilde{a}$. 
As $\lambda + a = \mu$, we know that $a \in \RC(\mu)$.
By \cref{lemma: RC of lambda + a}, we know that $\RC(\mu) = \RC(\widetilde{\lambda} + \widetilde{a}) \subseteq \RC(\widetilde{\lambda}) \cup \{\widetilde{a}\}$. 
Therefore, $a \in \RC(\widetilde{\lambda})$.
By symmetry, it also holds that $\widetilde{a} \in \RC(\lambda)$.
Like before, we can show that $\Cells(\lambda - \widetilde{a}) = \Cells(\widetilde{\lambda} - a)$.
We find $\kappa = \lambda - \widetilde{a} = \widetilde{\lambda} - a$, which uniquely identifies $r = \widetilde{a}$ and $\widetilde{r} = a$.
\end{proof}

We will also use Lemma 4.3 of \cite{DiagramCombinatorics}, which is as follows.
\begin{lemma}
    Just like in \cref{lemma: diamond}, let $\lambda, \widetilde{\lambda} \pt n{-}1$, $\kappa \pt n{-}2$ and $\mu \pt n$ such that $\kappa \to \lambda, \widetilde{\lambda} \to \mu$ and assume that $\lambda \neq \widetilde{\lambda}$.
    Then, it holds that 
    \begin{equation*}
        \frac{d_\lambda d_{\widetilde{\lambda}}}{d_\mu d_\kappa} = \frac{n}{n-1}\left(1 - \frac{1}{(\cont(a) - \cont(r))^2}\right)
    \end{equation*}
    Here, we have $a \in \AC(\lambda)$ and $r \in \RC(\lambda)$ such that $\lambda + a = \mu$ and $\lambda - r = \kappa$.
\end{lemma}

\subsubsection{Proof of \texorpdfstring{\cref{lemma: extended module isomorphism}}{extended module isomorphism}}
\begin{proof}
    By definition of a $\C[H]$-module isomorphism, we that know that $\phi$ is an invertible linear map, and that it is compatible with the action of $H$.
    Hence, for all $h \in H$ and $\ket{v} \in V$, we have that 
    \begin{equation}
        \phi(h \cdot_V \ket{v}) = h \cdot_W \phi \ket{v}
    \end{equation}
    First, for the $\C[G]$-module structure on $W$, define the action of $g \in G$ on $\ket{w} \in W$ as 
    \begin{equation}
        g \cdot_W \ket{w} \defeq \phi(g\cdot_V \phi^{-1}\ket{w}).
    \end{equation}
    This defines a $\C[G]$-module structure, as 
    \begin{enumerate}
        \item the action is a composition of linear maps $\phi^{-1}$, $\cdot_V$ and $\phi$ and therefore linear;
        \item it is compatible with the identity, because for all $\ket{w} \in W$, we have $e \cdot_W \ket{w} = \phi(e \cdot_V \phi^{-1}\ket{w}) = \phi( \phi^{-1}\ket{w}) = \ket{w}$;
        \item it is compatible with multiplication in $G$, because for all $g_1, g_2 \in G$ and $\ket{w} \in W$, we have that 
        \begin{align}
            (g_1 g_2)\cdot_{W}\ket{w} &= \phi((g_1 g_2)\cdot_{V}\phi^{-1}\ket{w})\\
            &= \phi(g_1 \cdot_V (g_2\cdot_{V}\phi^{-1}\ket{w}))\\
            &=  \phi(g_1 \cdot_V \phi^{-1}\phi(g_2\cdot_{V}\phi^{-1}\ket{w}))\\
            &= \phi(g_1 \cdot_V \phi^{-1} (g_2\cdot_{W}\ket{w}))\\
            &= g_1 \cdot_W (g_2\cdot_{W}\ket{w}).
        \end{align}
    \end{enumerate}
    Second, to show that this extends the $\C[H]$-module structure of $W$, denote the original action of $H$ on $W$ by $\cdot^{\text{old}}$.
    Then, for $h \in H \subset G$ and $\ket{w} \in W$, we have 
    \begin{align}
        h \cdot_W \ket{w} = \phi(h \cdot_V \phi^{-1}\ket{w}) = h \cdot_W^{\text{old}} \phi(\phi^{-1}\ket{w}) =h \cdot_W^{\text{old}} \ket{w} .
    \end{align}
    Third, for $\ket{v} \in V$ and $g \in G$, it holds that $ \phi(g\cdot_V \ket{v}) = \phi(g\cdot_V (\phi^{-1} \phi \ket{v})) = g \cdot_W \phi \ket{v}$, which shows that $\phi$ is compatible with the action of $G$ and therefore a $\C[G]$-module isomorphism.
\end{proof}

\subsubsection{Proof of \texorpdfstring{\cref{thm: second induction}}{second induction}}
\begin{proof}
    For the three isomorphisms $\IndMap[\lambda][n]$, $\mathrm{M}_n \otimes \mathrm{I}_{d_\lambda}$ and $\mathrm{I}_{d_\lambda} \oplus \left(\bigoplus_{\kappa \in \N^-(\lambda)}\IndMap[\kappa][n-1]\right)$, we use \cref{lemma: extended module isomorphism} applied to \cref{lemma: commutative diagram for embedding operation} to get the $\C[S_n]$-module structure.
    Then, we know that 
    \begin{equation*}
        \bigoplus_{\widetilde{\lambda}} V_{\widetilde{\lambda}} \otimes M^-_{\lambda, \widetilde{\lambda}} \simeq \bigoplus_{\widetilde{\lambda}} V_{\widetilde{\lambda}} \otimes M^+_{\lambda, \widetilde{\lambda}}
    \end{equation*}
    as $\C[S_n]$-modules.
    However, contrary to \cref{lemma: commutative diagram for embedding operation}, $ \bigoplus_{\widetilde{\lambda}} V_{\widetilde{\lambda}} \otimes M^-_{\lambda, \widetilde{\lambda}}$ is now not a direct sum of simple $\C[S_n]$-modules. 

    Still, the $\C[S_n]$-module structure should extend the $\C[S_{n-1}]$ modules by \cref{lemma: extended module isomorphism}, and therefore the map $\mathrm{A}$ should still have the form of \cref{eq: block structure A} as in \cref{lemma: commutative diagram for embedding operation}.
    To make the diagram commute, we will need to find matrices $\mathrm{A}_{\lambda, \widetilde{\lambda}}$ satisfying
    \begin{equation}\label{eq: objective for finding entries of A}
        \IndMap[\lambda][n] = \Big(\bigoplus_{\widetilde{\lambda} \in \widehat{S}_{n-1}} \mathrm{I}_{d_{\widetilde{\lambda}}} \otimes \mathrm{A}_{\lambda, \widetilde{\lambda}} \Big)\Big(\mathrm{I}_{d_\lambda} \oplus \bigoplus_{\kappa \in \N^-(\lambda)} \IndMap[\kappa][n-1]\Big) \left(\mathrm{M}_n \otimes \mathrm{I}_{d_{\lambda}}\right)
    \end{equation}
    The linear maps both sides of this equation have domain $\C[\T_n] \otimes V_\lambda$. 
    The map on the right hand side will act differently on basis vectors $\ket{t}\otimes\ket{P}$ according to the $S_{n-1}$-double coset to which $t$ belongs.
    As $\T_n = \T_{n-1}\sigma_{n-1} \sqcup \{e\}$ by \cref{prop: Mackey transversals for Sn}, we will split up the rest of this proof into the two corresponding cases, starting with $t = e \in \{e\}$.
    In both cases, we use that all maps besides $\mathrm{A}$ are orthogonal, and therefore $\mathrm{A}$ should be orthogonal, too.

\subsubsection*{Case 1: \texorpdfstring{$t = e$}{t is e}}
    
    By \cref{def: IndMap} of $\IndMap[\lambda]$, the entries of the LHS of \cref{eq: objective for finding entries of A} are
    \begin{equation*}
        [\IndMap[\lambda][n]]_{n, p}^{q, \mu} = \sqrt{\frac{d_\mu}{n d_\lambda}} [\Rep_\mu(e)]^{q}_{p \to \mu} = \sqrt{\frac{d_\mu}{n d_\lambda}} \delta^{q}_{p \to \mu}.
    \end{equation*}
    On the RHS, Mackey decomposition changes the transversal label into the special symbol $\star$.
    The second map acts on this subspace as identity. 
    Equating both sides of the equation gives
    \begin{equation}
        \sqrt{\frac{d_\mu}{n d_\lambda}}\delta^{q}_{p \to \mu} = \delta^{q'}_{p}[\mathrm{A}_{\lambda, \lambda}]_{\star}^{\mu}.
    \end{equation}
    Taking the partial trace over $V_{\lambda}$ and dividing by $d_{\lambda}$ gives
    \begin{equation}
        [\mathrm{A}_{\lambda, \lambda}]_{\star}^{\mu} = \sqrt{\frac{d_\mu}{n d_\lambda}}.
    \end{equation}

\subsubsection*{Case 2: \texorpdfstring{$t \in \T_{n-1} \sigma_{n-1}$}{t in T n minus 2 sigma n minus one}}
    
    \noindent In this case, we have transversal element $t = t_{i}^{(n)} = t_{i}^{n-1}\sigma_{n-1}$, where $i < n$.
    First, we will consider the case of \emph{detours}, where $\widetilde{\lambda} \neq \lambda$.
    In this case, we see that the multiplicity spaces $M^-_{\lambda, \widetilde{\lambda}} = \C^{\mathcal{N}^-(\lambda) \cap \N^-(\widetilde{\lambda})}$ on the LHS and $M^+_{\lambda, \widetilde{\lambda}} = \C^{\mathcal{N}^+(\lambda) \cap \N^+(\widetilde{\lambda})}$ on the RHS are both at most one-dimensional, by \cref{lemma: diamond}.
    To find matrix entries, we only need to consider the one-dimensional case.
    Then, the only possible entries of $\mathrm{A}_{\lambda, \widetilde{\lambda}}$ making the matrix orthogonal are 
    \begin{equation}\label{eq: detour embedding is plusminus one}
        \lambda \neq \widetilde{\lambda} \Rightarrow [\mathrm{A}_{\lambda, \widetilde{\lambda}}]^\mu_{\kappa} = \pm 1.
    \end{equation}
    In the next part of the proof, we will see that we always have that $\lambda \neq \widetilde{\lambda}$ implies that $[\mathrm{A}_{\lambda, \widetilde{\lambda}}]^{\mu}_\kappa = 1$.

    Now, consider \cref{eq: objective for finding entries of A} again.
    As outer indices, fix transversal element $t = t^{(n)}_{i}$ with $i < n$.
    The Mackey relabeling on the RHS maps $t^{(n)}_{i}$ to $t^{(n-1)}_{i}$.
    Further, let Yamanouchi word $p$ encode a path $P \in \Path(\lambda)$ such that $P' \in \Path(\kappa)$ for some $\kappa \in \N^-(\lambda)$ and let $q$ encode $Q \in \Path(\mu)$ for some $\mu \in \N^{+}(\lambda)$ with $Q' \in \Path(\widetilde{\lambda})$ for some $\widetilde{\lambda} \in \widehat{S}_{n-1}$.
    This is summarized in \cref{fig:induction2_paths}.    
    \begin{figure}[H]
       \centering
       \includegraphics[]{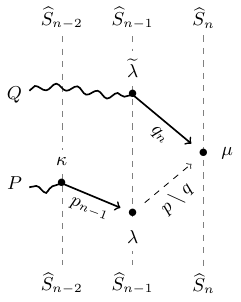}
       \caption{Visualization of the paths appearing in the proof of \cref{thm: second induction}. 
        Here, $\lambda$ and $\widetilde{\lambda}$ are drawn as distinct for clarity, but note that thet may also coincide.}
       \label{fig:induction2_paths}
    \end{figure}
    
    From \cref{eq: objective for finding entries of A}, we find that 
    \begin{align}
         \sqrt{\frac{d_\mu}{n d_\lambda}}[\Rep_\mu (t^{(n)}_i)]^{q}_{p \to \mu} &= \sum_{\widetilde{\lambda}' \in \widehat{S}_{n-1}}\sum_{u : \in \Path(\widetilde{\lambda}')}[\mathrm{A}_{\lambda, \widetilde{\lambda}'}]^{\mu}_{\kappa}\delta^{q'}_{u} \sqrt{\frac{d_{\widetilde{\lambda}'}}{(n{-}1)d_\kappa}} [\Rep_{\widetilde{\lambda}'}(t^{(n-1)}_{i})]^{u}_{p' \to \widetilde{\lambda}'}\\
         &= [\mathrm{A}_{\lambda, \widetilde{\lambda}}]^{\mu}_{\kappa}\sqrt{\frac{d_{\widetilde{\lambda}}}{(n{-}1)d_\kappa}} [\Rep_{\widetilde{\lambda}}(t^{(n-1)}_{i})]^{q'}_{p' \to \widetilde{\lambda}}
    \end{align}
    Here, the entry $p_{n-1}$ makes up the difference between $\kappa$ and $\lambda$.
    As $\lambda$ is fixed, this means that $p_{n-1}$ identifies the input $\kappa$ of $\mathrm{A}_{\lambda, \widetilde{\lambda}}$. 
    If we left-multiply with $\Rep_{\widetilde{\lambda}}((t_i^{(n-1)})^{-1})$, we get
    \begin{align}
         &\sum_{q' \in \Path(\widetilde{\lambda})} [\Rep_{\widetilde{\lambda}}((t^{(n-1)}_i)^{-1})]^{u}_{q'} \sqrt{\frac{d_\mu}{n d_\lambda}}[\Rep_\mu (t^{(n)}_i)]^{q' \to \mu}_{p \to \mu} 
         =\\
         & \qquad \qquad \sum_{q' \in \Path(\widetilde{\lambda})} [\Rep_{\widetilde{\lambda}}((t^{(n-1)}_i)^{-1})]^{u}_{q'} [\mathrm{A}_{\lambda, \widetilde{\lambda}}]^{\mu}_{\kappa}\sqrt{\frac{d_{\widetilde{\lambda}}}{(n{-}1)d_\kappa}} [\Rep_{\widetilde{\lambda}}(t^{(n-1)}_{i})]^{q'}_{p' \to \widetilde{\lambda}}\notag
    \end{align}
    On the RHS, we can use group homomorphism directly. 
    On the LHS, we can replace 
    \begin{equation}
        [\Rep_{\widetilde{\lambda}}((t^{(n-1)}_i)^{-1})]^{u}_{q'} = [\Rep_{\mu}((t^{(n-1)}_i)^{-1})]^{u \to \mu}_{q' \to \mu},
    \end{equation}
    as $\Rep_{\mu}$ is subgroup-adapted and $\mu \in \N^+(\widetilde{\lambda})$.
    Moreover, we can then replace the sum over $q' \to \mu$ with $q' \in \Path(\widetilde{\lambda})$ by a sum over $q \in \Path(\mu)$, because all entries where such $q$ does not go through $\widetilde{\lambda}$ will be zero. 
    Then, we can use group homomorphism on the LHS, too.
    If we also move the scalar factors to the LHS, we find
    \begin{align}
         \sqrt{\frac{(n{-}1)d_\mu d_\kappa}{n d_\lambda d_{\widetilde{\lambda}}}}[\Rep_\mu ((t^{(n-1)}_i)^{-1}t^{(n)}_i)]^{u \to \mu}_{p \to \mu} 
         &= [\mathrm{A}_{\lambda, \widetilde{\lambda}}]^{\mu}_{\kappa} [\Rep_{\widetilde{\lambda}}((t^{(n-1)}_i)^{-1}t^{(n-1)}_{i})]^{u}_{p' \to \widetilde{\lambda}}
    \end{align}
    For the RHS, note that $(t^{(n-1)}_i)^{-1}t^{(n-1)}_{i} = e$.
    For the LHS, note that $t^{(n)}_i = t^{(n-1)}_i \sigma_{n-1}$ by \cref{prop: Mackey transversals for Sn}.
    Hence, we have $(t^{(n-1)}_i)^{-1}t^{(n)}_{i} = (t^{(n-1)}_i)^{-1}t^{(n-1)}_{i}\sigma_{n-1} = \sigma_{n-1}$.
    We find that 
    \begin{align}
         \sqrt{\frac{(n{-}1)d_\mu d_\kappa}{n d_\lambda d_{\widetilde{\lambda}}}}[\Rep_\mu (\sigma_{n-1})]^{u \to \mu}_{p \to \mu} 
         &= [\mathrm{A}_{\lambda, \widetilde{\lambda}}]^{\mu}_{\kappa} \delta^{u'}_{p'}
    \end{align}
    Now, we can insert Young's orthogonal form for $[\Rep_\mu (\sigma_{n-1})]^{u \to \mu}_{p \to \mu}$.
    Using that we can express the Manhattan distance in terms of contents as $r_{P \to \mu}(n{-}1) = \cont(\mu {\setminus} \lambda) - \cont(\lambda{\setminus} \kappa)$, Young's orthogonal form gives
    \begin{equation}
        [\Rep_\mu (\sigma_{n-1})]^{u \to \mu}_{p \to \mu} = \delta^{u'}_{p'}\begin{cases}
            \frac{1}{\cont(\mu {\setminus} \lambda) -\cont(\lambda {\setminus} \kappa)} &\text{if } \lambda = \widetilde{\lambda},\\
            \sqrt{1-\frac{1}{(\cont(\mu {\setminus} \lambda) -\cont(\lambda {\setminus} \kappa ))^2}} &\text{if } \lambda \neq \widetilde{\lambda}.
        \end{cases}
    \end{equation}
    For detours, note that $\sqrt{1-\frac{1}{(\cont(\mu {\setminus} \lambda) -\cont(\lambda {\setminus} \kappa ))^2}} \geq 0$, meaning that \cref{eq: detour embedding is plusminus one} gives entries of $1$ for the detour case. 

    When $\lambda = \widetilde{\lambda}$, we find 
    \begin{align}
         [\mathrm{A}_{\lambda, \widetilde{\lambda}}]^{\mu}_{\kappa} \delta^{u'}_{p'} = \sqrt{\frac{(n{-}1)d_\mu d_\kappa}{n d_\lambda^2}}\delta^{u'}_{p'}
            \frac{1}{\cont(\mu {\setminus} \lambda) -\cont(\lambda {\setminus} \kappa)}.
    \end{align}
    Taking the partial trace over $V_\kappa$, and dividing by $d_\kappa$ gives
    \begin{align}
         [\mathrm{A}_{\lambda, \widetilde{\lambda}}]^{\mu}_{\kappa} = \sqrt{\frac{(n{-}1)d_\mu d_\kappa}{n d_\lambda^2}}
            \frac{1}{\cont(\mu {\setminus} \lambda) -\cont(\lambda {\setminus} \kappa)}.
    \end{align}
   This concludes our proof.
\end{proof}

\subsubsection*{Proof for \texorpdfstring{\cref{thm: second induction circuit}}{second induction circuit}}\label{app: proof second induction circuit}

\begin{proof}      
    We prove that for $1 \leq m \leq n$, the circuit implements $\IndMap[m]$ with the extra flag register as in \cref{prop: second induction induction step circuit}. 
    We do this by induction on $m$. 
    When $m = n$, the flag register is always $\ket{1}$ and can be uncomputed as in the circuit in \cref{fig:2nd induction cycle version}. 

    \subsubsection*{Base case: \texorpdfstring{$m=1$}{m is 1}}
    
    For the base case, we aim to implement a gate that acts as 
    \begin{align}
        i = 1 &\Rightarrow &\IndMap[][1] \ket{1} &= \ket{1}\ket{q_1 = 1}\ket{\lambda=(1)} &\\
        i > 1 &\Rightarrow &\IndMap[][1] \ket{i} &= \ket{i}\ket{q_1 = \star}\ket{\lambda = \varnothing}&
    \end{align}
    The circuit for the base case is as in \cref{fig:2nd induction cycle version base case}.
    \begin{figure}[H]
        \centering
        \includegraphics[]{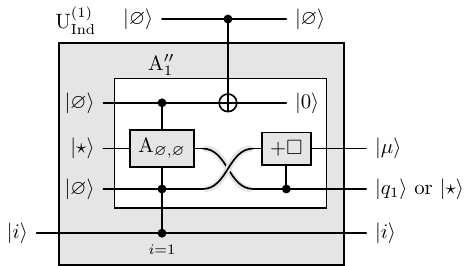}
        \caption{Circuit of \cref{fig:2nd induction cycle version} for the base case of $m=1$.}
        \label{fig:2nd induction cycle version base case}
    \end{figure}
    When $i = 1$, the one-dimensional gate $\mathrm{A}_{\varnothing,\varnothing} : \C^{\{\star\}} \to \C^{\{1\}}$ maps $\star$ to the only possible Yamanouchi word entry $q_1 = 1$ encoding the unique path in $\Path((1))$, with end point $(1) \in \widehat{S}_{m} = \widehat{S}_1$. 
    Hence, $\ket{\mu}$ becomes $\ket{(1)}$ at the output, showing that the circuit correctly implements $\IndMap[][1]$ for $i = 1$.
    For $i > 1$, the gate $\mathrm{A}_{\varnothing,\varnothing}$ is switched off, and we keep the special symbol $\star$ at the output $\ket{q_1}$. 
    As $+\square$ acts trivially, too, we get $\mu = \varnothing$, showing the correctness of the circuit in the base case.

    \subsubsection*{Induction step: \texorpdfstring{$m > 1$}{m gt 1}}

    For the induction step, assume the the circuit for $m {-} 1$ correctly implements $\IndMap[][m-1]$ with the flag register. 
    Then, obtain the circuit for $\IndMap[][m]$ as in \cref{fig:2nd induction cycle version induction step}.
    \begin{figure}[H]
        \centering
        \includegraphics[]{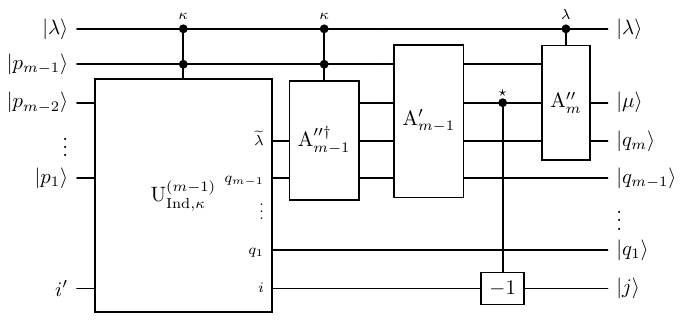}
        \caption{Circuit of \cref{fig:2nd induction cycle version} for $m>1$, where the induction hypothesis implements $\IndMap[][m-1]$ as in \cref{fig:2nd induction cycle version}.
        If we insert this circuit for $\IndMap[][m-1]$, note that canceling $\mathrm{A}''^\dagger_{m-1}$ with the last gate in the implementation of $\IndMap[][m-1]$ gives us the circuit of \cref{fig:2nd induction cycle version}, which should implement $\IndMap[][m]$.}
        \label{fig:2nd induction cycle version induction step}
    \end{figure}

    Let us now focus on the subcircuit after calling $\IndMap[][m-1]$.
    If we start by filling in the definitions from \cref{fig: A matrix wrapper definitions for second induction}, we find that the part after $\IndMap[][m-1]$ becomes as in \cref{fig: proof_full_mackey_circ_ind_step_pt_1}.
    \begin{figure}[H]
        \centering
        \includegraphics[width=\textwidth]{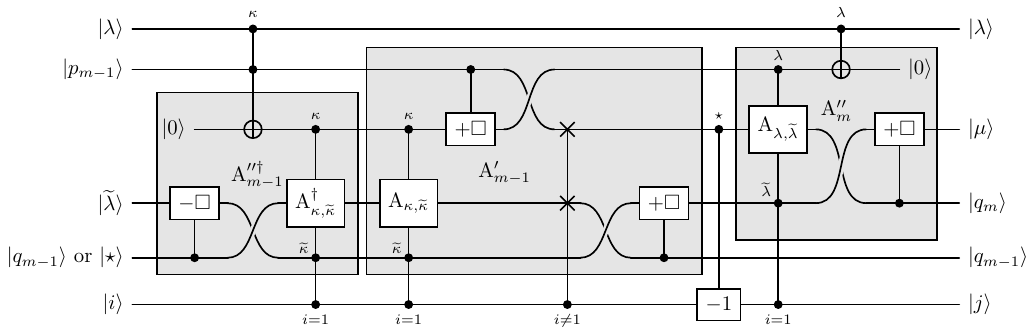}
        \caption{The last four gates of \cref{fig:2nd induction cycle version induction step}, with definitions inserted.}
        \label{fig: proof_full_mackey_circ_ind_step_pt_1}
    \end{figure}

    Then, note that $\mathrm{A}_{\kappa, \widetilde{\kappa}}^\dagger$ and $\mathrm{A}_{\kappa, \widetilde{\kappa}}$ cancel.
    Also, we can ``untangle'' the wires of $\ket{\widetilde{\lambda}}$ and $\ket{q_{m-1}}$, and find a simplified circuit as in \cref{fig: proof_full_mackey_circ_ind_step_pt_2}.
    \begin{figure}[H]
        \centering
        \includegraphics[]{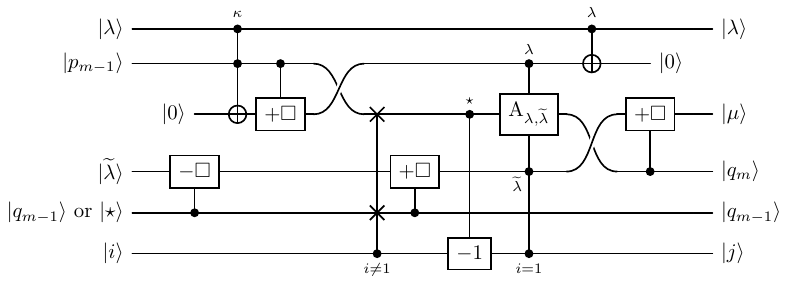}
        \caption{The circuit of \cref{fig: proof_full_mackey_circ_ind_step_pt_1}, after gate cancellations and untangling wires.}
        \label{fig: proof_full_mackey_circ_ind_step_pt_2}
    \end{figure}

    Recall that path $P = P'' \to \kappa \to \lambda$.
    Hence, the ancillary register is always set to $\ket{\kappa}$, and then to $\lambda$ by adding a box at row $p_{m-1}$. 
    As it is also uncomputed at the end of the circuit, we can completely remove it, and instead use the control register $\ket{\lambda}$ as a control for $\mathrm{A}_{\lambda, \widetilde{\lambda}}$.
    We obtain the circuit from \cref{fig: proof_full_mackey_circ_ind_step_pt_3}.

    \begin{figure}[H]
        \centering
        \includegraphics[]{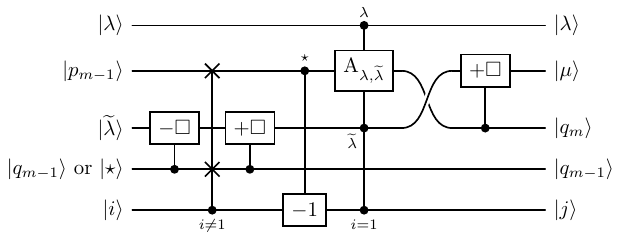}
        \caption{The circuit of \cref{fig: proof_full_mackey_circ_ind_step_pt_2} after removing the ancillary register.}
        \label{fig: proof_full_mackey_circ_ind_step_pt_3}
    \end{figure}

    Then, we commute the gate $+\square$ through the controlled swap as in \cref{fig: proof_full_mackey_circ_ind_step_pt_4}.

    \begin{figure}[H]
        \centering
        \includegraphics[]{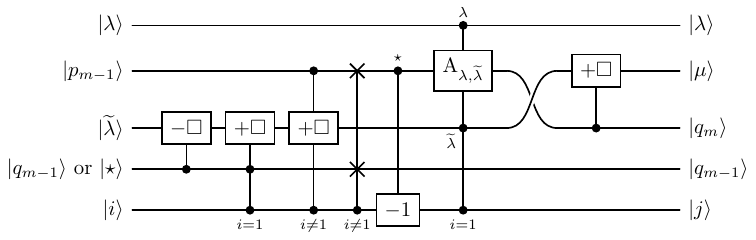}
        \caption{The circuit of \cref{fig: proof_full_mackey_circ_ind_step_pt_3} after commuting the $+\square$-gate through the controlled swap.}
        \label{fig: proof_full_mackey_circ_ind_step_pt_4}
    \end{figure}

    Recall that we did not draw the gate $\IndMap[][m-1]$ in this circuit, but that this gate determines the subspace on which we need to consider the action of our circuit.
    We defined $\IndMap[][m-1]$ such that we can only have $\ket{\star}$ in place of $\ket{q_{m-1}}$ when $\IndMap[m-1]$ sees an ``invalid'' input $\ket{i'}$ with $i' > m{-} 1$.
    In that case, $\ket{i}$ in \cref{fig: proof_full_mackey_circ_ind_step_pt_4} will be $i = i' - m + 2 > 1$.
    Hence, we have $i \neq 1$ if and only if $q_{m-1} = \star$.
    If $q_{m-1} = \star$, then $-\square$ acts trivially, which means that for our subspace, the gates $-\square$ and the first controlled $+\square$ in \cref{fig: proof_full_mackey_circ_ind_step_pt_4} cancel.
    We find the circuit as in \cref{fig: proof_full_mackey_circ_ind_step_pt_5}.

    \begin{figure}[H]
        \centering
        \includegraphics[]{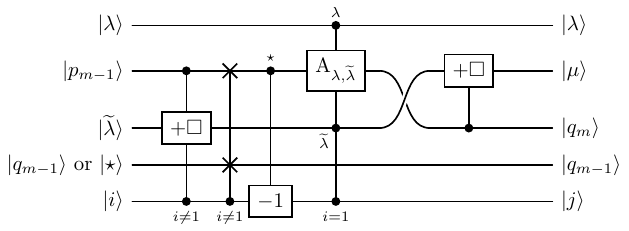}
        \caption{The circuit of \cref{fig: proof_full_mackey_circ_ind_step_pt_4} after using properties of the relevant subspace to cancel $-\square$ and the first $+\square$.}
        \label{fig: proof_full_mackey_circ_ind_step_pt_5}
    \end{figure}

    This is exactly the circuit at the end of $\cref{fig: second induction induction step}$. 
    If we put $\IndMap[][m-1]$ in front, we know by \cref{prop: second induction induction step circuit} that we have correctly implemented $\IndMap[][m]$.
\end{proof}

\subsubsection*{Proof of \texorpdfstring{\cref{prop: induction map with different transversals}}{second induction different transversals}}\label{app: proof induction map with different transversals}

\begin{proof}
    First, note that both $\T'_{n}$ and $\T_n$ contain exactly one representative for each left coset, by definition of a transversal. 
    Hence, for each $i \in [n]$, there is exactly one element $\tau_j \in \T'_n$ such that $t^{(n)}_i S_{n-1} = \tau_j S_{n-1}$, and we define $f(i) = j$.
    Since distinct $t^{(n)}_i$ represent distinct left cosets, $f$ is injective, and hence bijective since $[n]$ is finite.
    Then, we find the unique ``correcting'' subgroup element as $h^{(n)}_i = (t_{i}^{(n)})^{-1}\tau_{f(i)}^{(n)} \in S_{n-1}$ by equality of the cosets.

    Let $\mathrm{U}_{f^{-1}}$ implement $f^{-1}$ on the transversal register, so that for $j \in [n]$, we have
    \begin{equation}
        \mathrm{U}_{f^{-1}}\ket{j} = \ket{f^{-1}(j)}.
    \end{equation}
    Let $\mathrm{U}_{\mathrm{L}}$ apply irreps of $h^{(n)}_{i}$, conditioned on the transversal register, on the path register, so that for $P \in \Path(\lambda)$ with $\lambda \in \widehat{S}_{n-1}$ and $i \in [n]$, we have 
    \begin{equation}
        \mathrm{U}_{\mathrm{L}}\ket{i}\ket{P} = \ket{i}\Rep_{\lambda}\left(h^{(n)}_{i}\right)\ket{P}.
    \end{equation}
    
    For implementing $\IndMapDash[\lambda][n]$, recall from \cref{def: IndMap} that its entries are
    \begin{equation}
        \left[\IndMapDash[\lambda][n]\right]^{Q, \mu}_{j, P} = \sqrt{\frac{d_\mu}{nd_{\lambda}}} \left[\Rep_\mu(\tau^{(n)}_j)\right]^Q_{P \to \mu},
    \end{equation}
    with indices $j \in [n]$ labeling a transversal element $\tau^{(n)}_j \in \T'_n$, $P \in \Path(\lambda)$, $\mu \in \N^+(\lambda)$ and $Q \in \Path(\mu)$.
    From the definitions of $f : [n] \to [n]$ and $h^{(n)} : [n] \to S_{n-1}$, we have 
    \begin{equation}
        \tau^{(n)}_j = t_{f^{-1}(j)}^{(n)}h_{f^{-1}(j)}^{(n)}
    \end{equation}
    Then, we derive that
    \begin{align}
        \left[\IndMapDash[\lambda][n]\right]^{Q, \mu}_{j, P} &= \sqrt{\frac{d_\mu}{nd_{\lambda}}} \left[\Rep_\mu\left(t_{f^{-1}(j)}^{(n)}h_{f^{-1}(j)}^{(n)}\right)\right]^Q_{P \to \mu}&& \text{(\cref{eq: other choice of transversals})}\\
        &= \sum_{U \in \Path(\mu)}\sqrt{\frac{d_\mu}{nd_{\lambda}}} \left[\Rep_\mu\left(t_{f^{-1}(j)}^{(n)}\right)\right]^Q_{U}\left[\Rep_\mu\left(h_{f^{-1}(j)}^{(n)}\right)\right]^U_{P \to \mu} && \text{(group homomorphism)}\\
        &= \sum_{U \in \Path(\lambda)}\sqrt{\frac{d_\mu}{nd_{\lambda}}} \left[\Rep_\mu\left(t_{f^{-1}(j)}^{(n)} \right)\right]^Q_{U \to \mu}\left[\Rep_\lambda\left(h_{f^{-1}(j)}^{(n)}\right)\right]^{U}_{P} && \text{(} \Rep_\mu \text{ subgroup-adapted)}\\
        &= \sum_{U \in \Path(\lambda)}\left[\IndMap[\lambda][n]\right]^{Q, \mu}_{f^{-1}(j), U}\left[\Rep_\lambda\left(h_{f^{-1}(j)}^{(n)}\right)\right]^U_{P} && \text{(def. of } \IndMap[\lambda][n] \text{)}\\
        &= \sum_{U \in \Path(\lambda)}\sum_{i \in [n]}\left[\IndMap[\lambda][n]\right]^{Q, \mu}_{i, U}\left[\Rep_\lambda \left(h_{i}^{(n)}\right)\right]^U_{P}\left[\mathrm{U}_{f^{-1}}\right]^{i}_j && \text{(def. of } \mathrm{U}_{f^{-1}} \text{)}\\
        &= \sum_{U \in \Path(\lambda)}\sum_{i \in [n]}\sum_{k \in [n]}\left[\IndMap[\lambda][n]\right]^{Q, \mu}_{k, U}\left[\mathrm{U}_{\mathrm{L}}\right]^{k,U}_{i, P}\left[\mathrm{U}_{f^{-1}}\right]^{i}_j && \text{(def. of } \mathrm{U}_{\mathrm{L}} \text{)} 
    \end{align}
    This is exactly the decomposition of $\IndMapDash[\lambda][n]$ as in our proposition.
\end{proof}

\subsubsection*{Properties of \texorpdfstring{$\mathrm{A}_\lambda$}{A lambda}}

\addableCellsBoundLemma*
\begin{proof}
Let $k = |\AC(\lambda)|$, and let $\mu$ be an integer partition with $k$ addable cells having the minimal possible number of boxes.
As the number of addable cells is equal to the number of distinct part sizes of a partition plus $1$, we know that $\mu$ has $k-1$ distinct part sizes.
By minimality of the number of boxes, each distinct part size only occurs once in $\mu$.
Moreover, these distinct part sizes must be as small as possible.
Hence, $\mu$ is the ``staircase'' diagram $\mu = (k{-}1,k{-}2,\dots,2,1)$.
This shows that $\mu$ has $1+2+\cdots+(k-1)=\frac{k(k-1)}{2}$ boxes. 

Since $\lambda$ also has $k$ addable cells, and $\mu$ has the minimal possible number of boxes among all partitions with $k$ addable cells, we have $n{-}1 \geq \frac{k(k-1)}{2}$. 
By solving this quadratic inequality for $k$, we find that $k = |\AC(\lambda)| \leq \frac{1+\sqrt{8n-7}}{2}$. 
As the number of addable cells must be an integer, we can take the floor of this upper bound, which concludes our proof.
\end{proof}

\begin{lemma}\label{lemma: determinant of a Cauchy matrix}[Cauchy determinant, \parencite[Eq.~4]{schechter}]
    Let $C$ be an $n \times n$ Cauchy matrix over a field $k$, meaning that the entries of $C$, for $i, j \in [n]$, are given by
    \begin{equation*}
        [C]^{i}_j = \frac{1}{x_i-y_j},
    \end{equation*}
    with $x_i, y_i \in k$ injective sequences and $x_i - y_j \neq 0$ for all $i,j \in [n]$.
    Then, the determinant of $C$ is given by 
    \begin{equation*}
        \det(C) = \frac{\prod_{i=2}^{n}\prod_{j = 1}^{i-1}(x_i - x_j)(y_j - y_i)}{\prod_{i = 1}^{n}\prod_{j=1}^{n}(x_i - y_j)}.
    \end{equation*}
\end{lemma}

\AdetLemma*

\begin{proof}
    Let $e_\lambda = |\AC(\lambda)| = \dim(\mathrm{A}_\lambda)$.
    Index the rows of $\mathrm{A}_\lambda$ by $a \in \AC(\lambda)$ and its columns by $r \in \{\star\} \cup \RC(\lambda)$, as in \cref{def: embedding operation}.
    Define the diagonal matrices $D_\lambda^{(1)}$ and $D_\lambda^{(2)}$, and the $e_\lambda \times (e_\lambda-1)$ matrix $C_\lambda$, by
    \begin{align}
        \left[D_\lambda^{(1)}\right]^a_a &\defeq \sqrt{\frac{d_{\lambda+a}}{n d_\lambda}}, \label{eq: A decomp D1}\\
        \left[D_\lambda^{(2)}\right]^r_r &\defeq
        \begin{cases}
            1 & \text{if } r=\star,\\
            \sqrt{\frac{(n-1)d_{\lambda-r}}{d_\lambda}} & \text{if } r \in \RC(\lambda),
        \end{cases} \label{eq: A decomp D2}\\
        \left[C_\lambda\right]^a_r &\defeq \frac{1}{\cont(a)-\cont(r)} \quad \text{for } r \in \RC(\lambda). \label{eq: A decomp C}
    \end{align}
    Direct substitution into \cref{def: embedding operation} gives
    \begin{equation*}
        \mathrm{A}_\lambda = D_\lambda^{(1)} \left[\bm{1}, C_\lambda \right]D_\lambda^{(2)}.
    \end{equation*}
    Therefore,
    \begin{equation*}
        \det \left(\mathrm{A}_\lambda\right) = \det \left(D_\lambda^{(1)}\right)\det\left(\left[\bm{1}, C_\lambda \right]\right) \det \left(D_\lambda^{(2)}\right).
    \end{equation*}
    As $\mathrm{A}_\lambda$ is unitary, there exists a $\phi \in \mathbb{R}$ such that
    \begin{equation*}
        \det \left(\mathrm{A}_\lambda\right) = e^{i\phi}.
    \end{equation*}
    By \cref{eq: A decomp D1}, the entries of $D_\lambda^{(1)}$ are square roots of ratios of positive dimensions divided by $n>0$, so all its diagonal entries are real and positive.
    Similarly, \cref{eq: A decomp D2} shows that the diagonal entries of $D_\lambda^{(2)}$ are real and positive.
    By \cref{eq: A decomp C}, the entries of $C_\lambda$ are real, but not necessarily positive.
    Thus, $\mathrm{A}_\lambda$ is real, and its determinant can only be $1$ or $-1$.
    Moreover, $[\bm{1},C_\lambda]$ is the only factor in the decomposition of $\mathrm{A}_\lambda$ that can have a negative determinant, so
    \begin{equation*}
        \det\left(\mathrm{A}_\lambda\right) = \sign\left(\det\left(\left[\bm{1},C_\lambda\right]\right)\right).
    \end{equation*}
    Index the addable and removable cells from northeast to southwest along the outer boundary of $\lambda$.
    The addable and removable corners alternate, and their contents strictly decrease along this boundary, so
    \begin{equation*}
        \cont(a_1)>\cont(r_1)>\cont(a_2)>\cont(r_2)>\dots>\cont(r_{e_\lambda-1})>\cont(a_{e_\lambda}).
    \end{equation*}
    Hence, for $i \in [e_\lambda]$ and $j \in [e_\lambda-1]$,
    \begin{equation*}
        \left[C_\lambda\right]^{i}_{j} = \frac{1}{\cont(a_i) - \cont(r_j)}.
    \end{equation*}
    Expanding the determinant along the first column gives
    \begin{equation}\label{eq: cofactor expansion for det A}
        \det\left(\left[\bm{1}, C_\lambda\right]\right) = \sum_{k = 1}^{e_\lambda}(-1)^{k+1}\det\left(C_{\lambda,k}\right),
    \end{equation}
    where $C_{\lambda,k}$ denotes the matrix $C_\lambda$ after removal of the $k$-th row, corresponding to $a_k$.
    
    We will now show that the summands in \cref{eq: cofactor expansion for det A} all have an equal sign, depending on $e_\lambda$. 
    Define a function $f_k : [e_\lambda - 1] \to [e_\lambda]$ as
    \begin{equation*}
        f_k (i) = \begin{cases}
            i &\text{if } i < k,\\
            i + 1 &\text{if } i \geq k.
        \end{cases}
    \end{equation*}
    Note that $f_k$ is non-decreasing, so that $i > j \iff f_k(i) > f_k(j)$.
    We use \cref{lemma: determinant of a Cauchy matrix} to write
    \begin{equation}\label{eq: det Ck}
        \det(C_{\lambda,k}) = \frac{\prod_{i=2}^{e_\lambda - 1}\prod_{j = 1}^{i-1}\left(\cont\left(a_{f_k(i)}\right) - \cont\left(a_{f_k(j)}\right)\right)\left(\cont\left(r_j\right) - \cont\left(r_i\right)\right)}{\prod_{i = 1}^{e_\lambda - 1}\prod_{j=1}^{e_\lambda - 1}\left(\cont\left(a_{f_k(i)}\right) -\cont\left(r_j\right)\right)}.
    \end{equation}
    First, note that the sign of the denominator on the right-hand side of \cref{eq: det Ck} equals the sign of the product of all entries of $C_{\lambda,k}$, which is a subset of the entries of $C_\lambda$.
    By the ordering above, for $i \in [e_\lambda]$ and $j \in [e_\lambda - 1]$,
    \begin{equation*}
        \sign\left(\left[C_{\lambda}\right]^i_j\right) = \sign\left(\frac{1}{\cont(a_i) - \cont(r_j)}\right) = \sign\left(\cont(a_i) - \cont(r_j)\right) = \begin{cases}
            1 &\text{if }  i \leq j,\\
            -1 &\text{if } i > j.
        \end{cases}
    \end{equation*}
    In other words, all entries on and above the main diagonal are positive, and all entries under the main diagonal are negative. 
    We find that the total number of negative entries equals $\frac{e_{\lambda}^2 -e_\lambda}{2}$. 
    If we remove row $k$, we remove a total of $k-1$ negative entries, leaving us with a total of $\frac{e_{\lambda}^2 -e_\lambda}{2} - k + 1$ negative entries.
    Hence, the sign of the denominator is
    \begin{equation}\label{eq: sign of denominator in proof sign det A}
        \sign\left(\prod_{i = 1}^{e_\lambda - 1}\prod_{j=1}^{e_\lambda - 1}\left(\cont\left(a_{f_k(i)}\right) -\cont\left(r_j\right)\right)\right) = (-1)^{\left(\frac{e_\lambda^2 - e_\lambda}{2}\right)}(-1)^{\left(- k + 1\right)}.
    \end{equation}
    First, consider the cases where $e_\lambda$ is even. Then, $e_\lambda = 2a$ for some $a \in \mathbb{Z}_{>0}$, and $\frac{e_\lambda^2 - e_\lambda}{2} = 2a^2 - a$, which is only even if $a$ is even. 
    Second, if $e_\lambda$ is odd, then $e_\lambda = 2a + 1$ for some $a \in \mathbb{Z}_{>0}$ because $e_\lambda > 1$. 
    We may rewrite the number of negative entries in $C_\lambda$ as $2a^2 + a$, which is even if and only if $a$ is even. 
    We find that 
    \begin{equation*}
        (-1)^{\left(\frac{e_\lambda^2 - e_\lambda}{2}\right)} = \begin{cases}
            1 &\text{if } e_\lambda \mod 4 \in \{0, 1\},\\
            -1 &\text{if } e_\lambda \mod 4 \in \{2, 3\}.
        \end{cases}
    \end{equation*}
    For the other term on the RHS of \cref{eq: sign of denominator in proof sign det A}, we can rewrite
    \begin{equation*}
        (-1)^{\left(- k + 1\right)} = (-1)^{k+1}.
    \end{equation*}
    
    For the terms in the numerator of the RHS of \cref{eq: det Ck}, note that $i > j$ and therefore $f_k(i) > f_k(j)$. 
    By the same ordering,
    \begin{align*}
        \sign\left(\cont\left(a_{f_k(i)}\right) - \cont(a_{f_k(j)})\right) &= -1 & \sign\left(\cont\left(r_{j}\right) - \cont(r_{i})\right) &= 1.
    \end{align*}
    Hence, the number of negative factors in the numerator is equal to the number of pairs $(i,j)$ in our product, which is 
    \begin{equation*}
        |\{(i,j) \in \mathbb{Z}^2\,|\, 2 \leq i \leq e_\lambda -1, 1 \leq j \leq i-1 \}| = \frac{(e_\lambda-1)^2 - (e_\lambda-1)}{2}.
    \end{equation*}
    Hence, for the sign of the numerator, we find that 
    \begin{align*}
        &\sign\left(\prod_{i=2}^{e_\lambda - 1}\prod_{j = 1}^{i-1}\left(\cont\left(a_{f_k(i)}\right) - \cont\left(a_{f_k(j)}\right)\right)\left(\cont\left(r_j\right) - \cont\left(r_i\right)\right)\right) \\
        &\qquad\qquad\qquad\qquad\qquad\qquad\qquad\qquad\qquad\qquad= \begin{cases}
            1 &\text{if } e_\lambda \mod 4 \in \{1, 2\},\\
            -1 &\text{if } e_\lambda \mod 4 \in \{3, 0\}.
        \end{cases}
    \end{align*}
    For all summands in \cref{eq: cofactor expansion for det A}, we find
    \begin{align*}
        &\sign\left((-1)^{k+1} \det(C_{\lambda,k})\right) \\
        &\qquad\qquad\qquad\qquad = (-1)^{2(k+1)}\underbrace{\left(
        \begin{cases} 1 &\text{if } e_\lambda \mod 4 \in \{1,2\},\\
        -1 & \text{if } e_\lambda \mod 4 \in \{3,0\};
        \end{cases}\right)}_{\text{from numerator of Cauchy determinant}}\underbrace{\left(
        \begin{cases} 1 &\text{if } e_\lambda \mod 4 \in \{0,1\},\\
        -1 & \text{if } e_\lambda \mod 4 \in \{2,3\};
        \end{cases}\right)}_{\text{from denominator of Cauchy determinant}}\\
        & \qquad\qquad\qquad\qquad = \begin{cases}1 &\text{if } e_\lambda \text{ is odd},\\
        -1 &\text{if } e_\lambda \text{ is even}.\end{cases}
    \end{align*} 
    Going back to \cref{eq: cofactor expansion for det A}, we use that all terms in our summation have an equal sign depending on $e_\lambda$ to rewrite
    \begin{equation*}
        \sign\left(\det\left(\left[\bm{1}, C_\lambda\right]\right)\right) = \begin{cases}1 &\text{if } e_\lambda \text{ is odd},\\
        -1 &\text{if } e_\lambda \text{ is even}.\end{cases}
    \end{equation*}
    Finally, we prove our lemma by using that $\det{A_\lambda} = \sign\left(\det\left(\left[\bm{1}, C_\lambda\right]\right)\right)$.
\end{proof}

\begin{lemma}\label{lemma: AC and RC of lambda - r}
    Let $m \in \mathbb{N}, m > 1$. Let $\lambda = (\lambda_1, \hdots, \lambda_t)\in \Lambda_m$, and let $r = (x_r, y_r) \in \mathrm{RC}(\lambda)$. Note that $r$ is the rightmost box of its row, which means that $x_r = \lambda_{y_{r}+1}-1$. 
    Define the truth variables $x_1$ and $x_2$ in $\{0, 1\}$ as
    \begin{align*}
        x_1 &:= \begin{cases}
            1 & \text{ if }\lambda_{y_r+1} = \lambda_{y_r+2} +1 \lor y_r = t-1,\\
            0 & \text{ otherwise};
        \end{cases}& 
        x_2 &:= \begin{cases}
            1 & \text{ if }\lambda_{y_r} = \lambda_{y_r+1} \lor y_r = 0,\\
            0 & \text{ otherwise}.
        \end{cases}
    \end{align*}
    Then 
    \begin{align}
        \AC(\lambda) \setminus \AC(\lambda - r) &= \begin{cases}
        \{r + (1,0)\} &\text{if }x_1 = 0\text{ and } x_2 = 0,\\
        \emptyset 
        &\text{if }x_1 = 0\text{ and } x_2 = 1,\\ 
        \{r+(0,1), r +(1,0)\} &\text{if }x_1 = 1\text{ and } x_2 = 0,\\
        \{r + (0,1)\} &\text{if }x_1 = 1\text{ and } x_2 = 1;
        \end{cases}\\
        \AC(\lambda - r) \setminus \AC(\lambda) &= \{r\};\\
        \RC(\lambda) \setminus \RC(\lambda - r) &= \{r\};\\
        \RC(\lambda -r) \setminus \RC(\lambda) &= \begin{cases}
        \{r - (1,0)\} &\text{if }x_1 = 0\text{ and } x_2 = 0,\\
        \{r - (1,0), r - (0,1)\} 
        &\text{if }x_1 = 0\text{ and } x_2 = 1,\\ 
        \emptyset &\text{if }x_1 = 1\text{ and } x_2 = 0,\\
        \{ r - (0,1)\} &\text{if }x_1 = 1\text{ and } x_2 = 1.
        \end{cases}
    \end{align}
\end{lemma}
\begin{proof}
If we remove box $r$ from $\lambda$ to obtain $\lambda - r$, only the entry $\lambda_{y_r +1}$ changes. 
From the definition of addable and removable cells in \cref{def: AC and RC implicit}, recall that the existence of addable and removable cells at a certain row depends on the difference in length of the row above it and the row below it.
As removing $r$ only changes the length of row $y_r$ (entry $\lambda_{y_{r}+1}$), removing $r$ will only influence the addable and removable cells at the rows $y_{r}-1$, $y_r$, and $y_r +1$.

First, consider $x_1$. 
\cref{fig:annoying prove with cells x1} illustrates the two possibilities. 
It becomes clear that when $x_1 = 1$, we have $r + (0,1) \in \mathrm{AC}(\lambda)$, but $r + (0,1) \notin \mathrm{AC}(\lambda - r)$. 
When $x_1 = 0$, we have $r - (1,0) \notin \mathrm{RC}(\lambda)$, but $r - (1,0) \in \mathrm{RC}(\lambda)$. 
For the rest, the addable cells and removable cells in all rows below $y_r$ are unchanged.

\begin{figure}[H]
    \centering \includegraphics{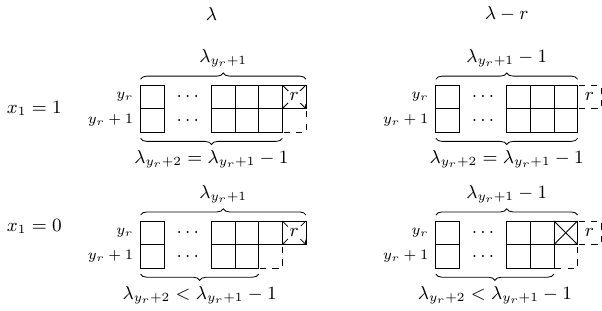}
    \caption{Influence of $x_1$ on $\lambda$ and $\lambda-r$. The brackets indicate the number of boxes in a row. A box with a cross indicates a removable box. A dashed box indicates an addable box.}
    \label{fig:annoying prove with cells x1}
\end{figure}
\noindent 
Second, consider $x_2$. The two possibilities are drawn in \cref{fig:annoying prove with cells x2}. The figure shows that when $x_2 = 1$, we have $r - (0,1) \in \mathrm{RC}(\lambda - r)$, but $r - (0,1) \notin \mathrm{AC}(\lambda)$. When $x_1 = 0$, we have $r - (1,0) \notin \mathrm{RC}(\lambda)$, but $r - (1,0) \in \mathrm{RC}(\lambda)$. For the rest, the addable cells and removable cells in all rows below $y_r$ are unchanged.
\begin{figure}[H]
    \centering
    \includegraphics{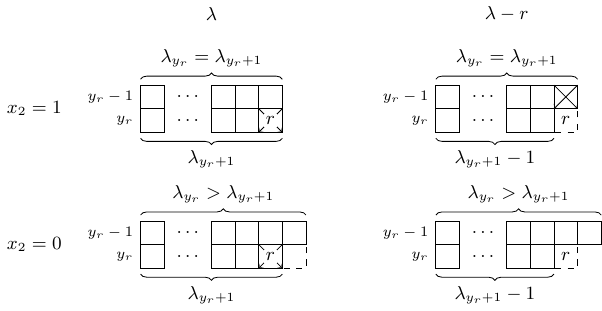}
    \caption{Influence of $x_2$ on $\lambda$ and $\lambda-r$. The brackets indicate the number of boxes in a row. A box with a cross indicates a removable box. A dashed box indicates an addable box.}
    \label{fig:annoying prove with cells x2}
\end{figure}
Lastly, it is clear that $r$ itself is a removable cell of $\lambda$ but not of $\lambda - r$, because $r$ is not a cell of $\lambda$. The other way around, $r$ is an addable cell of $\lambda-r$, but not of $\lambda$.
\end{proof}

\section{On coherent arithmetic}
We will now provide a list of coherent arithmetic operations used in our algorithms, and for each a short discussion of their complexities.
In this, we list complexities for application on an $N$-dimensional register, implemented using $n = \lceil\log_2(n)\rceil$ qubits.
The coherent arithmetic operations in our algorithms are
\begin{itemize}
    \item Addition and subtraction of two integer-valued registers, or of one register and a constant, can be implemented, for example, using the construction of \cite{DraperAdder} with $\mathcal{O}(n^2)$ gates, the same asymptotic depth, and no ancillary qubits.
    The implementations in \cref{fig: theorem first induction full circuit relative paths,fig: circuit for irreps of generators on path register} use this operation.
    \item Finding multiplicative inverses (\cref{fig: circuit for irreps of generators on path register}), as implemented by \cite{qrisp}.
    \item Finding square roots, as implemented by \cite{sqrt}.
\end{itemize}

\printbibliography
\end{document}